\documentclass[11pt,reqno]{amsart} 

\usepackage[backend=bibtex, style=alphabetic, maxalphanames=5, maxnames=99, useprefix=true]{biblatex}
\usepackage{amsmath,amssymb,amsfonts,amsthm,amsbsy,mathrsfs,mathtools}

\usepackage[foot]{amsaddr}
\usepackage{soul}
\usepackage[usenames,dvipsnames]{xcolor}

\usepackage[hidelinks]{hyperref}
\hypersetup{
    colorlinks,
    linkcolor={red!50!black},
    citecolor={blue!50!black},
    urlcolor={blue!80!black}
}

\usepackage[top=1in, bottom=1.25in, left=1.2in, right=1.2in,footskip=.5in]{geometry}
\usepackage{tikz}

\usepackage{dsdshorthand}

\newcommand{\optimal}{{\color{red}${}^\textbf{*}$}}

\renewcommand{\@}{\backslash}

\newcommand{\potimes}{\,\hat\otimes\,}
\newcommand{\D}{{\mathbb{D}}}
\newcommand{\sD}{{\mathscr{D}}}
\newcommand{\sO}{{\mathscr{O}}}
\newcommand{\sC}{{\mathscr{C}}}

\makeatletter
\newcommand{\thickhline}{%
    \noalign {\ifnum 0=`}\fi \hrule height 1pt
    \futurelet \reserved@a \@xhline
}
\makeatother

\begin{document}

\title{Automorphic Spectra and the Conformal Bootstrap} 

\author{Petr Kravchuk\textsuperscript{1,2}}
\address[1]{Institute for Advanced Study, Princeton, NJ 08540, USA}
\address[2]{Simons Center for Geometry and Physics, Stony Brook University, Stony Brook, NY 11794, USA}
\curraddr{}
\email{pkravchuk@ias.edu}

\author{Dalimil Maz\'{a}\v{c}\textsuperscript{1}}
\email{dmazac@ias.edu}

\author{Sridip Pal\textsuperscript{1}}
\email{sridip@ias.edu}

\begin{abstract}
We describe a new method for constraining Laplacian spectra of hyperbolic surfaces and 2-orbifolds. The main ingredient is consistency of the spectral decomposition of integrals of products of four automorphic forms. Using a combination of representation theory of $\mathrm{PSL}_2(\mathbb{R})$ and semi-definite programming, the method yields rigorous upper bounds on the Laplacian spectral gap. In several examples, the bound is nearly sharp. For instance, our bound on all genus-2 surfaces is $\lambda_1\leq 3.8388976481$, while the Bolza surface has $\lambda_1\approx 3.838887258$. The bounds also allow us to determine the set of spectral gaps attained by all hyperbolic 2-orbifolds. Our methods can be generalized to higher-dimensional hyperbolic manifolds and to yield stronger bounds in the two-dimensional case. The ideas were closely inspired by modern conformal bootstrap.
\end{abstract}

\date{}

\makeatletter
\gdef\@fpheader{}
\makeatother

\maketitle

\newpage
\setcounter{tocdepth}{1}
\tableofcontents

\section{Introduction}\label{sec:introduction}
The goal of this article is to propose and study a connection between the spectral geometry of hyperbolic manifolds and the conformal bootstrap. These are fields of mathematics and mathematical physics that have been hitherto evolving separately, each having its own methods and goals. Nevertheless, both subjects are deeply rooted in the harmonic analysis on the non-compact Lie group $\mathrm{SO}(1,d)$, which will allow us to build a robust bridge between them. In particular, we will explain how methods of the conformal bootstrap can be adapted to produce new rigorous bounds on the spectra of hyperbolic manifolds.

The study of Laplacian spectra on hyperbolic manifolds has a long history and enjoys a rich interplay with various areas of mathematics~\cite{sarnak03}. Indeed, hyperbolic manifolds are multi-faceted objects which can be simultaneously viewed through the lens of differential and algebraic geometry, harmonic analysis, and even number theory. In this paper, we will focus on the spectral gap $\lambda_1(M)$, defined as the first positive eigenvalue of the Laplace-Beltrami operator on a hyperbolic manifold $M$. Although our methods generalize to $M$ of any dimensionality, in the present work we will restrict to two-dimensional manifolds and orbifolds.

Key open questions in this subject revolve around establishing bounds on the spectral gap. For example, the Selberg's 1/4 conjecture~\cite{selberg65}, which can be viewed as a special case of the functoriality conjecture of the Langlands program~\cite{langlands70}, asserts that $\lambda_1\geq 1/4$ for all hyperbolic surfaces arising from congruence subgroups of $\mathrm{SL}_2(\mathbb{Z})$. In the present work, we will address the question: What is the set of spectral gaps attained by all compact hyperbolic surfaces and orbifolds? Our main tool will be a new method for proving upper bounds on $\lambda_1$ for all orbifolds of a given topological type. This method originates from mathematical physics, specifically from the field of the conformal bootstrap, which we now briefly describe.

The conformal bootstrap is an approach to studying conformal field theories (CFTs) in general dimension. CFTs can be most concretely realized as scaling limits of critical lattice models of statistical mechanics. On general grounds, such models are expected to develop a symmetry under conformal transformations. Proving conformal invariance of critical lattice models is an important and notoriously hard problem which has only been solved in certain two-dimensional examples~\cite{schramm00,smirnov01}. The approach of the conformal bootstrap is orthogonal. It attempts to abstract the notion of conformal field theory by \emph{defining} it as a system satisfying a list of precise axioms. These axioms capture properties generally expected to hold in the continuum limit, such as conformal invariance. In this way, the conformal bootstrap directly studies the continuum limit without reference to any lattice realization. This holds a particular merit in light of universality -- the idea that different lattice models can have identical continuum limits. Indeed, there is abundant evidence that the bootstrap definition of a CFT describes beautiful, rigid, and still largely mysterious mathematical objects whose identity is quite independent of statistical lattice models.

The connection between the spectral geometry of $d$-dimensional hyperbolic manifolds and the conformal bootstrap in $d-1$ dimensions is controlled to a large extent by the group $G=\mathrm{SO}_0(1,d)$. On one hand, $G$ is the isometry group of the $d$-dimensional hyperbolic space. On the other hand, it is the group of conformal isomorphism of $S^{d-1}$, thus underlying the conformal bootstrap of $(d-1)$-dimensional CFTs.

The conformal bootstrap (see~\cite{Poland:2018epd,rychkov2023new} for modern reviews), in its most common incarnation, is the study of correlation functions of the form
\be
	\<O_1(x_1)\cdots O_n(x_n)\>,
\ee
where $x_i$ belong to the conformally compactified Euclidean space $\R^{d-1}\cup\{\infty\} = S^{d-1}$, and $O_i$ are the local operators. The key ingredient which makes the subject non-trivial is the existence of an operator product expansion, which has the following schematic form
\be
	O_1(x_1)O_2(x_2)=\sum_i f_{12i}|x_1-x_2|^{\De_i-\De_1-\De_2}O_i(x_2),
\ee
where $\De_i\in \R$ are the scaling dimensions of the local operators $O_i$. The structure of the correlation functions and of the operator product expansion is constrained by the conformal symmetry $\SO_0(1,d)$.

\begin{figure}[t]
	\centering
	\begin{tikzpicture}
		\node [] at (0,0) {\includegraphics[scale=0.2]{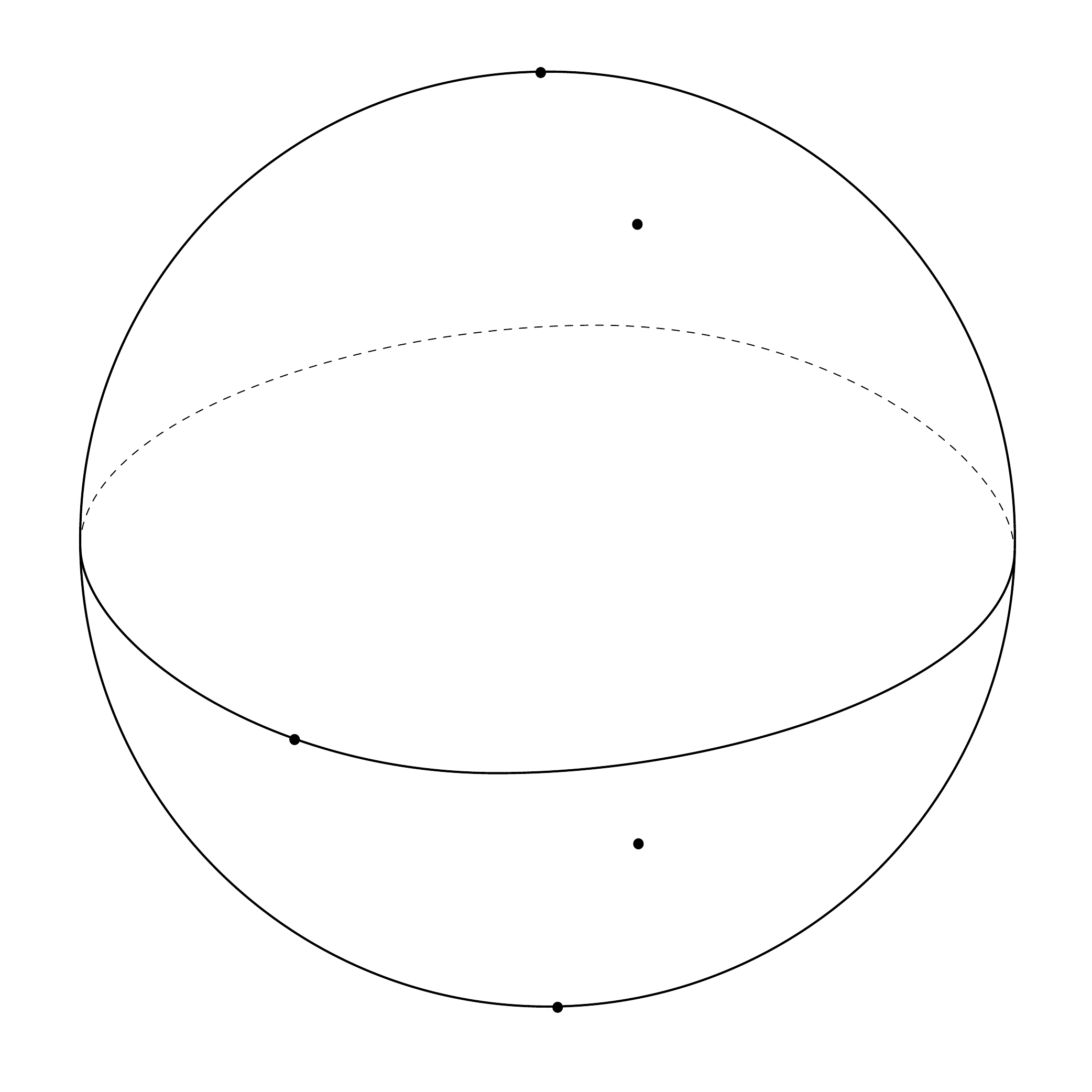}};
		\node [above] at (1.2,1.4) {$\widetilde\sO_n(z_2)$};
		\node [above] at (1.3,-2.4) {$\sO_m(z_1)$};
		\node [above] at (-1.0,-1.3) {$\mathbb{O}_k(x)$};
		\node [above] at (0.1,-3.65) {$0$};
		\node [above] at (0,+3.3) {$\oo$};
	\end{tikzpicture}
	\caption{The operators that we construct for a hyperbolic orbifold are labeled by points on the Riemann sphere. Depending on the type of the operator, it is labeled by a point in the upper or the lower hemisphere, or by a point on the equator.}
	\label{fig:hyperbolicoperators}
\end{figure}
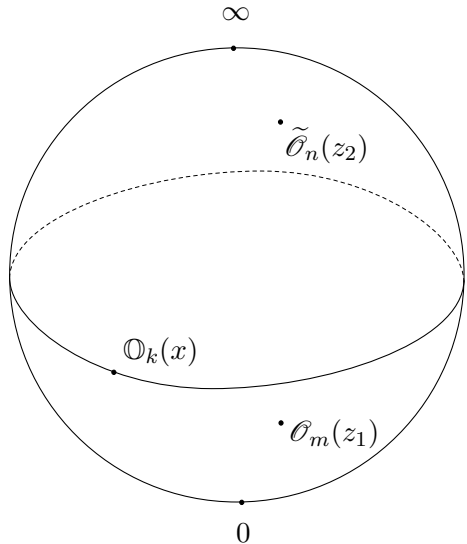

In this work we will consider the case $d=2$ when $G=\SO_0(1,2)=\PSL(2,\R)$.\footnote{Much of the discussion can be straightforwardly generalized to $d>2$} We will show that, given a hyperbolic surface $X$, one can define local operators $\sO_n(z)$ and $\widetilde\sO_n(z)$ where $z$ lives, respectively, in the lower or the upper hemisphere of the Riemann sphere (see Figure~\ref{fig:hyperbolicoperators}). These local operators have well-defined correlation functions, for example
\be\label{eq:introcorr}
	\<\sO_n(z_1)\sO_n(z_2)\widetilde\sO_n(z_3)\widetilde\sO_n(z_4)\>,
\ee
as well as an operator product expansion of the schematic form
\be
	\sO_m(z_1)\sO_n(z_2)=\sum_p f_{mnp}(z_1-z_2)^{p-m-n} \sO_p(z_2)+\cdots,
\ee
and also one for $\sO_m(z_1)\widetilde\sO_n(z_2)$ involving another kind of local operators $\mathbb{O}_k(x)$ with $x$ living on the equator of the Riemann sphere.\footnote{The local operators $\mathbb{O}_k(x)$ also admit an operator product expansion, but we will not study it in this paper.} The spectrum of the scaling dimensions $\De_k$ of $\mathbb{O}_k(x)$ will turn out to be related to the spectrum of the Laplace-Beltrami operator (the Laplacian from here on) on $X$ by $\l_k=\De_k(1-\De_k)$. While some features of this construction are unusual from the point of view of the conformal bootstrap, we will see in Section~\ref{sec:motivation} that the dictionary between the two subjects is remarkably precise, and the unusual features are also present on the conformal field theory side.

This correspondence allows us to study both problems from new points of view. On the one hand, we will use the familiar numerical conformal bootstrap techniques to derive new rigorous bounds on the spectra of hyperbolic surfaces and orbifolds. On the other hand, we hope to use the model given by hyperbolic orbifolds to gain new insights into the conformal bootstrap. For example, even though we do not pursue it in this paper, our discussion gives a practical way of studying the numerical de Sitter bootstrap suggested in~\cite{Hogervorst:2021uvp,DiPietro:2021sjt}, as will be explored elsewhere.

To set up the problem more concretely, let $\mathbb{H} = \{z=x+iy:\,x,y\in\mathbb{R};y>0\}$ be the upper half-plane equipped with the Riemannian metric of constant sectional curvature $-1$, i.e.\
\be
ds^2 = \frac{dx^2+dy^2}{y^2}\,.
\ee
The group of orientation-preserving isometries of $\mathbb{H}$ is $G=\PSL(2,\mathbb{R})$, acting by fractional linear transformations
\be
\begin{pmatrix}
a & b\\ c&d
\end{pmatrix}\cdot z = \frac{a\,z + b}{c\,z +d}\,.
\ee
Let $\Gamma$ be a discrete subgroup of $G$ such that the quotient space $X = \Gamma\backslash \mathbb{H}$ is compact. In other words, $\Gamma$ is a lattice in $G$ with no parabolic elements. If $\Gamma$ has no elliptic elements, $X$ is a smooth hyperbolic surface. If $\Gamma$ is allowed to have elliptic elements, $X$ is a hyperbolic 2-orbifold. Every closed, connected, orientable hyperbolic surface or 2-orbifold arises in this way from some $\Gamma\subset G$ satisfying the above conditions. In what follows, we will refer to the two cases simply as \emph{hyperbolic surface} and \emph{hyperbolic orbifold}, i.e.\ we will assume that these are closed, connected and orientable. Note that any hyperbolic surface is a hyperbolic orbifold.

The Laplacian on $\mathbb{H}$ takes the form $\Delta_{\mathbb{H}} = y^2(\partial^2_x+\partial^2_y)$. When restricted to $\G$-invariant functions, it agrees with the Laplacian $\De_X$ on $X$. We take as the domain of $\De_X$ those functions on $X$ which lift to smooth functions on $\mathbb{H}$. Thus defined, $\De_X$ is essentially self-adjoint~\cite{bucicovschi1999seeleys} and we are interested in its spectrum. Since $X$ is compact, the spectrum of $-\De_X$ is a discrete subset of $\mathbb{R}_{\geq 0}$ whose elements we denote by
\be
\lambda_0=0<\lambda_1 < \lambda_2 < \cdots \rightarrow\infty,
\ee
with multiplicities $d_i>0$. The corresponding eigenfunctions can be viewed as smooth functions on $\mathbb{H}^2$ satisfying
\be
-\Delta_{\mathbb{H}} h_i = \lambda_i\,h_i\,,
\label{eq:SpectralProblem}
\ee
and  $h_i(\gamma\cdot z)=h_i(z)$ for all $\gamma\in\Gamma$. The eigenfunctions $h_i$ are examples of automorphic forms, known as Maass forms.

The above spectral problem has received a lot of attention in both mathematics and physics literature. Despite its apparent simplicity, there is no analytic formula for the eigenvalues and eigenfunctions of any given surface. The lack of exact solvability is partly a reflection of chaoticity of the geodesic motion on $X$~\cite{Hadamard}. Indeed, \eqref{eq:SpectralProblem} is the Schr\"odinger equation coming from quantizing this geodesic motion.

The focus of the present work will be on the low-energy spectrum $\lambda_1$, $\lambda_2$, \ldots.
Specifically, we will introduce a new method, improving and simplifying previous works \cite{Bonifacio:2020xoc,Bonifacio:2021msa}, which will allow us to prove rigorous upper bounds on $\lambda_1$ valid for all $X$ of a fixed topology. For example, we will prove

\begin{theorem}\label{thm:BestBounds}
\phantom{a}
\begin{enumerate}
\item Any hyperbolic orbifold must satisfy $\lambda_1 < 44.8883537$.
\item Any hyperbolic orbifold of genus 2 must satisfy $\lambda_1 < 3.8388976481$.
\item Any hyperbolic orbifold of genus 3 must satisfy $\lambda_1 < 2.6784823893$.
\end{enumerate}
\end{theorem}

The quoted bounds are rigorously valid with the decimal expansion interpreted as an exact rational number. The precise values were obtained using computer-assisted optimization described in the body of the paper. We checked the validity of the proof using exact rational arithmetic. Note that our method also yields simpler bounds which can be derived by hand without resorting to computers. For example, a short pen-and-paper calculation described in Section \ref{ssec:bootstrap} shows that any hyperbolic orbifold must satisfy $\lambda_1 \leq \frac{\sqrt{1297}+55}{2}\approx 45.50694 $, a slighlty weaker version of bound (1) in Theorem~\ref{thm:BestBounds}.

The values appearing in Theorem~\ref{thm:BestBounds} turn out to be very close to $\lambda_1$ of certain well-known hyperbolic surfaces and orbifolds. As shown by Siegel \cite{Siegel}, the hyperbolic orbifold of the smallest area is the unique orbifold of genus 0 and three orbifold points of orders $2,\,3,\,7$. By numerically solving the Laplace equation on this orbifold, we find it has $\lambda_1\approx 44.88835$, indistinguishable from the upper bound (1) of Theorem \ref{thm:BestBounds}. The hyperbolic surface of genus 2 with the largest isometry group is the Bolza surface. Its first positive eigenvalue was computed numerically in \cite{strohmaier2013algorithm}, with the result $\lambda_1 \approx 3.838887258$. Similarly, the most symmetric hyperbolic surface of genus 3 is the Klein quartic, which has $\lambda_1 \approx 2.6779$, see \cite{cook2021properties}.\footnote{See Section~\ref{sec:ff}.} We see that in all three cases, our bound agrees with the known eigenvalues to several decimal places, which makes us hopeful that our method or its generalizations may lead to sharp bounds. In the main text, we will prove several more types of bounds on classes of hyperbolic orbifolds, finding that they are nearly saturated by known orbifolds in almost all cases.

Our bounds improve upon the previous best bounds. For hyperbolic surfaces of genus 2, the previous record is due to Yang and Yau \cite{yang1980eigenvalues}, who showed $\lambda_1 \leq 4$. Note however, that, unlike our bound, the Yang-Yau bound applies to all genus-2 surfaces, and not only those with a hyperbolic metric. For hyperbolic surfaces of genus 3, the previous best bound was found recently by Ros \cite{ros2020first}, who showed $\lambda_1 \leq 2(4-\sqrt{7})\approx 2.7085$. 

We will also ask the following question: what is the image $E\subset \mathbb{R}_{>0}$ of the map $X\mapsto \lambda_1(X)$ when $X$ ranges over all hyperbolic orbifolds? By numerically computing $\l_1$ for a wide class of orbifolds, one may put forward the following Conjecture~\ref{conjecture:theoneandtheonlys}: the structure of $E$ is as shown in Figure~\ref{fig:conjecture}. There, the labels $[g;k_1,\cdots,k_r]$ describe concrete orbifolds and are defined in Section~\ref{sec:bounds}, where we also give a precise statement of the Conjecture~\ref{conjecture:theoneandtheonlys}. Our methods allow us to prove a large part of this conjecture (see Theorem~\ref{thm:setE}).

\begin{figure}[t]
\centering
\includegraphics[width=.75\textwidth]{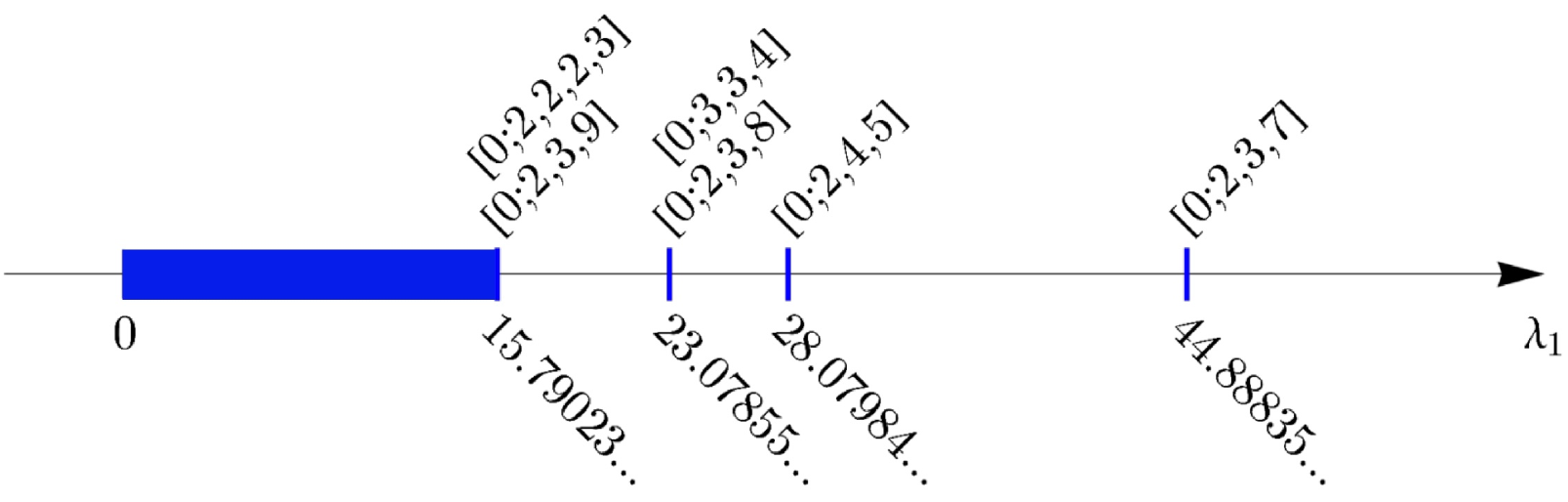}
\caption{\label{fig:conjecture} The conjectured structure of the set of the eigenvalues $\l_1$ attained in hyperbolic orbifolds. There are several discrete points at large $\l_1$ and a continuum below.}
\end{figure}

\subsection{Sketch of the method}
Let us start with a brief sketch of the main idea in mathematical terms. The spectrum of the Laplace operator on $X=\Gamma\backslash\mathbb{H}$ is directly related to the spectrum of unitary irreducible representation of $G=\PSL(2,\mathbb{R})$ appearing in the Hilbert space $L^2(\Gamma\backslash G)$, with $G$ acting by right multiplication. Smooth functions in this Hilbert space form a commutative and associative algebra under pointwise  multiplication. The algebra product is compatible with the action of $G$, and is therefore fully determined in terms of integrals of triple products of modular forms and Maass forms on $\Gamma\backslash \mathbb{H}$. Imposing associativity of the product then severely restricts the allowed spectra of irreducible representations and the integrals of triple products. In practice, we impose associativity as consistency of integrals of quadruple products of $K$-finite vectors in discrete series representations appearing in $L^2(\Gamma\backslash G)$. The resulting constraints can be analyzed using linear and semidefinite programming, and computer-assisted optimization then leads to bounds such as in Theorem \ref{thm:BestBounds}.

\subsection{Relation to previous work}
Let us comment on the relation of our method to the previous works \cite{Bonifacio:2020xoc,Bonifacio:2021msa}. Their authors studied consistency conditions satisfied by quadruple product integrals of Laplacian eigenfunctions on Einstein manifolds \cite{Bonifacio:2020xoc} and hyperbolic manifolds \cite{Bonifacio:2021msa}. They also used these consistency conditions to derive bounds on the Laplacian spectrum.

The main novelty of our approach is that we point out the role played by $G=\PSL(2,\mathbb{R})$ and derive spectral constraints on hyperbolic orbifolds in the framework of representation theory of $G$. This clarifies the relation of this method to the conformal bootstrap and also streamlines practical computations. Another new aspect here is that we consider quadruple integrals of holomorphic forms rather than Laplacian eigenfunctions, which allows us to prove universal bounds at fixed topology.

\begin{table}[h]
\centering
  \begin{tabular}{l @{\hskip 0.3in} l}
\thickhline\vspace{-5pt}\\
the conformal group & $G = \PSL(2,\mathbb{R})$\vspace{5pt}\\
a conformal field theory & a hyperbolic orbifold $\Gamma\backslash \mathbb{H}$ \vspace{5pt}\\
space of field configurations $\Phi$ & $\Gamma\backslash G$\vspace{5pt}\\
a (local) operator & automorphic function $F\in L^2(\Gamma\backslash G)$ \vspace{5pt}\\
Casimir eigenvalue & Laplacian eigenvalue $\lambda = \Delta(1-\Delta)$\vspace{5pt}\\
operator product expansion & pointwise product $F_iF_j = \sum\limits_{k}c_{ij}^{k}F_k$\vspace{5pt}\\
correlation function & overlap integral $\int\limits_{\Gamma\backslash G}\!\!d\mu(g)\,F_1(g)\ldots F_n(g)$\vspace{5pt}\\
\thickhline
\vspace{0pt}
  \end{tabular}
  \caption{A dictionary between conformal field theories and hyperbolic orbifolds}
  \label{tab:dictionary}
\end{table}

\subsection{Physics motivation}
\label{sec:motivation}

Our method was inspired by ideas familiar in the context of the conformal bootstrap. The conformal bootstrap is a set of techniques for rigorously constraining conformal field theories (CFTs) in a general number of spacetime dimensions starting from elementary axioms. These axioms are rooted in the physical principles of conformal invariance and unitarity. The basic idea of the conformal bootstrap goes back to the work of Ferrara, Gatto and Grillo \cite{Ferrara:1973yt} and of Polyakov \cite{Polyakov:1974gs}. It received a renewed impetus when the authors of \cite{Rattazzi:2008pe} explained how it can be implemented using linear programming to derive non-perturbative bounds on the spectral data of conformal field theories. Since then, the conformal bootstrap has given rise to a diverse set of analytical and numerical tools for constraining CFTs, see \cite{Poland:2018epd} for a recent review and \cite{Simmons-Duffin:2016gjk} for a pedagogical exposition. Perhaps the most striking realization has been that the bootstrap bounds on the spectral data are often nearly saturated by interesting CFTs, such as the 3D Ising model \cite{ElShowk:2012ht,Kos:2016ysd}.

While inspired by these ideas, we take a slightly different point of view on the conformal field theory. In particular, we interpret (non-rigorously) the correlation functions as being computed by a Euclidean path integral in $d-1$ dimensions
\be
	\<O_1(x_1)\cdots O_n(x_n)\>=\int_\Phi d\f e^{-S[\f]}O_1(x_1)\cdots O_n(x_n)=\int_{\Phi}d\mu(\f) O_1(x_1)\cdots O_n(x_n),
\ee
where $\Phi$ is the space of all field configurations and $S[\phi]$ is the action (in the sense of Lagrangian mechanics) for the field configuration $\phi$. The second equality interprets $\cD\f e^{-S[\f]}$ as some probability measure $d\mu$ on $\Phi$.  In this picture, the operators $O_i(x_i)$ (at fixed $x_i$) become functions $O_i(x_i):\Phi\to \C$. We will denote the value of $O_i(x_i)$ on a field configuration $\f$ as $F_{i,x_i}(\f)$. With this notation, the correlation function becomes
\be
	\<O_1(x_1)\cdots O_n(x_n)\>=\int_{\Phi}d\mu(\f) F_{1,x_1}(\f)\cdots F_{n,x_n}(\f).
\ee
The role of conformal symmetry in this picture is that we expect the conformal group $G=\SO_0(1,d)$ to act on $\Phi$ and thus also on the functions on $\Phi$. Conformal invariance of the quantum theory means that the measure $d\mu$ is invariant under this action. If $O_i$ are conformal primary operators, then the functions $F_{i,x_i}(\f)$ transform nicely under the action of $\SO_0(1,d)$. General local and non-local operators become more general functions on $\Phi$.

It will become clear in the following section that the spectral theory of hyperbolic manifolds is obtained immediately from this construction by setting $\Phi=\G\@ G$. Indeed, $\G\@ G$ has a natural action by $G$ from the right. The functions on $\Phi=\G\@ G$ are just the functions on $G$ invariant under left multiplication by $\G$, and are known as automorphic functions. The hyperbolic manifold $X$ appears as a quotient of $\Phi$, $X=\Phi/\SO(d)$. In Table~\ref{tab:dictionary}, we show the mapping between the objects from the world of hyperbolic $d$-orbifolds and from the world of conformal field theories in $d-1$ dimensions, specialized to $d=2$ which is the case throughout the paper.

Importantly, this correspondence is obtained by simply setting $\Phi=\G\@ G$, and without any further assumptions or modifications. In particular, the two situations are indistinguishable from the point of view of $G$. This means that the bootstrap equations that we derive are valid (with a caveat that we will come to momentarily) both for conformal field theories \textit{and} hyperbolic orbifolds. Thus, hyperbolic orbifolds provide explicit solutions to the conformal bootstrap equations of the kind that follow from this path-integral picture.

The caveat is that the conformal bootstrap equations that follow from the path-integral picture are not the usual conformal bootstrap equations. In particular, these equations stress Euclidean unitarity instead of Lorentzian unitarity. Here, by Euclidean unitarity we do not mean reflection positivity, but instead the positivity of the measure $d\mu$, i.e.\ the positivity of the Boltzman weights in the statistical interpretation of Euclidean conformal field theory. The unitary representations that are relevant for this setup are those of $G=\SO_0(1,d)$ and not of $\widetilde{\SO}_0(2,d-1)$, the universal cover of $\SO_0(2,d-1)$, as is usual in the conformal field theory.

The conformal bootstrap equations that we discuss here are precisely those proposed recently in the context of de Sitter bootstrap~\cite{Hogervorst:2021uvp,DiPietro:2021sjt} (and their generalizations). However, in our perspective they arise from the standard path integral of a conformal field theory. In this paper we will only study them in the case of $d=2$ and with applications to $\Phi=\G\@ G$. They should also apply to the usual conformal field theories in general $d$. Exploration of this is an ongoing work that will be reported elsewhere.

\subsection*{Note added} When this work was complete, we learned of the upcoming article \cite{Bonifacio:2021aqf}, which contains partially overlapping results. We agreed to coordinate our submissions to the arXiv.

\subsection*{Acknowledgements}
We are indebted to Peter Sarnak for detailed comments on the first version of the article and illuminating discussions. We would also like to thank Leonardo Rastelli for useful discussions. DM gratefully acknowledges funding provided by Edward and Kiyomi Baird as well as the grant DE-SC0009988 from the U.S. Department of Energy. SP acknowledges a debt of gratitude for the funding provided by Tomislav and Vesna Kundic as well as the support from the grant DE-SC0009988 from the U.S. Department of Energy. PK acknowledges the support by DOE grant DE-SC0009988 and the Adler Family Fund at the Institute for Advanced Study. Some computations in this work have been run on the Helios cluster at the Institute for Advanced Study.

\section{How to bootstrap hyperbolic orbifolds}\label{sec:bootstrapMethod}
This section is meant to serve as a pedagogical overview of our method for deriving bounds on the spectral data of hyperbolic orbifolds. We have tried to make the exposition accessible to both mathematicians and physicists. Several technical discussions and proofs are deferred to Section~\ref{sec:proofs}.

\subsection{The space $L^2(\Gamma\backslash G)$}\label{ssec:L2Space}
We will begin by reviewing some background material about $G=\PSL(2,\mathbb{R})$. Recall that we can think of the upper half-plane as $\mathbb{H}=G/K$, where $K = \SO(2,\mathbb{R})/\{\pm I\}$ is a maximal compact subgroup of $G$, which is the stabilizer of $z = i\in \mathbb{H}$. Let us parametrize elements of $K$ as
\be
r_{\theta} = \pm\begin{pmatrix}
\cos\tfrac{\theta}{2} & -\sin\tfrac{\theta}{2}\vspace{0.02in}\\
\sin\tfrac{\theta}{2} &\cos\tfrac{\theta}{2}
\end{pmatrix}\in K\,,
\ee
where $\theta\in\mathbb{R}/2\pi\mathbb{Z}$ and $\pm$ is included because we work with $\PSL(2,\mathbb{R})$ rather than $\SL(2,\mathbb{R})$. This parametrization can be extended to all of $G$ using the Iwasawa decomposition $G = NAK$, which allows us to write every element $g\in G$ in a unique way as
\be
g(x,y,\theta) =
\pm\begin{pmatrix}
1 & x\\
0 &1
\end{pmatrix}
\begin{pmatrix}
\sqrt{y} & 0\\
0 &\frac{1}{\sqrt{y}}
\end{pmatrix}
\begin{pmatrix}
\cos\tfrac{\theta}{2} & -\sin\tfrac{\theta}{2}\vspace{0.02in}\\
\sin\tfrac{\theta}{2} &\cos\tfrac{\theta}{2}
\end{pmatrix}\,,
\label{eq:NAK}
\ee
where $x\in\mathbb{R}$, $y\in\mathbb{R}_{>0}$ are the usual coordinates in $\mathbb{H}$. In other words, we have $G = \mathbb{H}\times S^1$ as a smooth manifold. Note that $g(x,y,\theta)\cdot i = x+i y$. As stated in the Introduction, we will think of a general hyperbolic orbifold as the double quotient $X = \Gamma\backslash G/K$, where $\Gamma$ is a discrete cocompact subgroup of $G$.

Recall that $G$ admits a measure $\mu$ invariant under left and right multiplication, which we will normalize so that $\mu(\Gamma\backslash G) = 1$. This is possible because $\Gamma\backslash\mathbb{H}$ is compact, and hence of finite volume. In the above coordinates, we have
\be
d\mu(g) = \frac{1}{2\pi\vol(\Gamma\backslash \mathbb{H})}\frac{dx\,dy\,d\theta}{y^2}\,.
\ee

In order to study the spectral problem for the Laplacian on $X$, it is convenient to consider the quotient space $\Gamma\backslash G$, depicted in Figure~\ref{fig:principalBundle}. This space is a smooth 3-manifold which is a principal $K$-bundle with base $X$.  $\Gamma\backslash G$ identifies canonically with the unit tangent bundle of the underlying hyperbolic surface. Its advantage is that unlike $X$, $\Gamma\backslash G$ admits an action of $G$ by right multiplication. This allows one to express the spectral problem for the Laplacian on $X$ using representation theory of $G$.

\begin{figure}[h]
\centering
\includegraphics[width=.45\textwidth]{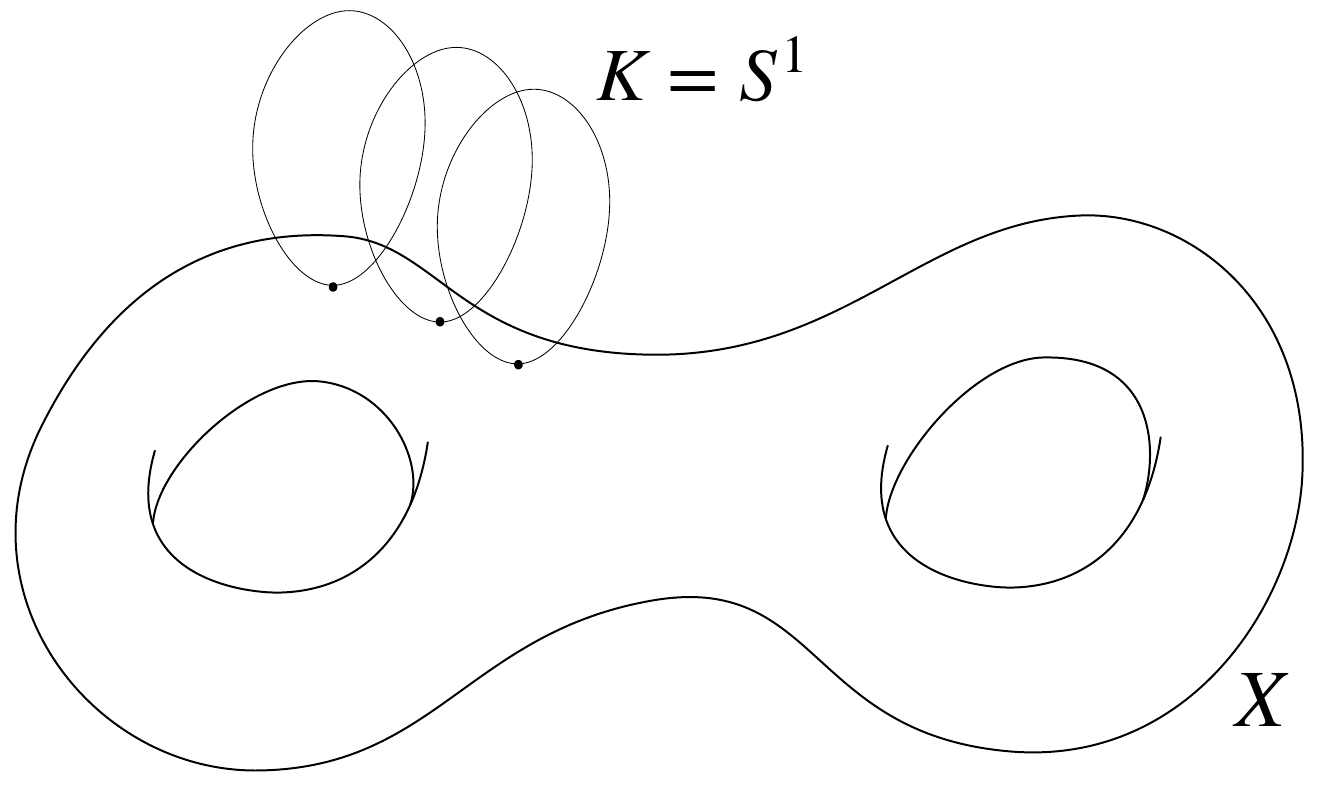}
\caption{A schematic representation of the space $\Gamma\backslash G$ as a principal $K$-bundle over base $X$.  Here $K$ is the circle group and $X=\Gamma\backslash G/K$ is a hyperbolic orbifold.}
\label{fig:principalBundle}
\end{figure}

Let us consider the Hilbert space $L^2(\Gamma\backslash G)$. It is the space of functions $F:G\rightarrow \mathbb{C}$, satisfying $F(\gamma g) = F(g)$ for all $\gamma\in\Gamma$ and all $g\in G$, with finite norm
\be
\|F\|^2 = \!\int\limits_{\Gamma \backslash G}\!\!d\mu(g) |F(g)|^2\,.
\ee
The norm gives rise to the following inner product\footnote{Note that we define $(\cdot,\cdot)$ to be linear in the second argument, in agreement with the usual convention in physics literature.}
\be
(F_1,F_2) = \!\int\limits_{\Gamma \backslash G}\!\!d\mu(g) \overline{F_1(g)}F_2(g)\,.
\label{eq:innerProduct}
\ee
Since $\mu$ is invariant under right multiplication, this inner product is also invariant. In other words, $L^2(\Gamma\backslash G)$ is a unitary representation of $G$ under the action
\be
\tilde{g}\in G: F(g) \mapsto F(g\tilde{g})\,.
\ee
In particular, $L^2(\Gamma\backslash G)$ is a unitary representation of the maximal compact subgroup $K\subset G$. It is instructive to decompose $L^2(\Gamma\backslash G)$ as a direct sum according to the $K$-action
\be
L^2(\Gamma\backslash G) = \bigoplus\limits_{n\in\mathbb{Z}} V_n\,.
\ee
Here $V_n$ consists of elements which transform as $r_{\theta}\cdot F = e^{- i n\theta} F$ under the action of $r_{\theta}\in K$. In particular, elements of $V_0$ are invariant under $K$ and therefore $V_0 = L^2(X)$. More generally, $V_n$ is the space of square-integrable sections of the $n$th power of the canonical line bundle over $X$. Indeed, let $F\in V_n$. Then we can write $F(x,y,\theta) = y^{|n|}e^{-i n\theta}h(x,y)$ where $h$ is a function on $\mathbb{H}$. Invariance of $F$ under left multiplication by $\Gamma$ implies that $h$ transforms under $\Gamma$ like a holomorphic modular form of weight $2n$ if $n\geq 0$ and like an anti-holomorphic modular form of weight $-2n$ when $n\leq 0$
\be
\forall \gamma = \pm\begin{pmatrix}a & b\\c & d\end{pmatrix}\in\Gamma:\qquad h(z) =
\left\{\,
 \begin{array}{@{}ll@{}}
(cz+d)^{-2n}\,h(\gamma\cdot z)\quad &\text{if }n\geq 0\vspace{2pt}\\
(c\bar{z}+d)^{2n}\,h(\gamma\cdot z)\quad &\text{if }n\leq 0.
\end{array}
\right.
\label{eq:formTransformation}
\ee
However, $h$ is not necessarily holomorphic or antiholomorphic.\footnote{By a small abuse of notation, we will use $h(x,y)$ and $h(z)$ interchangeably. $h(z)$ will not be holomorphic unless stated otherwise.} Another way to state \eqref{eq:formTransformation} is that $h(z)dz^n$, $h(z)d\bar{z}^{|n|}$ is $\Gamma$-invariant if $n\geq 0$, $n\leq 0$ respectively. Hence $h(z)dz^n$ is a section of the $n$th power of the canonical line bundle over $X$.

Let $L_{-1},L_{0},L_{1}$ be the basis for the complexified Lie algebra of $G$, introduced in Section~\ref{ssec:lieAlgebra}. Then if $F\in V_n$ is a smooth function, the action of $L_m$ on $F$ is defined, and we have $L_{0}\cdot F=n\,F$ and $L_{\pm 1}\cdot F \in V_{n\mp 1}$. In summary, $L^2(\Gamma\backslash G)$ nicely unifies spaces of sections of all powers of the canonical line bundle over $X$ into a single space, which, unlike the individual direct summands $V_n$, is a representation of $G$.

\subsection{The spectral decomposition}\label{ssec:spectralDecomposition}
The next step is to decompose $L^2(\Gamma\backslash G)$ into unitary irreducible representations of $G$. It is a general results that thanks to the compactness of $\Gamma\backslash G$, this decomposition is a discrete direct sum \cite{GelfandGraev}. The unitary irreducible representations of $\PSL(2,\mathbb{R})$ were classified in \cite{Bargmann}. The complete list consists of
\begin{enumerate}
\item the trivial representation
\item holomorphic discrete series $\cD_n$ and anti-holomorphic discrete series $\bar\cD_n$ with $n\in\mathbb{Z}_{>0}$
\item principal series $\cP^{+}_{i\nu}$ with $\nu\in \R_{\geq 0}$
\item complementary series $\cC_s$ with $s\in (0,\thalf)$\,.
\end{enumerate}
We give a more detailed discussion of these representations in Section \ref{sec:reptheory}. Let us explain the interpretation of the four representation types when they appear as subspaces of $L^2(\Gamma\backslash G)$. Our exposition will follow \cite{Gelbart1975}.

\subsubsection*{The trivial representation} Firstly, the trivial representation appears exactly once in $L^2(\Gamma\backslash G)$ and corresponds to constant functions. Indeed, constant functions are square-integrable due to compactness of $\Gamma\backslash G$. No other functions can transform in the trivial representation since the action of $G$ on $\Gamma\backslash G$ is transitive.

\subsubsection*{The discrete series} An ingredient crucial for our analysis is that holomorphic discrete series representations inside $L^2(\Gamma\backslash G)$ correspond to holomorphic modular forms for $\Gamma$. To see that, first note that $\cD_n$ are lowest-weight representations, meaning that they contain a nonzero vector annihilated by the Lie algebra element $L_{1}$. Suppose then that $F \in L^2(\Gamma\backslash G)$ is a lowest-weight vector of representation $\cD_n$. Then $F\in V_n$ and we have $F(x,y,\theta) = y^{n}e^{-i n\theta}h(z)$ as in the previous subsection. Now, $L_1$ is represented on smooth functions in $L^2(\Gamma\backslash G)$ by the differential operator $L_1 = -e^{i\theta}\left[y(\partial_x + i\partial_y)+\partial_{\theta}\right]$. Therefore, the condition $L_1\cdot F = 0$ is equivalent to $(\partial_x + i\partial_y)h(z) = 0$, i.e.\ $h(z)$ is holomorphic. We conclude that $h(z)$ is a holomorphic modular form of weight $2n$ for the discrete group $\Gamma$.

Conversely, given such a modular form $h(z)$, we can construct a basis for a copy of $\cD_n$ inside $L^2(\Gamma\backslash G)$ by acting with powers of the raising operator $L_{-1} = e^{-i\theta}\left[y(\partial_x - i\partial_y)+\partial_{\theta}\right]$ on the lowest-weight vector $F(x,y,\theta) = y^{n}e^{-i n\theta}h(z)$. In other words, there is a one-to-one correspondence between normalized holomorphic modular forms of weight $2n$ for the group $\Gamma$ and irreducible subspaces of $L^2(\Gamma\backslash G)$ isomorphic to the discrete series representation $\cD_n$.\footnote{We must include the adjective normalized because constant multiples of $h$ all generate the same subspace of $L^2(\Gamma\backslash G)$.}

Note that the defining inner product \eqref{eq:innerProduct} of $L^2(\Gamma\backslash G)$, when restricted to lowest-weight vectors in $V_n$, becomes the Petersson inner product of weight-$2n$ modular forms
\be
(F_1,F_2) =
\frac{1}{\vol(\Gamma\backslash \mathbb{H})}\int\limits_{\Gamma\backslash\mathbb{H}}\!\!dxdy\,y^{2n-2}\,\overline{h_{1}(z)}h_2(z)\,.
\ee

Finally, if $F$ is a lowest-weight vector in $\cD_{n}\subset L^2(\Gamma\backslash G)$, then $\overline{F} = y^{n}e^{i n\theta}\overline{h(z)} \in V_{-n}$ is a highest-weight vector of an anti-holomorphic discrete series $\bar\cD_n$. This means $\cD_{n}$ always appears in $L^2(\Gamma\backslash G)$ together with its dual $\bar\cD_n$.

\subsubsection*{The principal and complementary series} Principal and complementary series representations inside $L^2(\Gamma\backslash G)$ correspond to Maass forms for $\Gamma$, i.e.\ eigenfunctions of the hyperbolic Laplacian on $X$. To see that, it is convenient to label these representations by their value of the quadratic Casimir. In the notation of Section~\ref{ssec:lieAlgebra}, we have $-c_2 = \lambda = \Delta(1-\Delta)$. For the principal series $\cP^{+}_{i\nu}$, we have $\Delta = \frac{1}{2}+i\nu$, and thus $\lambda \in [\frac{1}{4},\infty)$, while for the complementary series $\cC_s$, we have $\Delta = \frac{1}{2}+s$ and so $\lambda \in (0,\frac{1}{4})$. Now, each principal or complementary series representation contains a precisely one-dimensional subspace of vectors invariant under $K$. Let $F\in V_0\subset L^2(\Gamma\backslash G)$ be such vector, belonging to $\cP^{+}_{i\nu}$ or $\cC_s$. In particular, $F(x,y,\theta) = h(x,y)$, where $h(\gamma\cdot z) = h(z)$ for all $\gamma\in\Gamma$. The quadratic Casimir is represented on smooth functions in $L^2(\Gamma\backslash G)$ by the differential operator $c_2 =  \left[y^2(\partial^2_x+\partial^2_y) + 2y\partial_x\partial_{\theta}\right]$. It follows that	
\be
-y^2 (\partial_x^2+\partial_y^2)h(x,y) = \lambda\,h(x,y)\,,
\ee
i.e.\ $h$ is an eigenfunction of the Laplacian on $X$ with eigenvalue $\lambda$.

Conversely, given such an eigenfunction $h$, we can construct a basis for a copy of $\cP^{+}_{i\nu}$ or $\cC_s$ inside $L^2(\Gamma\backslash G)$ by acting with powers of the raising and lowering operators on $h$. In other words, there is a one-to-one correspondence between normalized Maass forms for the group $\Gamma$ and irreducible subspaces of $L^2(\Gamma\backslash G)$ isomorphic to the principal or complementary series representations.

Finally, we note that the defining inner product of $L^2(\Gamma\backslash G)$, when restricted to $V_0$, becomes the standard inner product on $L^2(X)$
\be
(F_1,F_2) =
\frac{1}{\vol(\Gamma\backslash \mathbb{H})}\int\limits_{\Gamma\backslash\mathbb{H}}\!\!dxdy\,y^{-2}\,\overline{h_{1}(z)}h_2(z)\,.
\ee

\subsubsection*{Summary}
We can summarize the above as follows. The decomposition of $L^2(\Gamma\backslash G)$ into irreducibles takes the form
\be\label{eq:SL2RdecompositionFine}
	L^2(\G\@ G)=\C\oplus\bigoplus_{n=1}^\oo (\sD_n\oplus\bar\sD_n)\oplus\bigoplus_{k=1}^\oo \sC_{\l_k}.
\ee
Here $\mathbb{C}$ is the trivial representation. $\sD_n$ and $\bar{\sD}_n$ are isomorphic respectively to $\mathbb{C}^{\ell_n}\otimes \cD_{n}$ and $\mathbb{C}^{\ell_n}\otimes \bar{\cD}_{n}$ for some $\ell_n\geq 0$. Here $\ell_n$ is the number of independent copies of $\cD_n$ inside $L^2(\G\@ G)$, or equivalently the number of linearly independent holomorphic modular forms of weight $2n$ for $\Gamma$. Finally, in the last term, $\lambda_1<\lambda_2<\ldots$ are the distinct positive eigenvalues of the Laplacian on $X$. Let $d_k\geq 1$ be the multiplicity of $\lambda_k$. Then
\be\label{eq:sNota}
\sC_{\l_k} \simeq
\left\{\,
 \begin{array}{@{}lll@{}}
\mathbb{C}^{d_k}\otimes \cC_s &\text{with }\lambda_k = \frac{1}{4}-s^2 &\text{ if }\lambda_k<\frac{1}{4}
\vspace{6pt}\\
\mathbb{C}^{d_k}\otimes \cP^{+}_{i\nu} &\text{with } \lambda_k = \frac{1}{4}+\nu^2&\text{ if }\lambda_k\geq\frac{1}{4}\,.
\end{array}
\right.
\ee

\subsubsection*{Multiplicity of the discrete series} It will be convenient to associate to each $\Gamma$ its \emph{signature} $[g;k_1,\ldots,k_r]$. Here $g\in \mathbb{Z}_{\geq 0}$ is the genus of $X=\Gamma\backslash\mathbb{H}$ and $k_1,\ldots,k_r \in \mathbb\{2,3,\ldots\}$ are the orders of its orbifold points. The multiplicities $\ell_{n}$ of the discrete series $\cD_n$ are entirely determined by the signature. Specifically, one can use the Riemann-Roch theorem to show (\cite{MilneModularForms}, Theorem 4.9)
\be
	\ell_n=(2n-1)(g-1)+\sum_{i=1}^r\left\lfloor n\frac{k_i-1}{k_i}\right\rfloor+\de_{n,1}.
\ee
In particular, the number of linearly independent holomorphic one-forms on $X$ is $\ell_1 = g$. We will be able to use this partial information about the decomposition \eqref{eq:SL2RdecompositionFine} as an input to constrain the remaining spectral data at fixed signature.

\subsubsection*{Remark about notation}
Here and in the following, we try to follow the convention where objects existing inside $L^2(\Gamma\backslash G)$ are denoted by the script font, e.g.\ $\sD_{n}$, $\sC_{\lambda}$ and $\sO(z)$. On the other hand, objects existing independently of their concrete realization in $L^2(\Gamma\backslash G)$ are denoted using the calligraphic font, e.g.\ $\cD_{n}$, $\cP^{+}_{i\nu}$ and $\cO(z)$.

\subsection{Product expansion}\label{ssec:productExpansion}
Given the direct sum of a random collection of unitary irreducible representations of $G$, there is generally no reason for it to arise as $L^2(\Gamma\backslash G)$ for some $\Gamma$. In fact, $L^2(\Gamma\backslash G)$ admits an additional structure whose existence severely restricts spectra which can possibly appear on the RHS of \eqref{eq:SL2RdecompositionFine}.

This extra structure is the pointwise product of smooth functions on $\Gamma\backslash G$. Let  $C^{\infty}(\Gamma\backslash G)$ stand for the space of smooth functions on $\Gamma\backslash G$. Thanks to compactness of $\Gamma\backslash G$, we have the inclusion $C^{\infty}(\Gamma\backslash G)\subset L^2(\Gamma\backslash G)$. The product of smooth functions is smooth, and thus we get the symmetric bilinear map
\be
C^{\infty}(\Gamma\backslash G)\times C^{\infty}(\Gamma\backslash G)\rightarrow C^{\infty}(\Gamma\backslash G):\quad (F_1,F_2)\mapsto F_1F_2\,.
\label{eq:product}
\ee
Crucially, this map commutes with the action of $G$ on $L^2(\Gamma\backslash G)$.

The product $F_1F_2$ can be expanded using the direct sum decomposition \eqref{eq:SL2RdecompositionFine}
\be
F_1F_2 = P_{\C}(F_1F_2) + \sum\limits_{n=0}^{\infty}[P_{\sD_n}(F_1F_2)+P_{\bar\sD_n}(F_1F_2)] + \sum\limits_{k=1}^{\infty}P_{\sC_{\l_k}}\!(F_1F_2)\,,
\label{eq:productExpansion}
\ee
where $P_{\C},\,P_{\sD_n},\,P_{\bar\sD_n},\,P_{\sC_{\l_k}}$ are the orthogonal projections onto the respective subspaces in the decomposition~\eqref{eq:SL2RdecompositionFine}. This expansion is analogous to the operator product expansion (OPE) in conformal field theory.

We will use the notation $H$ to stand for any of the subspaces $\C,\,\sD_n,\,\bar\sD_n,\,\sC_{\l_k}$. A point central to our analysis is that the dependence on $F_1$ and $F_2$ of the individual summands $P_{H}(F_1F_2)$ appearing on the RHS of \eqref{eq:productExpansion} is strongly constrained by $G$-invariance. Indeed, suppose $F_1$ ranges over $H_i^{\infty}$ and $F_2$ ranges over $H_j^{\infty}$, where $H_{i}^{\infty}$ stands for the space of smooth functions inside $H_i$. Then $G$-invariance alone fixes $P_{H_m}(F_1F_2)$ up to finitely many constants. This is because the space of $G$-invariant maps
\be
H_i^{\infty}\times H^{\infty}_{j} \rightarrow H^{\infty}_{m}
\ee
is finite-dimensional. For example, suppose $H_i \simeq \cD_{n_1}$, $H_j \simeq \cD_{n_2}$ and that $H_m$ is irreducible. Then the space of such maps is one-dimensional if $H_m\simeq \cD_{n_3}$ with $n_3\geq n_1+n_2$ and zero-dimensional otherwise. Similarly, suppose $H_i \simeq \cD_{n}$, $H_j \simeq \bar{\cD}_{n}$ and that $H_m$ is irreducible. Then the space of such maps is one-dimensional if $H_m$ is the trivial representation, principal or complementary series, and zero-dimensional otherwise. We will prove these claims in Section~\ref{sec:proofs}.\footnote{Strictly speaking, we only prove an upper bound on the dimension.}

We will see that the finitely-many constants in $P_{H_m}(F_1F_2)$ which are left undetermined by representation theory can be identified with triple product integrals between elements of $H_i$, $H_j$ and $\bar{H}_m$.

The bootstrap is the idea of considering a triple product $F_1F_2F_3\in L^2(\Gamma\backslash G)$ and using the expansion \eqref{eq:productExpansion} twice in succession and in two inequivalent ways
\be
((F_1F_2)F_3) = (F_1(F_2 F_3))\,.
\label{eq:associativity}
\ee
We can take $F_1$, $F_2$ and $F_3$ to range over $K$-finite vectors in $H_i$, $H_j$ and $H_k$. In this way, \eqref{eq:associativity} leads to an infinite set of identities satisfied by the spectral data. The identities are bilinear in the triple product integrals. Before we explain how to obtain such identities in practice, it will be useful to introduce the concepts of coherent states and function correlators.

\subsection{Coherent states}\label{ssec:coherentStates}
Our ultimate goal is to apply \eqref{eq:associativity} with $F_{1},F_{2},F_{3}$ belonging to discrete series representations. To that end, we will first introduce a parametrization of vectors in these representations in the form of coherent states. First, pick an orthonormal basis of lowest-weight vectors in each $\sD_{n}$, i.e.\ $\{F_{n,a} \in \sD_{n}\cap V_n:\,a=1,\ldots,\ell_n\}$, where $(F_{m,a},F_{n,b})=\delta_{m,n}\delta_{a,b}$. Note that $F_{n,a} \in C^{\infty}(\Gamma\backslash G)$. Then $\overline{F_{n,a}}$ is an orthonormal basis of highest-weight vectors in each $\bar{\sD}_{n}$. We will make the following
\begin{definition}[Coherent state]\label{def:coherent}
The coherent state $\sO_{n,a}(z)$ in the holomorphic discrete series representation $\sD_n$ is
$$
\sO_{n,a}(z) = \exp(z\,L_{-1})\cdot F_{n,a}\,,
$$
where $z\in\mathbb{C}$, $|z|<1$. Similarly, the coherent state $\widetilde{\sO}_{n,a}(z)$ in the anti-holomorphic discrete series representation $\bar{\sD}_n$ is
$$
\widetilde{\sO}_{n,a}(z) = z^{-2n} \exp(-z^{-1}L_{1})\cdot \overline{F_{n,a}}\,,
$$
where $z\in\mathbb{C}$, $|z|>1$. Finally, define
$$
\widetilde{\sO}_{n,a}(\infty) = \lim_{z\rightarrow\infty} z^{2n} \widetilde{\sO}_{n,a}(z) = \overline{F_{n,a}}\,.
$$
\end{definition}
Here $\exp(z\,L_{-1})$ and $\exp(-z^{-1}L_{1})$ are elements of the complexified group $G_{\mathbb{C}} = \PSL(2,\mathbb{C})$. We will show in section~\ref{ssec:coherentStatesRigorous} that their action on $F_{n,a}$ is well-defined. Note that $z$ here is distinct from the variable parametrizing $\mathbb{H}$. We will show in Section~\ref{ssec:coherentStatesRigorous} that $\sO_{n,a}(z)\in \sD_n^{\infty}$ for all $|z|<1$ and $\widetilde{\sO}_{n,a}(z)\in \bar{\sD}_n^{\infty}$ for all $|z|>1$. In particular, the coherent states represent smooth functions on $\Gamma\backslash G$. By expanding $\sO_{n,a}(z)$ around $z=0$ and $\widetilde{\sO}_{n,a}(z)$ around $z=\infty$, we can think of them as generating functions for all $K$-finite vectors in $\sD_n$ and $\bar{\sD}_n$.

The coherent states are directly analogous to local primary operators in a 1D CFT, with $z$ parametrizing one-dimensional complexified spacetime. Indeed, the Lie algebra acts on the coherent states just like the 1D conformal algebra acts on local operators
\be
\begin{aligned}
L_{-1}\cdot \sO_{n,a}(z) &= \partial_{z} \sO_{n,a}(z)\\
L_{0}\cdot \sO_{n,a}(z) &= (z\partial_{z}+n) \sO_{n,a}(z)\\
L_{1}\cdot \sO_{n,a}(z) &= (z^2\partial_{z}+2n z) \sO_{n,a}(z)
\end{aligned}
\qquad
\begin{aligned}
L_{-1}\cdot \widetilde{\sO}_{n,a}(z) &= \partial_{z} \widetilde{\sO}_{n,a}(z)\\
L_{0}\cdot \widetilde{\sO}_{n,a}(z) &= (z\partial_{z}+n) \widetilde{\sO}_{n,a}(z)\\
L_{1}\cdot \widetilde{\sO}_{n,a}(z) &= (z^2\partial_{z}+2n z) \widetilde{\sO}_{n,a}(z)\,,
\end{aligned}
\ee
see Section~\ref{ssec:coherentStatesRigorous} for a proof. To simplify notation, in the following, $\sO_n(z)$ will denote any one of $\sO_{n,a}(z)$, e.g.\ $\sO_n(z) = \sO_{n,1}(z)$, and similarly $\widetilde{\sO}_n(z)=\widetilde{\sO}_{n,1}(z)$. Our next goal is to describe the product expansion of coherent states.

\subsubsection*{Product expansion of $\sO\sO$}
We will start with the product $\sO_n(z_1)\sO_n(z_2)$. It is possible to show that $P_H(\sO_n(z_1)\sO_n(z_2))=0$ unless $H=\sD_p$ with $p$ even and $p\geq 2n$, see Lemma~\ref{lemma:DDselection}.

This result implies that the product $\sO_n(z_1)\sO_n(z_2)$ has the following expansion
\be
	\sO_n(z_1)\sO_n(z_2)=\sum_{p=2n\atop p\text{ even}}^{\oo}P_{\sD_p}(\sO_n(z_1)\sO_n(z_2)).
\ee
The restriction on parity of $p$ comes from symmetry of both sides under $z_1\leftrightarrow z_2$.

$P_{\sD_p}(\sO_n(z_1)\sO_n(z_2))$ is uniquely determined by $G$-invariance up to finitely many structure constants $f_{p,a}$, with $a\in\{1,\ldots,\ell_p\}$. The full product expansion takes the form
\be
\sO_n(z_1)\sO_n(z_2) =
\sum_{p=2n\atop p\text{ even}}^{\oo}\sum_{a=1}^{\ell_p}\sum\limits_{m=0}^{\infty}
f_{p,a}\,(z_1-z_2)^{p+2m-2n}
\frac{(p)_m}{m!(2p)_m}
L_{-1}^m\cdot\sO_{p,a}(z_2)\,.
\label{eq:productOO}
\ee
Here $(p)_m = p(p+1)\ldots(p+m-1)$ is the Pochhammer symbol. The structure constants $f_{p,a}$ will be identified with triple product integrals $\int\!d\mu \,F_n F_n \overline{F_{p,a}}$. This is explained in section~\ref{sec:correlators}.

\subsubsection*{Product expansion of $\sO\widetilde{\sO}$}
Similarly, we will show in the proof of Lemma~\ref{lemma:DDbarselection} that
\be
\sO_n(z_1)\widetilde{\sO}_n(z_2) = P_{\mathbb{C}}(\sO_n(z_1)\widetilde{\sO}_n(z_2))+
\sum_{k=1}^{\infty}P_{\sC_k}(\sO_n(z_1)\widetilde{\sO}_n(z_2))\,.
\label{eq:productOOb}
\ee
It follows from $G$-invariance that
\be
P_{\mathbb{C}}(\sO_n(z_1)\widetilde{\sO}_n(z_2)) = \frac{1}{(z_1-z_2)^{2n}}\,,
\ee
where the RHS is interpreted as a constant function on $\Gamma\backslash G$. Note that it is always finite since $|z_1|<1$ and $|z_2|>1$. The overall normalization of the RHS is a consequence of $\|F_n\|^2 = 1$.

In order to describe the projection $P_{\sC_k}(\sO_n(z_1)\widetilde{\sO}_n(z_2))$, we will introduce in section~\ref{sec:continuouscoherent} continuous series coherent states $\mathbb{O}_{k,a}(z)$ for $|z|=1$. They have the property that $\widehat{F}_{k,a}\equiv N_k\int_{|z|=1}|dz|\mathbb{O}_{k,a}(z)$ form a real orthonormal basis of Maass forms. Here $N_k$ are normalization constants.
In terms of these states, the expansion takes the form
\be\label{eq:sect2oob}
	\sO_n(z_1)&\widetilde{\sO}_n(z_2)=\nn\\
	&\frac{1}{(z_1-z_2)^{2n}}
	+\sum_{k=1}^{\infty}\sum_{a=1}^{d_k}\int_{|z|=1}|dz|\frac{z_2^{-2n}c_{k,a}N_k}{(1-z_1z_2^{-1})^{2n-\De_k}(1-z_1z^{-1})^{\De_k}(1-z_2^{-1}z)^{\De_k}}\mathbb{O}_{k,a}(z).
\ee

We can describe the structure constants $c_{k,a}$ alternatively as follows. In the $z_1=0$, $z_2=\infty$ limit, which projects onto $V_0$ the above equation yields
\be
\sO_n(0)\widetilde{\sO}_n(\infty) = F_{n}\overline{F_{n}} =
1 + \sum_{k=1}^{\oo}\sum_{a=1}^{d_k}c_{k,a} \widehat{F}_{k,a}\,,
\label{eq:productOObSimp}
\ee 
which is now an equation between objects on $X$.

\subsection{Correlators}
\label{sec:correlators}
It will be convenient to introduce the following terminology and notation
\begin{definition}[Correlator]
Let $F_1,\ldots,F_n\in C^{\infty}(\Gamma\backslash G)$. The $n$-function correlator $\<F_1\cdots F_n\>$ is the integral
\be
\<F_1\cdots F_n\> = \int\limits_{\G\@G}\! d\mu(g) F_1(g)\cdots F_n(g).
\ee
\end{definition}
Note that $F\mapsto \int_{\Gamma\backslash G}d\mu(g) F(g)$ is a $G$-invariant map from $L^2(\Gamma\backslash G)$ to the trivial representation. Therefore, it must be a constant multiple of the projection $P_{\mathbb{C}}$. Since we had normalized $\mu(\Gamma\backslash G) = 1$, this constant is one and we have $\<F_1\cdots F_n\> = P_{\mathbb{C}}(F_1\cdots F_n)$. Note that the 2-function correlators are related to the inner product on $L^2(\Gamma\backslash G)$ as follows $(F_1,F_2)=\<\overline{F_1}F_2\>$.

Let us discuss correlators of coherent states $\sO_{n,a}(z)$ and $\widetilde{\sO}_{n,a}(z)$. It follows directly from $G$-invariance that the one-function correlators vanish
\be
\<\sO_{n,a}(z)\> = \<\widetilde{\sO}_{n,a}(z)\> = 0\,.
\ee
By the results of the previous subsection, the form of the two-function correlators must be
\be
\<\sO_{m,a}(z_1)\sO_{n,b}(z_2)\> = 0\,, \qquad
\<\sO_{m,a}(z_1)\widetilde{\sO}_{n,b}(z_2)\> = \frac{\delta_{m,n}\delta_{a,b}}{(z_1-z_2)^{2n}}\,.
\ee
The three-function correlators $\<\sO_n(z_1)\sO_n(z_2)\widetilde{\sO}_{p,a}(z_3)\>$ are uniquely determined by $G$-invariance up to an overall constant $f_{p,a}$
\be
\<\sO_n(z_1)\sO_n(z_2)\widetilde{\sO}_{p,a}(z_3)\> = f_{p,a}\,\frac{z_{12}^{p-2n}}{z_{13}^{p}z_{23}^{p}}\,,
\label{eq:threePointFunction}
\ee
where we introduced the notation $z_{ij} = z_{i}-z_{j}$. Coefficients $f_{p,a}$ appearing here are the same as the structure constants appearing in the product expansion \eqref{eq:productOO}. In fact, formulas \eqref{eq:productOO} and \eqref{eq:threePointFunction} are completely equivalent, as can be checked by multiplying both sides of \eqref{eq:productOO} by $\widetilde{\sO}_{p,a}(z_3)$ and evaluating the correlator.

We see that as promised, the structure constants $f_{p,a}$ are related to triple product integrals of modular forms. In particular, for $p=2n$, \eqref{eq:threePointFunction} implies
\be
f_{2n,a} = \<\sO_n(0)\sO_n(0)\widetilde{\sO}_{2n,a}(\infty)\> = \int\limits_{\G\@G}\!\! d\mu\,F_{n,1}F_{n,1}\overline{F_{2n,a}}\,.
\label{eq:f2na}
\ee
It is easy to rewrite this as an integral over $X = \Gamma\backslash\mathbb{H}$. Let $F_{n,a}(x,y,\theta) = y^n e^{-i n\theta} h_{n,a}(z)$, so that $h_{n,a}$ with $a\in\{1,\ldots,\ell_n\}$ is an orthonormal basis for the space of holomorphic modular forms of weight $2n$. Then \eqref{eq:f2na} becomes
\be
f_{2n,a} =  \frac{1}{\vol(\Gamma\backslash \mathbb{H})}\int\limits_{\Gamma\backslash\mathbb{H}}\!\!dxdy\,y^{4n-2}\,h_{n,1}h_{n,1}\overline{h_{2n,a}}\,.
\ee
Alternatively, $f_{2n,a}$ are the coefficients of the expansion of the product $h_{n,1}h_{n,1}$ in terms of $h_{2n,a}$.

Coefficients $f_{p,a}$ with $p>2n$ are proportional to the inner product of the Rankin-Cohen bracket~\cite{Zagier} $[h_{n,1},h_{n,1}]_{p-2n}$ with $h_{p,a}$. Equivalently, they express the Rankin-Cohen bracket $[h_{n,1},h_{n,1}]_{p-2n}$ as a linear combination of $h_{p,a}$.

Similarly, the structure constants $c_{k,a}$ appearing in the product expansion $\sO_n(0)\widetilde{\sO}_n(\infty)$, equation~\eqref{eq:productOObSimp}, are related to the three-function correlators
\be
c_{k,a} = \<\sO_n(0)\widetilde{\sO}_n(\infty) \widehat{F}_{k,a}\> =
 \int\limits_{\G\@G}\!\! d\mu\,F_{n,1}\overline{F_{n,1}}\widehat{F}_{k,a}\,.
\ee
If we let $\widehat{F}_{k,a}(x,y,\theta)=\widehat{h}_{k,a}(x,y)$, where $\widehat{h}_{k,a}$ is an orthonormal basis of real Maass forms, we get
\be
c_{k,a} =
\frac{1}{\vol(\Gamma\backslash \mathbb{H})}\int\limits_{\Gamma\backslash\mathbb{H}}\!\!dxdy\,y^{2n-2}\,h_{n,1}\overline{h_{n,1}}\widehat{h}_{k,a}\,.
\ee
In particular, $\overline{c_{k,a}} = c_{k,a}$.

\subsection{Four-function correlator and the crossing equation}
We are now ready to discuss the main object of our study, namely the four-function correlator $\<\sO_n(z_1)\sO_n(z_2)\widetilde\sO_n(z_3)\widetilde\sO_n(z_4)\>$. It is determined by $G$-invariance in terms of a function of a single variable $g(z)$
\be\label{eq:simplefourpt}
	\<\sO_n(z_1)\sO_n(z_2)\widetilde\sO_n(z_3)\widetilde\sO_n(z_4)\>=
		\frac{1}{z_{12}^{2n}z_{34}^{2n}}
		g(z),
\ee
where
\be
z = \frac{z_{12}z_{34}}{z_{13}z_{24}}
\ee
is the cross-ratio of the four points. The central idea used in this paper is that $g(z)$ can be evaluated using pairwise product expansion in two inequivalent ways. Either we use the product expansion of $\sO_n(z_1)\sO_n(z_2)$, or the product expansion of $\sO_n(z_2)\widetilde\sO_n(z_3)$. In the first case, we get $g(z)$ in the form of an infinite sum over the spectrum of holomorphic modular forms. In the second case, the sum is over the spectrum of Maass forms. Equality of the two expansions can be thought of as a special case of associativity of the product expansion \eqref{eq:associativity}.

In section~\ref{sec:proofs} we will prove the following

\begin{theorem}\label{theorem:main}
	Let $g(z)$ be as in~\eqref{eq:simplefourpt}. Then
	\begin{enumerate}
		\item We have the following expansion, which we refer to as the $s$-channel expansion
		\be\label{eq:schannel}
			g(z)=\sum_{p=2n\atop p\text{ even}}^{\oo}\sum_{a=1}^{\ell_p}|f_{p,a}|^2 \cG_p(z),
		\ee
		where the structure constants $f_{p,a}\in \C$ are as above, and where the sum and its derivatives converge uniformly on compact subsets of $z\in\C\setminus[1,+\oo)$, and
		\be\label{eq:sblock}
			\cG_\De(z)\equiv z^\De {}_2F_1(\De,\De;2\De;z),
		\ee
		where ${}_2F_1$ is the Gauss hypergeometric function.
		\item We have the following expansion, which we refer to as the $t$-channel expansion
		\be\label{eq:tchannel}
		g(z)=\p{\frac{z}{1-z}}^{2n}\p{1+\sum_{k=1}^{\oo}\sum_{a=1}^{d_k}c_{k,a}^2 \cH_{\De_k}(z)},
		\ee
		where the structure constants $c_{k,a}\in \R$ are as above, $\Delta_k=1/2+i\sqrt{\lambda_k-1/4}$, where the sum and its derivatives converge uniformly on compact subsets of $z\in\C\setminus[1,+\oo)$,  and
		\be\label{eq:tblock}
		\cH_\De(z)\equiv {}_2F_1(\De,1-\De;1;\tfrac{z}{z-1}).
		\ee
	\end{enumerate}
\end{theorem}

This theorem is proved in section~\ref{sec:proofs} in the same way as anyone familiar with conformal field theory would imagine. We first constrain the terms appearing in the product expansions~\eqref{eq:productExpansion}, then compute the conformal blocks~\eqref{eq:sblock} and~\eqref{eq:tblock} using the Casimir equation. The key element, which is different from the usual conformal bootstrap story, is that $g(z)$ as well as the blocks appearing in the above expansions have to be analytic for $z\in\C\setminus[1,+\oo)$, owing to the analyticity properties of the coherent states $\sO_n(z)$: the product $\sO_n(z_1)\sO_n(z_2)$ is analytic for all $z_1$ in the unit disk, regardless of the presence of $z_2$.

Theorem~\ref{theorem:main} provides two expansions for the same quantity $g(z)$. Requiring the two to be consistent, we get the crossing equation
\be\label{eq:crossineq}
\sum_{p=2n\atop p\text{ even}}^{\oo}\sum_{a=1}^{\ell_p}|f_{p,a}|^2 \cG_p(z)=\p{\frac{z}{1-z}}^{2n}\p{1+\sum_{k=1}^{\oo}\sum_{a=1}^{d_k}c_{k,a}^2 \cH_{\De_k}(z)}.
\ee
This consistency condition and its generalizations can be analyzed using linear or semidefinite programming methods as is standard in the conformal bootstrap~\cite{Poland:2018epd}. We will explore some simple consequences of \eqref{eq:crossineq} in the following section, and perform a more extensive analysis in Section~\ref{sec:bounds}.

\subsection{Bootstrap}\label{ssec:bootstrap}
First, we rewrite~\eqref{eq:crossineq} in the following form,
\be\label{eq:crossineq21}
\sum_{p=2n\atop p\text{ even}}^{\oo}S_p\,\cG_p(z)=\p{\frac{z}{1-z}}^{2n}\p{1+\sum_{k=1}^{\oo}T_k\,\cH_{\De_k}(z)},
\ee
where $S_p=\sum_{a=1}^{\ell_p}|f_{p,a}|^2\geq 0$ and $T_k=\sum_{a=1}^{d_k}c_{k,a}^2\geq 0$.

The crossing equation~\eqref{eq:crossineq21} holds order by order in the expansion around $z=0$, starting at $z^{2n}$. Since $\cG_p(z) \sim z^{p}$ as $z\rightarrow 0$, we can use this expansion to write each $S_p$ as a sum over the Laplacian spectrum. Furthermore, note that each term on the LHS is manifestly symmetric under $z\leftrightarrow \frac{z}{z-1}$, which is just the symmetry of \eqref{eq:simplefourpt} under $z_1\leftrightarrow z_2$. On the other hand, the individual terms on the RHS of \eqref{eq:crossineq21} do not have this symmetry. It follows that by antisymmetrizing \eqref{eq:crossineq21} under $z\leftrightarrow \frac{z}{z-1}$, we obtain constraints on the Laplacian spectral data $\{(\lambda_k,T_k)\}$ which do not involve the holomorphic data $\{S_p\}$.

Both types of constraints can be written down efficiently starting from orthogonality of hypergeometric functions,\footnote{For $p=q$ this follows by evaluating the residue at $z=0$. For $p\neq q$ the result can be obtained by observing that $\cG_p(z)$ is an eigenfunction of the differential operator $z^2(1-z)\ptl_z^2-z^2\ptl_z$ with eigenvalue $p(p-1)$, and that this operator is symmetric with respect to the considered bilinear pairing.}
\be\label{eq:orthogonalityG}
\frac{1}{2\pi i}\oint dz z^{-2}\cG_{1-p}(z) \cG_{q}(z) = \delta_{p,q}\,.
\ee
The contour integral computes the residue at $z=0$. Applying this to \eqref{eq:crossineq2}, we find the spectral identities
\be\label{eq:crossingfinal}
-S_p+\cF_{p}^n(0)+\sum_{k=1}^\oo T_k \cF_{p}^n(\l_k)&=0,\qquad \text{for even }p\geq 2n,\\
\cF_{p}^n(0)+\sum_{k=1}^\oo T_k \cF_{p}^n(\l_k)&=0,\qquad\text{for odd }p> 2n\,,\label{eq:crossingfinal2}
\ee
where
\be
\cF_{p}^n(\l) = \frac{1}{2\pi i}\oint dz z^{-2}\cG_{1-p}(z)\left(\tfrac{z}{1-z}\right)^{2n}\cH_{1/2+i\sqrt{\lambda-1/4}}(z)\,.
\ee
The functions $\cF_p^n(\l)$ $(p\geq 2n)$ are in fact polynomials in $\l$ given by the equation~\eqref{eq:Fexpression} below. For example $\cF_{2n}^n(\l)=1$ and $\cF_{2n+1}^{n}(\l)=n-\l$, and the respective spectral identities read
\be
1+\sum\limits_{k=1}^{\infty}T_{k} = S_{2n}\,,\qquad n+\sum\limits_{k=1}^{\infty}(n-\lambda_k)T_{k} = 0\,.\label{eq:simplesumrules}
\ee
Since $T_k\geq 0$, the first identity immediately implies the lower bound on the triple overlaps $f_{2n}\geq 1$. It is also possible to form a linear combination of the $p=2n+1$ and $p=2n+3$ identities from which the contribution of $\lambda = 0$ drops out,
\be
\sum\limits_{k=1}^{\infty}\lambda_{k}[\lambda_k^2 -(9n+1)\lambda_k +12n^2]T_k = 0\,.
\label{eq:identity13}
\ee
The quadratic polynomial $\lambda^2 -(9n+1)\lambda +12n^2$ has two positive real roots for any $n>0$. It follows that \eqref{eq:identity13} can only be satisfied if $\lambda_1$ is at most the greater root\footnote{Strictly speaking, this conclusion only follows if $T_k\neq 0$ (and hence $T_k>0$) for at least one value of $k$. This follows from the second identity in~\eqref{eq:simplesumrules}}
\be
\lambda_1\leq \frac{\sqrt{33 n^2+18 n+1}+9 n+1}{2}\,.
\label{eq:boundToy}
\ee
This inequality is a valid bound for all $X$ which admit a non-vanishing modular form of weight $2n$. It is possible to show that every hyperbolic 2-orbifold admits a modular form of weight at most 12. That means that we get a bound valid for all hyperbolic 2-orbifolds by substituting $n=6$ in \eqref{eq:boundToy}
\be
\lambda_1\leq \frac{\sqrt{1297}+55}{2}\approx 45.50694\,.
\ee
This value is not far from $\lambda_1$ of the orbifold with the smallest area, i.e.\ the $[0;2,3,7]$ orbifold, which has $\lambda_1\approx 44.88835$. The simple bound \eqref{eq:boundToy} as well as the lower bound $f_{2n}\geq 1$ can be systematically improved by exploiting the spectral identities \eqref{eq:crossingfinal} for $p\in \{2n,\ldots,p_{\text{max}}\}$. This will be explained in more detail in Section~\ref{sec:derivation}. The resulting bounds appear to converge rapidly with increasing $p_{\text{max}}$. Remarkably, both are nearly saturated by the $[0;2,3,7]$ orbifold.

\section{Representation theory and product expansions}\label{sec:proofs}
The main goal of this section is to supply a proof of Theorem~\ref{theorem:main}, as well as proofs of several further claims cited in Section~\ref{sec:bootstrapMethod}. Our proofs are facilitated by working with the standard realizations of unitary irreducible representations of $\PSL(2,\mathbb{R})$ in spaces of functions of a single variable, reviewed in Section~\ref{sec:reptheory}.

Many of the results presented in this section are well-known to experts on conformal field theory, although perhaps in a less rigorous form. We include them here for completeness and with the hope that they will be of independent interest.

\subsection{Lie algebra}\label{ssec:lieAlgebra}
Let us review some useful material about the Lie algebra $\mathfrak{g}=\mathfrak{sl}_2(\mathbb{R})$. Choose the following basis for $\mathfrak{g}$
\be
J_{-1} =
\begin{pmatrix}
0 & 1\\ 0 & 0
\end{pmatrix}\qquad
J_{0} =
\frac{1}{2}\begin{pmatrix}
1 & 0\\ 0 & -1
\end{pmatrix}\qquad
J_{1} =
\begin{pmatrix}
0 & 0\\ -1 & 0
\end{pmatrix}\,.
\ee
In unitary representations of $G$, these basis elements are represented by anti-self-adjoint operators $(J_{n})^{\dagger} = -J_{n}$. 
It will also be convenient to introduce a different basis for the complexification $\mathfrak{g}_{\mathbb{C}}$
\be
L_{-1}= \frac{1}{2}(J_{-1}-J_{1})-iJ_{0}\,,\qquad
L_0 = -\frac{i}{2}(J_{-1}+J_{1})\,,\qquad
L_{1}= -\frac{1}{2}(J_{-1}-J_{1})-iJ_{0}\,.
\ee
The point of this basis is that it diagonalizes the adjoint action of $K$ on $\mathfrak{g}_{\mathbb{C}}$. Indeed, $iL_0$ is the generator of $K$
\be
r_{\theta} = e^{-i \theta L_{0}}=
\pm\begin{pmatrix}
\cos\tfrac{\theta}{2} & -\sin\tfrac{\theta}{2}\vspace{0.02in}\\
\sin\tfrac{\theta}{2} &\cos\tfrac{\theta}{2}
\end{pmatrix}\in K
\label{eq:rotation}
\ee
and we have the commutation relations
\be
[L_{m},L_{n}] = (m-n) L_{m+n}\,.
\ee
In unitary representations, these basis elements satisfy $(L_{n})^{\dagger} = L_{-n}$. Finally, we will also need the quadratic Casimir element of $Z(\cU(\mathfrak{g}))$, which takes the form
\be\label{eq:casimir}
c_2 = L_0^2 - \frac{L_{-1}L_{1}+L_{1}L_{-1}}{2}= J_0^2 - \frac{J_{-1}J_{1}+J_{1}J_{-1}}{2}\,.
\ee

Consider now the unitary representation of $G$ on $L^2(\Gamma\backslash G)$, described in Section~\ref{ssec:L2Space}. In this representation, elements of $\mathfrak{g}_{\mathbb{C}}$ act on smooth functions in $L^2(\Gamma\backslash G)$ by differential operators. In particular, the above basis is represented by
\be
\begin{aligned}
L_{-1}\cdot F(x,y,\theta) &= e^{-i\theta}\left[y(\partial_x - i\partial_y)+\partial_{\theta}\right]\!F(x,y,\theta)\\
L_{0}\cdot F(x,y,\theta) &= i\partial_{\theta}F(x,y,\theta)\\
L_{1}\cdot F(x,y,\theta) &=-e^{i\theta}\left[y(\partial_x + i\partial_y)+\partial_{\theta}\right]\!F(x,y,\theta)\,.
\end{aligned}
\label{eq:LactionL2}
\ee
From these formulas, we find that the quadratic Casimir element acts as
\be
c_2\cdot F(x,y,\theta)  = \left[y^2(\partial^2_x+\partial^2_y) + 2y\partial_x\partial_{\theta}\right]\!F(x,y,\theta)\,.
\label{eq:C2L2}
\ee

\subsection{Unitary irreducible representations of $\PSL(2,\R)$}
\label{sec:reptheory}
Our next step is to review the classification of unitary irreducible representations of $\PSL(2,\R)$, in a form that is most useful for our purpose. We start with its double cover $\SL(2,\R)$. 
\begin{theorem}[See, e.g.~\cite{Knapp}]\label{thm:SL2Rirreps}
	Up to equivalence, the only unitary irreducible representations of $\SL(2,\R)$ are
\begin{enumerate}
	\item the trivial representation,
	\item the holomorphic discrete series $\cD_n$ and anti-holomorphic discrete series $\bar\cD_n$, $n\geq 1,\, 2n\in \Z$,
	\item the principal series $\cP^\pm_{i\nu}$ for $\nu\in \R$, except $\cP^-_0$,
	\item the complementary series $\cC_s$ for $s\in (0,\thalf)$,
	\item the limits of discrete series $\cD_{\half}$, $\bar\cD_{\half}$.
\end{enumerate}
	The only equivalences between these representations are $\cP^\pm_{i\nu}\simeq \cP^\pm_{-i\nu}$. (And the decomposition of $\cP^-_0$ is $\cP^-_0\simeq \cD_{\half}\oplus\bar\cD_{\half}$.)
\end{theorem}
The explicit construction of these representations is given, for example, in~\cite{Knapp}. We will only give the constructions for those which are representations of $\PSL(2,\R)$. The following proposition follows from the above classification and the explicit constructions.
\begin{proposition}\label{prop:PSL2Rirreps}
	Up to equivalence, the only unitary irreducible representations of $\PSL(2,\R)$ are
\begin{enumerate}
	\item the trivial representation,
	\item the holomorphic discrete series $\cD_n$ and anti-holomorphic discrete series $\bar\cD_n$, $n\geq 1,\, n\in \Z$,
	\item the principal series $\cP^+_{i\nu}$ for $\nu\in \R$,
	\item the complementary series $\cC_s$ for $s\in (0,\thalf)$\,.
\end{enumerate}
The only equivalences between these representations are $\cP^+_{i\nu}\simeq \cP^+_{-i\nu}$.
\end{proposition}
We will often use the notation $\cP^+_{s}$ with $s=i\nu$ ($\nu\in \R$) to treat $\cC_s$ and $\cP^+_s$ in a uniform way, and refer to them collectively as the continuous series. We will label the continuous series representations using the parameter $\De=\thalf+s$, which is real and in $(\thalf,1)$ for the complementary series, and complex with $\mathrm{Re}\, \De=\thalf$ for principal series.

We now give explicit constructions of all non-trivial unitary irreps of $\PSL(2,\R)$. For this, we first identify $\SL(2,\R)$ with $\SU(1,1)$ by conjugation inside $\SL(2,\C)$
\be
	\SU(1,1)=\begin{pmatrix}
		1 & i \\ i & 1
	\end{pmatrix}^{-1}
	\SL(2,\R)
	\begin{pmatrix}
		1 & i \\ i & 1
	\end{pmatrix}.
\ee
The elements $u\in\SU(1,1)$ have the explicit form
\be\label{eq:uparameters}
	u=\begin{pmatrix}
		\a & \b \\ \bar\b & \bar\a
	\end{pmatrix},\quad 
	u^{-1}=\begin{pmatrix}
	\bar\a & -\b \\ -\bar\b & \a
\end{pmatrix},
\ee
with $|\a|^2-|\b|^2=1$. Note that inside $\SU(1,1)$, elements of $K$ become diagonal matrices
\be
r_{\theta} = e^{-i \theta L_{0}}=
\begin{pmatrix}
e^{i\theta/2} & 0 \\ 0 & e^{-i\theta/2}
\end{pmatrix}\in K\,.
\ee

For $z\in \C$ we define the action of $\SU(1,1)$ by fractional linear transformations
\be
	u\. z \equiv  \frac{\a z+\b}{\bar\b z+\bar\a}.
\ee
It is easy to check that this action leaves invariant the unit circle $\ptl \D=\{z:|z|=1\}$, and thus also the unit disk $\D=\{z:|z|<1\}$, and the complementary disk $\D'\equiv \{z:|z|>1\}\cup\{\oo\}$. We will realize our representations in the spaces of functions on $\D,\D'$, and $\ptl \D$.

\subsubsection*{Discrete series}

The anti-holomorphic discrete series representations $\bar\cD_n$ are realized in the space of holomorphic functions $f$ on $\D$ for which the following norm is finite\footnote{We normalize $d^2z=dxdy$ where $z=x+iy$.}
\be\label{eq:Dbnorm}
	\|f\|^2_{\bar\cD_n}=\int\limits_{|z|<1}\!\!d^2z(1-|z|^2)^{2n-2}|f(z)|^2.
\ee
This norm also defines an inner product which makes $\bar\cD_n$ into a Hilbert space.\footnote{See~\cite{Hedenmalm}, Proposition 1.2, for the proof of completeness.} The group action is given by, with the notation of~\eqref{eq:uparameters},
\be\label{eq:Dbaction}
	(u\.f)(z)\equiv (-\bar\b z+\a)^{-2n} f(u^{-1}\.z).
\ee
It is easy to check that~\eqref{eq:Dbnorm} is invariant under this action, giving rise to a unitary representation of $\SU(1,1)\simeq \SL(2,\R)$. The basis $\{L_{-1},L_0,L_{1}\}$ of $\mathfrak{g}_{\mathbb{C}}$ is represented in $\bar\cD_n$ by the differential operators
\be
\begin{aligned}
(L_{-1}\cdot f)(z) &= -\partial_zf(z)\\
(L_{0}\cdot f)(z) &= -(z\partial_z+n)f(z)\\
(L_{1}\cdot f)(z) &= -(z^2\partial_z+2nz)f(z)\,.
\end{aligned}
\label{eq:Ldiscrete}
\ee
From this, we find the eigenvalue of the quadratic Casimir $c_2 = n(n-1)$. Let us decompose $\bar\cD_n$ into irreducible representations of $K$. This can be done by choosing the following basis for $\bar\cD_n$
\be
f_{j}(z) = z^{j}\quad\text{with}\quad j\in\mathbb{Z}_{\geq 0}.
\ee
Indeed, we have
\be
L_0 \cdot f_{j} = -(n+j)f_{j}\,.
\ee
In other words, the spectrum of $L_0$ in $\bar\cD_n$ consists of $\{-n,-n-1,\ldots\}$ and each of these weights appears exactly once. The highest-weight vector $f_0$ is annihilated by $L_{-1}$.

We could view the discrete series representations $\cD_n$ as complex conjugate to $\bar\cD_n$. That is, we could view them as being realized in the space of anti-holomorphic functions on $\D$, with the norm
\be
\|f\|^2_{0,\cD_n}=\int\limits_{|z|<1}d^2z(1-|z|^2)^{2n-2}|f(z)|^2,
\ee
and action given by
\be
	(u\.f)(z)\equiv (-\b \bar{z}+\bar\a)^{-2n} f(u^{-1}\.z).
\ee
However, it will be more convenient for us to consider a different realization, in which the group action is formally the same as~\eqref{eq:Dbaction}. Specifically, for $f$ anti-holomorphic in $\D$ we define $I f$ by
\be
	(If)(z)=z^{-2n}f((\bar{z})^{-1}).
\ee
We then find
\be
	(I(u\.f))(z)
	=\p{{-\bar\b z+\a}}^{-2n} (If)\p{u^{-1}\.z}\,.
\ee
It is easy to check that $I$ is a bijection between anti-holomorphic functions in $\D$ and holomorphic functions $f$ in $\D'$. By functions holomorphic in $\D'=\{z:|z|>1\}\cup\{\oo\}$ we mean functions such that $f(z)$ is holomorphic for finite $z$ and $|f(z)|\leqslant C|z|^{2n}$ for some $C>0$. This notion depends on $n$, but this will not cause any confusion in what follows.

Thus we will view $\cD_n$ as realized in the space of holomorphic functions in $\D'$, with action~\eqref{eq:Dbaction}, and the norm defined by
\be
	\|f\|^2_{\cD_n}\equiv\|I^{-1}f\|^2_{0,\cD_n}.
\ee
The quadratic Casimir for $\cD_n$ is also $c_2 = n(n-1)$. A basis for $\cD_n$ consisting of eigenvectors of $L_0$ can be chosen as $f_{j}(z) = z^{-2n-j}$ for $j\in\mathbb{Z}_{\geq 0}$ and we have
\be
L_0\cdot f_j= (n+j)f_j\,.
\ee
In other words, the spectrum of $L_0$ in $\bar\cD_n$ consists of $\{n,n+1,\ldots\}$ and each of these weights appears exactly once. The lowest-weight vector $f_0$ is annihilated by $L_{1}$.

\subsubsection*{Principal series}

The principal series representations $\cP_{i\nu}^+$ are realized in the space $L^2(\ptl\D)$ of (the equivalence classes of) square-integrable functions on the unit circle $\ptl\D$. In particular, we have the norm
\be\label{eq:Pnorm}
	\|f\|_{\cP^+_{i\nu}}^2\equiv\int\limits_{|z|=1} |dz|\, |f(z)|^2.
\ee
The group action is defined by
\be\label{eq:ContAction}
	(u\.f)(z) \equiv |-\b \bar{z}+\bar\a|^{-2\De} f(u^{-1}\.z),
\ee
where $\De=\half+i\nu$. It is easy to check that the norm~\eqref{eq:Pnorm} is invariant under this action. The quadratic Casimir for $\cP_{i\nu}^+$ is $c_2 = \Delta(\Delta-1)$. The set of functions $f_j(z) = z^j$ with $j\in\mathbb{Z}$ forms a basis for $\cP_{i\nu}^+$ and we have $L_0\cdot f_j = - j f_j$. Thus, the spectrum of $L_{0}$ consists of all integers and each integer weight appears exactly once. There are no lowest-weight or highest-weight vectors.

\subsubsection*{Complementary series}

The complementary series representations $\cC_{s}$ are realized in Hilbert space that is the completion of $L^2(\ptl\D)$ with respect to the norm
\be\label{eq:Cnorm}
	\|f\|_{\cC_s}^2\equiv\int\limits_{|z|=1,|w|=1} |dz\|dw|\, \frac{\overline{f(z)}f(w)}{|z-w|^{2-2\De}}.
\ee
Here $\De=\half+s$. The group action is defined again by~\eqref{eq:ContAction}. Just like for $\cP_{i\nu}^+$, the quadratic Casimir is $c_2 = \Delta(\Delta-1)$, and $f_j(z) = z^j$ with $j\in\mathbb{Z}$ is a basis for $\cC_{s}$.

\subsubsection*{Real and complex representations} Note that the representations $\cD_n$ and $\bar\cD_n$ are complex and complex-conjugate to each other. On the other hand, the complementary series representations $\cC_s$ and the principal-series representations $\cP_{i\nu}$ are real. It is straightforward to see in the former case, where a complex conjugation (i.e.\ an anti-unitary involution) can be defined as the point-wise complex conjugation of the functions in $\cC_s$. In the case of $\cP_{i\nu}$, the point-wise complex conjugation is an anti-unitary map to $\cP_{-i\nu}$, and needs to be followed by an isomorphism $\cP_{-i\nu}\simeq \cP_{i\nu}$. The resulting anti-unitary operator on $\cP_{i\nu}$ can be verified to be an involution. Alternatively, it is straightforward to construct the complex conjugation in the basis of $L_0$ eigenvectors. We will from now on assume that a standard complex conjugation has been fixed on $\cP_{i\nu}$ and $\cC_s$.

\subsubsection*{Cautionary remark}
We caution the reader that there is no direct relationship between the following two types of functions:
\begin{enumerate}
\item functions $f(z)$ representing vectors in the models of irreducible representations described in this subsection
\item automorphic functions $h(z)$, corresponding to vectors in the irreducible subspaces of $L^2(\Gamma\backslash G)$, described in Section~\ref{ssec:spectralDecomposition}
\end{enumerate}
In particular, the $z$ variables in (1) and (2) are not directly related. For example, a lowest-weight vector in the model of $\cD_{n}$ described in this subsection is a constant multiple of $f_0(z) = z^{-2n}$. On the other hand, we saw in Section~\ref{ssec:spectralDecomposition} that lowest-weight vectors belonging to the $\sD_{n}$ subspace of $L^2(\Gamma\backslash G)$ are weight-$2n$ modular forms $h(z)$. Unlike $f_0(z)$, these functions depend on $\Gamma$, there can be several linearly independent ones, and certainly none of them equals $z^{-2n}$!

On the other hand, the $z$ variable in (1) is the same as the variable parametrizing the coherent states $\sO_n(z)$, $\widetilde{\sO}_n(z)$, which is the subject of the next subsection.

\subsection{Coherent states in the discrete series}\label{ssec:coherentStatesRigorous}
In Section~\ref{ssec:coherentStates}, we defined the coherent states $\sO_n(z)\in\sD_{n}$ by formally acting with group elements in $\PSL(2,\mathbb{C})$ on lowest-weight vectors. In this subsection, we will identify these coherent states in the model for the discrete series described in Section \ref{sec:reptheory}. This is an important stepping stone towards the proof of Theorem~\ref{theorem:main}.

Let us first consider $\cD_n$, realized in the space of functions holomorphic in $\D'$. For a given $w\in \D$, we define the coherent state $\cO(w)\in\cD_n$ as the following function of $z\in \D'$
\be\label{eq:coherent}
\cO(w)(z) \equiv \sqrt{\frac{2n-1}{\pi}}\frac{1}{(z-w)^{2n}}.
\ee
Let us repeat for clarity that the variable $w\in \D$ labels different coherent states, and the variable $z\in \D'$ reminds us that each element of $\cD_n$ is a holomorphic function of $z\in \D'$.  

For $\bar\cD_n$, we define $\widetilde\cO(w)$ by the same formula~\eqref{eq:coherent}, but now $w\in \D'$ and $z\in \D$. Explicitly,
\be\label{eq:barcoherent}
\widetilde \cO(w)(z) \equiv \sqrt{\frac{2n-1}{\pi}}\frac{1}{(z-w)^{2n}}.
\ee
It is easy to verify that $\widetilde\cO$ satisfies the same transformation rule and has the same properties as $\cO$, except that $\widetilde\cO(w)$ is defined and holomorphic for $w\in\D'$. 
For future reference, we note that for $w\in \D'$ and $z\in \D'$
\be\label{eq:coherentinversion}
z^{-2n}\overline{\p{\widetilde \cO(w)((\bar{z})^{-1})}}&=z^{-2n}\sqrt{\frac{2n-1}{\pi}}\frac{1}{(z^{-1}-\bar{w})^{2n}}=(\bar{w})^{-2n}\cO((\bar{w})^{-1}).
\ee

In the following, we will suppress the second variable from our notation, and call the first variable $z$. The coherent states have the inner products
\be\label{eq:coherentinner}
	(\cO(z_1),\cO(z_2))_{\cD_n}=\frac{1}{({\bar{z}_1}z_2-1)^{2n}}\.
\ee
and 
\be\label{eq:coherentinner2}
	(\widetilde\cO(z_1),\widetilde\cO(z_2))_{\bar{\cD}_n}=\frac{1}{({\bar{z}_1}z_2-1)^{2n}}\.
\ee

It is easy to verify that $\|\cO(z)\|_{\cD_n}<\oo$. The action of $\PSL(2,\R)=\mathrm{PSU}(1,1)$ on $\cO(z)$ is easy to work out from \eqref{eq:Dbaction}
\be\label{eq:coherentaction}
u\.\cO(z) = (-\bar\beta z+\a)^{-2n}\cO(u\.z).
\ee
The infinitesimal form of this action is
\be
\begin{aligned}
L_{-1}\cdot (\cO(z)) &= \partial_{z} \cO(z)\\
L_{0}\cdot (\cO(z)) &= (z\partial_{z}+n) \cO(z)\\
L_{1}\cdot (\cO(z)) &= (z^2\partial_{z}+2n z) \cO(z)\,.
\end{aligned}
\label{eq:coherentactionLs}
\ee

For a unitary representation $V$ of $G$ we write $V^\oo$ for its Frech\'et space of smooth vectors~\cite{Warner}. We have the following simple result, the holomorphicity part of which means that all objects constructed out of $\cO(z)$ by continuous operations are holomorphic in $z$.
\begin{proposition}\label{prop:coherent}
	The coherent states $\cO$ and $\widetilde\cO$ take values in $\cD^\oo_n$, $\bar\cD^\oo_n$. They are holomorphic as functions $\D\to \cD^\oo_n$ and $\D'\to \bar\cD^\oo_n$.
	The span of the coherent states $\cO(z)\in \cD_n$ for $z\in\D$ is dense in $\cD_n$. The span of the coherent states $\widetilde\cO(z)\in\bar\cD_n$ for $z\in\D'$ is dense in $\bar\cD_n$.
\end{proposition}
\begin{proof}
	We prove the density statement for $\bar\cD_n$, and the statement for $\cD_n$ follows by complex conjugation. Let $K_w(z)\equiv\tfrac{2n-1}{\pi}(1- \bar{w} z)^{-2n}$. Note that for $w\in\D$ we have $K_w\in \bar\cD_n$. The function $K_w$ is the reproducing kernel for $\bar\cD_n$, i.e.\ for $w\in \D$ and $f\in \bar\cD_n$ we have
	\be\label{eq:reproduceid}
	(K_w,f)=f(w),
	\ee
	see~\cite{Hedenmalm} Corollary 1.5. Clearly, 
	\be
	K_w=\sqrt{\tfrac{2n-1}{\pi}} (\bar{w})^{-2n}\widetilde\cO((\bar{w})^{-1}).
	\ee
	Suppose now that $\overline{\mathrm{Span}\{\widetilde\cO(w)|w\in \D\}}\neq \bar\cD_n$. Then there is a non-zero $f\in\bar\cD_n$ such that $(\widetilde\cO(z),f)=0$ for all $z\in\D$. The above discussion then clearly implies that $f=0$, which is a contradiction. This proves denseness of the span.
	
	Now we would like to prove that for any $z\in \D$, $\cO(z)\in \cD_n^\oo$. This is equivalent to the function $u\mapsto u\. \cO(z)$ being a smooth function of $u$. Using~\eqref{eq:coherentaction}, this is equivalent to $\cO(z)$ being a smooth function for $z\in \D$. The above relation to the reproducing kernel implies that $(f,\cO(z))$ is a holomorphic function of $z\in \D$ for any $f\in\cD_n$, which in particular means that the map $z\mapsto \cO(z)$ is a weakly smooth function. But any weakly smooth function is smooth~\cite{Neto}. Finally, we have $(f,\ptl_{\bar z}\cO(z))=\ptl_{\bar z}(f,\cO(z))=0$ for any $f\in \cD_n$, and so $\ptl_{\bar z}\cO(z)=0$, i.e. $\cO(z)$ satisfies the Cauchy-Riemann equations.
	
	It remains to show that $\cO:\D\to \cD_n^\oo$ is a smooth function (the derivatives in $\cD_n^\oo$ and $\cD_n$ then automatically coincide). This is not immediate because $\cD_n^\oo$ has a stronger topology than $\cD_n$. However, using~\eqref{eq:coherentaction}, this is equivalent to $u\mapsto u\.\cO(0)$ being smooth in $\cD_n^\oo$, which is true due to the standard fact that $(\cD_n^\oo)^\oo=\cD_n^\oo$.
	
	The analogous statements for $\bar\cD_n$ follow by complex conjugation.
\end{proof}

We are ready to make contact with the definition of the coherent states in Section \ref{ssec:coherentStates}. Let us consider the subspaces $\sD_n$  of $L^2(\Gamma\backslash G)$ which appear in the decomposition~\eqref{eq:SL2RdecompositionFine}. For all $n$ that are present, we choose a \textit{unitary} isomorphism of unitary representations $\tau_n:\C^{\ell_n}\otimes \cD_n\to \mathscr{D}_n$. In physics literature, this is referred to as ``choosing a basis of local operators.'' Given $\tau_n$, we define the isomorphism $\bar\tau_n:\C^{\ell_n}\otimes \bar\cD_n\to \bar{\mathscr{D}}_n$ by
\be
	\bar\tau_n(v\otimes f)=\overline{\tau_n(\bar v\otimes \tl f)},
\ee
where
\be
	\tl f(z)=z^{-2n}\overline{f((\bar{z})^{-1})}.
\ee
We can now define the coherent states inside $\sD_n$, $\bar{\sD}_n$ as follows
\be
\begin{aligned}
	\mathscr{O}_{n,a}(z)&\equiv \tau_n(e_a\otimes \cO(z))\in \mathscr{D}_n^\oo\subseteq C^\oo(\G\@G),\\
	\widetilde{\mathscr{O}}_{n,a}(z)&\equiv \bar\tau_n(e_a\otimes \widetilde\cO(z))\in \bar{\mathscr{D}}_n^\oo\subseteq C^\oo(\G\@G)\,,
\end{aligned}
\label{eq:coherentDefinition}
\ee
where $a\in\{1,\cdots,\ell_n\}$ and $e_a$ is the standard basis of $\C^{\ell_n}$.
A consequence of these definitions is that $\mathscr{O}_{n,a}(z),\widetilde{\mathscr{O}}_{n,a}(z)$ transform as coherent states in~\eqref{eq:coherentaction} and \eqref{eq:coherentactionLs} and satisfy
\be\label{eq:coherentconj}
	\overline{\p{\mathscr{O}_{n,a}(z)}}=(\bar{z})^{-2n}\widetilde{\mathscr{O}}_{n,a}((\bar{z})^{-1}).
\ee
Indeed, 
\be
	\overline{\p{\mathscr{O}_{n,a}(z)}}&=\overline{\bar\tau_n(e_a\otimes \widetilde\cO(z))}=\tau_n(e_a\otimes \tl {\widetilde\cO(z)})\nn\\
	&=\tau_n(e_a\otimes z^{-2n}\cO((\bar{z})^{-1}))=(\bar{z})^{-2n}\widetilde{\mathscr{O}}_{n,a}((\bar{z})^{-1}),
\ee
where we used~\eqref{eq:coherentDefinition} and~\eqref{eq:coherentinversion}.
The proof that \eqref{eq:coherentDefinition} agrees with Definition~\ref{def:coherent} is left as an exercise to the reader.

\subsection{Coherent states in continuous series}
\label{sec:continuouscoherent}

For the principal and complementary series representations, which we denote by $R=\cP^+_{s}$ or $R=\cC_s$, we define $\mathrm{O}$ as a $R^\oo$-valued distribution on $C^\oo(\ptl\D)$ by
\be
	\mathrm{O}(f)=N_s f
\ee
for any test function $f\in C^\oo(\ptl\D)$, where $N_s>0$ is defined by $N_s\|1\|_R=1$ with $1$ being the constant function. The fact that $\mathrm{O}(f)$ is well-defined and continuous follows from the identification $R^\oo=C^\oo(\ptl\D)$.${}^{\ref{foot:smoothvectors}}$ As is customary when working with distributions, we will abuse the notation by introducing $\mathrm{O}(z)$, formally depending on $z\in\ptl\D$, and write
\be
	\int_{|z|=1}|dz|\mathrm{O}(z)f(z)\equiv \mathrm{O}(f).
\ee
With this notation, the condition $N_s\|1\|_R=1$ implies
\be
	\big\|\int_{|z|=1}|dz|\mathrm{O}(z)\big\|^2=1.
\ee
The naturally-defined action of $G$ is
\be
	(u\.\mathrm{O})(z) \equiv |\bar{\b} z+\bar{\a}|^{2-2\De} \mathrm{O}(u\.z),
\ee
where $\De=\thalf+s$.

Now, for each $\sC_{\l_k}$ we choose an unitary isomorphism $\kappa_k:\C^{d_k}\otimes R\to \sC_{\l_k}$ (where $R$ is the appropriate irrep for the value of $\l_k$) which also preserves the complex conjugation (which in the case of $\sC_{\l_k}$ is defined by point-wise complex conjugation on $G\@\G$). We define
\be
	\mathbb{O}_{k,a}(f)\equiv \kappa_k(e_a\otimes \mathrm{O}(f))\in \sC_{\l_k}^\oo.
\ee

\subsection{General properties of $L^2(\G\@G)$ and product expansions}

In this subsection we collect some useful technical results that will be needed later in the paper.

For a continuous representation $V$ of $G$ let $V^\oo$ denote the smooth representation of $G$ on the smooth vectors of $V$. The following proposition is a direct consequence of Corollary III.7.9 in \cite{BorelWallach}.
\begin{proposition} $(L^2(\G\@G))^\oo=C^\oo(\G\@G)$ as smooth representations of $G$.\label{prop:smooth}
\end{proposition}

Now let $U,V$, $U\subseteq V$ be continuous representations of $G$. A simple exercise verifies that $U^\oo=U\cap V^\oo$. For the decomposition~\eqref{eq:SL2RdecompositionFine} this and Proposition~\ref{prop:smooth} imply that the notation $\sD_n^\oo$, etc., that we used to denote the space of smooth functions inside $\sD_n$, etc., is consistent with the notation smooth representations introduced above.

Using the fact that the topology on $C^\oo(\G\@G)$ is that of uniform convergence of all derivatives, we obtain the following proposition.
\begin{proposition}
	The $n$-function correlator $\<\cdots\>$ is a well-defined continuous multilinear functional on the smooth vectors of the representations appearing in~\eqref{eq:SL2RdecompositionFine}.\label{prop:continuous}
\end{proposition}

In the rest of this subsection, let $H_1$, $H_2$, and $H_3$ be any three direct summands in~\eqref{eq:SL2RdecompositionFine}. Let $P_3$ be the orthonormal projector $P_3:L^2(\G\@G)\to H_3$. Note that the following map $C:H_1^\oo\times H_2^\oo\to H_3$
\be
	C(F_1,F_2)\equiv P_3(F_1F_2)
\ee
is well-defined. Furthermore, it is continuous (in the product topology, since the $P_3$ is continuous and the convergence in $H_2^\oo$ and $H_3^\oo$ implies uniform convergence), and $G$-equivariant. 
This means that we can view $C$ as an element
\be
	C\in \Hom_G(H_1^\oo\potimes H_2^\oo,H_3),
\ee
where $\potimes$ denotes the projective tensor product. We then have
\be
	\Hom_G(H_1^\oo\potimes H_2^\oo,H_3)=\Hom_G(H_1^\oo\potimes H_2^\oo,H_3^\oo),
\ee
which is due to $H_1^\oo\potimes H_2^\oo$ being a smooth representation of $G$ (\cite{Warner}, 4.4.1.10).

This implies that every term in~\eqref{eq:productExpansion} is in $C^\oo(\G\@ G)$. Furthermore, the following lemma holds.
\begin{lemma}
	The expansion~\eqref{eq:productExpansion} converges in $(L^2(\G\@G))^\oo=C^\oo(\G\@ G)$.
\end{lemma}
\begin{proof}
	It is convenient to view the topology on $(L^2(\G\@G))^\oo$ as defined by seminorms labeled by $D\in \cU(\mathfrak{g})$ and compact subsets of $G$~\cite{Warner}. Continuity of $P_3$ implies $DP_3(F_1F_2)=P_3(D(F_1F_2))$, which implies that~\eqref{eq:productExpansion} still converges in $L^2(\G\@G)$ after acting with $D$ (it becomes just~\eqref{eq:SL2RdecompositionFine} applied to $D(F_1F_2)$),
	\be
		D(F_1F_2) = DP_{\C}(F_1F_2) + \sum\limits_{n=0}^{\infty}[DP_{\sD_n}(F_1F_2)+DP_{\bar\sD_n}(F_1F_2)] + \sum\limits_{k=1}^{\infty}DP_{\sC_{\l_k}}\!(F_1F_2)\,.\label{eq:Dsum}
	\ee
	On the other hand, since~\eqref{eq:Dsum} comes from~\eqref{eq:SL2RdecompositionFine}, it follows that the $L^2$ norms of its remainders are $G$-invariant, and so after acting with $g\in G$ on~\eqref{eq:Dsum}, the result converges uniformly on $G$. This shows that~\eqref{eq:productExpansion} converges with respect to any seminorm on $(L^2(\G\@G))^\oo$.
\end{proof}
This lemma allows an almost unrestricted use of the expansion~\eqref{eq:productExpansion} in function correlators.

Finally, we will need the following proposition.
\begin{proposition}\label{prop:delta}
	Let $R$ be $\cD_n$, $\bar\cD_n$, $\cP_s^+$ or $\cC_s$, and $Y$ be, respectively, $\D',\D,\ptl\D$ or $\ptl\D$. Let $z\in Y$ and the delta functional $\de_z:R^\oo\to \C$ be given by $\de_z(f)=f(z)$, where $f\in R^\oo$ is interpreted as a function in the construction of $R$ from section~\ref{sec:reptheory}. Then $\de_z$ is a well-defined continuous linear functional.
\end{proposition}
\begin{proof}
For all these representations, $R^\oo$ contains only smooth (or even holomorphic) functions, and evaluation at a point is a continuous functional (delta functional) on $R^\oo$. Indeed, $\cD_n$ and $\bar\cD_n$ consist of holomorphic functions and the delta functional is given by inner product with the reproducing kernel, and so is well-defined and continuous even without restricting to smooth vectors. For principal and complementary series $(\cP^+_s)^\oo=\cC_s^\oo=C^\oo(\ptl\D)$, with the standard topology of uniform convergence of all derivatives.\footnote{For principal series, see~\cite{BorelWallach} Corollary III.7.9. For complementary series, we couldn't find a proof in the literature, and so we sketch an elementary argument in appendix~\ref{app:complementarysmooth}.\label{foot:smoothvectors}} Again, the delta-functional is well-defined and continuous.
\end{proof}

\subsection{Correlators of coherent states}

In this subsection we discuss the properties of the correlators of the coherent states in $\sD_n$ and $\bar\sD_n$. First of all, we note that Proposition~\ref{prop:continuous}, Proposition~\ref{prop:coherent}, and the continuity of $\tau_n,\bar\tau_n$ imply that all the correlators below are holomorphic in the arguments $z_i$ of the coherent states when $z_i$ are in their domains.

\subsubsection*{Two-function correlators}

Equations~\eqref{eq:coherentDefinition},~\eqref{eq:coherentconj}, unitarity of $\tau_p$ and orthogonality of the decomposition~\eqref{eq:SL2RdecompositionFine} together with~\eqref{eq:coherentinner} imply
\be\label{eq:twopointnormalized}
\<\mathscr{O}_{n,i}(z_1)\widetilde{\mathscr{O}}_{m,j}(z_2)\>&=z_2^{-2n}(\sO_{m,j}((\bar{z}_2)^{-1}),\sO_{n,i}(z_1))=z_2^{-2n}\de_{m,n}\de_{i,j}(\cO((\bar{z}_2)^{-1}),\cO(z_1))\nn\\
&=\frac{\de_{m,n}\de_{i,j}}{(z_1-z_2)^{2n}}.
\ee

\subsubsection*{Three-function correlators}

$G$-invariance of the  three-functions correlators implies that
\be
	\<(u\.\sO_{k,a}(z_1))(u\.\sO_{l,b}(z_2))(u\.\widetilde\sO_{m,c}(z_3))\>=\<\sO_{k,a}(z_1)\sO_{l,b}(z_2)\widetilde\sO_{m,c}(z_3)\>
\ee
for any $u\in\mathrm{PSU}(1,1)$. Evaluating the left-hand side using~\eqref{eq:coherentaction} we obtain a non-trivial constraint on $\<\sO_{k,a}(z_1)\sO_{l,b}(z_2)\widetilde\sO_{m,c}(z_3)\>$ as a function of $z_1,z_2,z_3$,
\be
&\<\sO_{k,a}(z_1)\sO_{l,b}(z_2)\widetilde\sO_{m,c}(z_3)\>\nn\\
&=(\bar{\b}z_1+\bar{\a})^{-2k}(\bar{\b}z_2+\bar{\a})^{-2l}(\bar{\b}z_3+\bar{\a})^{-2m}\<\sO_{k,a}(u\.z_1)\sO_{l,b}(u\.z_2)\widetilde\sO_{m,c}(u\.z_3)\>,\label{eq:3ptPSU}
\ee
where $\a$ and $\b$ are defined via~\eqref{eq:uparameters}.
A standard observation is that
\be
	u\.z_i-u\.z_j=(\bar{\b}z_i+\bar{\a})^{-1}(\bar{\b}z_j+\bar{\a})^{-1}(z_i-z_j),
\ee
which shows that the following expression satisfies~\eqref{eq:3ptPSU}
 \be\label{eq:threeptcoherentgeneral}
\<\sO_{k,a}(z_1)\sO_{l,b}(z_2)\widetilde\sO_{m,c}(z_3)\>=\frac{f_{klm;abc}}{z_{12}^{k+l-m}z_{13}^{k+m-l}z_{23}^{l+m-k}},
\ee
where $f_{klm;abc}\in \C$ is a constant. In fact, this is the unique form which is both holomorphic in $z_1,z_2,z_3$ and satisfies~\eqref{eq:3ptPSU}. 

To see this, note that holomorphicity allows one to analytically continue~\eqref{eq:3ptPSU} to be valid for $u\in\PSL(2,\C)$ as long as $u\.z_i$ stay in their domains. Given some values of $z_1,z_2,z_3$, one can find $u\in \PSL(2,\C)$ such that $u\.z_1=0,u\.z_2=\half,u\.z_3=2$. Then~\eqref{eq:3ptPSU} expresses $\<\sO_{k,a}(z_1)\sO_{l,b}(z_2)\widetilde\sO_{m,c}(z_3)\>$ in terms of $\<\sO_{k,a}(0)\sO_{l,b}(\half)\widetilde\sO_{m,c}(2)\>$, which proves that there is indeed a unique solution, up to an overall constant.

Finally, we note  holomorphicity also implies that $f_{klm;abc}$ can only be non-zero for $k+l-m\leq 0$ since otherwise the limit $z_1\to z_2$ of~\eqref{eq:threeptcoherentgeneral} is singular, while the correlator has to be analytic for $z_1,z_2\in \D$.

\subsubsection*{Four-function correlator}

We now consider the four-function correlator
\be
\<\sO_{k,a}(z_1)\sO_{l,b}(z_2)\widetilde\sO_{m,c}(z_3)\widetilde\sO_{n,d}(z_4)\>.
\ee
$G$-invariance and holomorphicity imply the following proposition.
\begin{proposition}\label{prop:fourptred}
	\be\label{eq:fourpt}
	\<\sO_{k,a}(z_1)\sO_{l,b}(z_2)\widetilde\sO_{m,c}(z_3)\widetilde\sO_{n,d}(z_4)\>=
	\frac{1}{z_{12}^{k+l}z_{34}^{m+n}}
	\p{\frac{z_{24}}{z_{14}}}^{k-l}
	\p{\frac{z_{14}}{z_{13}}}^{m-n}g(z),
	\ee
	where $g:\C\@[1,+\oo)\to\C$ is a holomorphic function, $z_{ij}=z_i-z_j$, and $z$ is the cross-ratio
	\be\label{eq:crossratio}
	z=\frac{z_{12}z_{34}}{z_{13}z_{24}}.
	\ee
\end{proposition}
\begin{proof}
	The argument is similar to the one for three-function correlators given above, so we will be schematic. Holomorphicity of $\sO$ and $\widetilde\sO$ allows us to extend $\mathrm{PSU}(1,1)\simeq \PSL(2,\R)$ invariance to $\PSL(2,\C)$ invariance as long as $z_i$ remain in their domains under the transformation. But then this ``local'' $\PSL(2,\C)$ invariance lets us analytically continue the correlator to all configurations of $z_i\in \C\cup\{\oo\}$ in which there is a circle separating $z_1,z_2$ from $z_3,z_4$. After this continuation the correlator is fully $\PSL(2,\C)$-invariant. 
The function $g(z)$ can be obtained by setting $z_1=0,z_2=z,z_3=1,z_4=\oo$,\footnote{\label{footnote:infinity}By the definition of $\widetilde\cO(z)\in \cD_n$, setting $z\to\oo$ is to be understood as the limit $z^{2n}\widetilde\cO(z)$ as $z\to\oo$.} which by the previous condition is possible for $z\not\in[1,+\oo)$. Conversely, any allowed configuration of $z_i$ can be mapped by $\PSL(2,\C)$ to a configuration of such form, which establishes the result.
\end{proof}

\subsection{Product expansion of coherent states}

In this subsection, we will use the technology developed so far to constrain the product expansion of coherent states.

\subsubsection*{Expansion of $\sO_n(z_1)\sO_n(z_2)$} We will prove the following

\begin{lemma}\label{lemma:DDselection}
		We have $P_H(\sO_n(z_1)\sO_n(z_2))=0$ unless $H=\sD_p$ with $p$ even and $p\geq 2n$.
\end{lemma}
This lemma implies that the product $\sO_n(z_1)\sO_n(z_2)$ has the following expansion
\be\label{eq:preOPEDD}
	\sO_n(z_1)\sO_n(z_2)=\sum_{p=2n\atop p\text{ even}}^{\oo}P_{\sD_p}(\sO_n(z_1)\sO_n(z_2)).
\ee

Using $G$-invariance we can say more about the individual terms. In particular, we have
\begin{lemma}\label{lemma:DDOPE}
		This expansion takes the form
		\be
		\sO_n(z_1)\sO_n(z_2)=\sum_{p=2n\atop p\text{ even}}^{\oo}\sum_{a=1}^{\ell_p}f_{p,a}\tau_p(e_a\otimes C_{p}(z_1,z_2))
		\ee
		where $f_{p,a}$ are as in~\eqref{eq:threePointFunction} and $C_{p}(z_1,z_2)\in \cD_p$ is defined by
		\be
		C_{p}(z_1,z_2)(z)=\sqrt{\frac{2p-1}{\pi}}\frac{1}{z_{12}^{2n-p}z_{13}^pz_{23}^p}.
		\ee
		Equation~\eqref{eq:productOO} follows by expanding in small $z_1-z_2$ and comparing with the definition of the coherent state~\eqref{eq:coherent}.
	\end{lemma}
\begin{proof}[Proof of Lemma~\ref{lemma:DDselection}]
	Consider first $P_\C(\sO_n(z_1)\sO_n(z_2))$. This gives us a holomorphic $G$-invariant function of $z_1,z_2\in \D$. The only such function is (locally)
	\be
	P_\C(\sO_n(z_1)\sO_n(z_2))=A (z_1-z_2)^{-2n}
	\ee
	for some $A\in \C$. However, for $n>0$ this is singular at $z_1=z_2$ unless $A=0$. We conclude $P_\C(\sO_n(z_1)\sO_n(z_2))=0$.
	
	Consider now $P_{\sC_{\l_k}}(\sO_n(z_1)\sO_n(z_2))$ for some $k$, and let $\pi:\sC_{\l_k}\to \cP_{s}^+$ or $\pi:\sC_{\l_k}\to \cC_{s}$ (depending on $\l_k$) be any non-trivial homomorphism. By Proposition~\ref{prop:delta} and the discussion preceding it, 
	\be
	\de_{z_3}(\pi(P_{\sC_{\l_k}}(\sO_n(z_1)\sO_n(z_2))))
	\ee
	is a holomorphic function of $z_1,z_2\in \D$ and a smooth function of $z_3$ that is also $G$-invariant. The only possible (local) form is 
	\be
	\frac{A z_3^{\De_k}}{(z_1-z_2)^{2n-{\De_k}}(z_1-z_3)^{\De_k}(z_2-z_3)^{\De_k}}.
	\ee
	However, this is again singular at $z_1=z_2$ unless $A=0$. So $\pi(P_{\sC_{\l_k}}(\sO_n(z_1)\sO_n(z_2)))$ is zero for any $\pi$, and thus
	$P_{\sC_{\l_k}}(\sO_n(z_1)\sO_n(z_2))=0$.
	
	Consider $P_{\bar\sD_{p}}(\sO_n(z_1)\sO_n(z_2))$ and any $\pi:\bar\sD_{p}\to\bar\cD_n$.
	We now have to consider a holomorphic function of $z_1,z_2,z_3\in\D$. The only possibility is
	\be
	\de_{z_3}(\pi(P_{\bar\sD_n}(\sO_n(z_1)\sO_n(z_2))))=\frac{A}{(z_1-z_2)^{2n-p}(z_1-z_3)^{p}(z_2-z_3)^{p}}.
	\ee
	But this is singular at $z_1=z_3$ unless $A=0$. We conclude $P_{\bar\sD_{p}}(\sO_n(z_1)\sO_n(z_2))=0$.
	
	Consider finally $P_{\sD_{p}}(\sO_n(z_1)\sO_n(z_2))$ and any $\pi:\sD_{p}\to\cD_n$.
	We now have to consider a holomorphic function of $z_1,z_2\in\D$ and $z_3\in \D'$. The only possibility is
	\be
	\de_{z_3}(\pi(P_{\sD_n}(\sO_n(z_1)\sO_n(z_2))))=\frac{A}{(z_1-z_2)^{2n-p}(z_1-z_3)^{p}(z_2-z_3)^{p}}.
	\ee
	This has the required regularity and invariance properties provided $p\geq 2n$. However, for odd $p$ the right-hand side is odd under the permutation $z_1\leftrightarrow z_2$ (unless $A=0$), while the left-hand side is even.
	We conclude $P_{\sD_{p}}(\sO_n(z_1)\sO_n(z_2))=0$ unless $p\geq 2n$ and $p$ is even.
\end{proof}

\begin{proof}[Proof of Lemma~\ref{lemma:DDOPE}]
	Consider now a single term $\psi_p=P_{\sD_p}(\sO_n(z_1)\sO_n(z_2))\in \sD_p$ in~\eqref{eq:preOPEDD}. Since $\sD_p\simeq \C^{\ell_p}\otimes \cD_p$, there exist vectors $h_1,\cdots, h_{\ell_p}\in \cD_p$, uniquely determined, such that
	\be
	\psi_p=\tau_p\p{\sum_{a=1}^{\ell_p}e_a\otimes h_a}.
	\ee
 	Recall that $\cD_p$ is a reproducing kernel Hilbert space, and $\cO(z)$ is essentially the reproducing kernel:
	\be
	\p{(\bar{w})^{-2p}\cO((\bar{w})^{-1}),h_a}=\sqrt{\frac{\pi}{2p-1}}h_a(w).
	\ee
	Unitarity of $\tau_p$ and the definition~\eqref{eq:coherentDefinition} then imply
	\be
	h_a(w)=\sqrt{\frac{2p-1}{\pi}}\p{(\bar{w})^{-2p}\sO_{p,a}((\bar{w})^{-1}),\sO_n(z_1)\sO_n(z_2)}.
	\ee
	Using~\eqref{eq:coherentconj} we find
	\be
	h_a(z_3)=\sqrt{\frac{2p-1}{\pi}}\<\sO_n(z_1)\sO_n(z_2)\widetilde\sO_{p,a}(z_3)\>.
	\ee
	$G$-invariance and holomorphicity of coherent states imply that~\eqref{eq:threePointFunction} holds for some constants $f_{p,a}$, so we conclude
	\be
	\psi_p=\sum_{a=1}^{\ell_p}f_{p,a}\tau_p(e_a\otimes C_{p}(z_1,z_2)),
	\ee
	where $C_{p}(z_1,z_2)\in \cD_p$ is defined by
	\be\label{eq:cpExplicit}
	C_{p}(z_1,z_2)(z_3)=\sqrt{\frac{2p-1}{\pi}}\frac{1}{z_{12}^{2n-p}z_{13}^pz_{23}^p}.
	\ee
\end{proof}

\subsubsection*{Expansion of $\sO_n(z_1)\widetilde{\sO}_n(z_2)$} Our discussion of the product $\sO_n(z_1)\widetilde{\sO}_n(z_2)$ has the same structure as above. First, we will prove
\begin{lemma}\label{lemma:DDbarselection}
		We have $P_H(\sO_n(z_1)\widetilde\sO_n(z_2))=0$ unless $H=\C$ or $H=\sC_{\l_k}$.
	\end{lemma}
	The projection onto $H=\C$ is easy to compute. For this we just need the inner product
	\be
	(1,\sO_n(z_1)\widetilde\sO_n(z_2))=\<\sO_n(z_1)\widetilde\sO_n(z_2)\>=\frac{1}{(z_1-z_2)^{2n}}.
	\ee
	Thus we have
	\be\label{eq:presum}
	\sO_n(z_1)\widetilde\sO_n(z_2)=\frac{1}{(z_1-z_2)^{2n}}+\sum_k P_{\sC_{\l_k}}(\sO_n(z_1)\widetilde\sO_n(z_2)).
	\ee
	Above we have abused the notation slightly by writing $c$ for a constant function on $G$ identically equal to $c\in \C$.
	Similarly to Lemma~\ref{lemma:DDOPE}, we can say more about $P_{\sC_{\l_k}}(\sO_n(z_1)\widetilde\sO_n(z_2))$. First, recall that we have
	unitary isomorphisms $\kappa_k:\C^{d_k}\otimes R_k\to \sC_{\l_k}$ where $R_k=\cP_{s_k}^{+}$ or $R_k=\cC_{s_k}$ depending on whether
	$\l_k\geq\tfrac{1}{4}$ or $\l_k<\tfrac{1}{4}$, which are compatible with the standard complex conjugations defined on $\C^{d_k}\otimes R_k$ and $\sC_{\l_k}$. We have
	\begin{lemma}
		\label{lemma:opeDDbar}
		\be
		P_{\sC_{k}}(\sO_n(z_1)\widetilde\sO_n(z_2))=\sum_{a=1}^{d_k}c_{k,a}\kappa_k(e_a\otimes \tl C_{k}(z_1,z_2))
		\ee
		for the real $c_{k,a}$ defined in section~\ref{sec:bootstrapMethod}, and where $\tl C_{k}(z_1,z_2)\in R_k$ is defined by
		\be
		\tl C_{k}(z_1,z_2)(z_0)=\frac{N_{k}z_2^{-2n}}{(1-z_1z_2^{-1})^{2n-\De_k}(1-z_1z_0^{-1})^{\De_k}(1-z_2^{-1}z_0)^{\De_k}}.
		\ee
		Here $N_{k}\in \mathbb{C}$ is a constant such that the constant function equal to $N_{k}$ has unit norm in $R_k$ and is real under the standard complex conjugation in $R_k$. Comparing with the definition of $\mathbb{O}_{k,a}$ in Section~\ref{sec:continuouscoherent} results in the formula~\eqref{eq:sect2oob}.
	\end{lemma}
	
\begin{proof}[Proof of Lemma~\ref{lemma:DDbarselection}]
	The proof is completely analogous to the proof of Lemma~\ref{lemma:DDselection}, so we will be brief. We must have
	\be
	P_{\C}(\sO_n(z_1)\widetilde\sO_n(z_2))&=A(z_1-z_2)^{-2n},\label{eq:proof1}\\
	\de_{z_3}(\pi(P_{\sD_p}(\sO_n(z_1)\widetilde\sO_n(z_2))))&=\frac{A}{(z_1-z_2)^{2n-p}(z_1-z_3)^{p}(z_2-z_3)^{p}},\label{eq:proof2}\\
	\de_{z_3}(\pi(P_{\bar\sD_p}(\sO_n(z_1)\widetilde\sO_n(z_2))))&=\frac{A}{(z_1-z_2)^{2n-p}(z_1-z_3)^{p}(z_2-z_3)^{p}},\label{eq:proof3}\\
	\de_{z_3}(\pi(P_{\sC_{\l_k}}(\sO_n(z_1)\widetilde\sO_n(z_2))))&=\frac{A z_2^{-2n}}{(1-z_1z_2^{-1})^{2n-\De_k}(1-z_1z_3^{-1})^{\De_k}(1-z_2^{-1}z_3)^{\De_k}},\label{eq:proof4}
	\ee
	where $A$ denotes a generic constant (it can differ between these equations). We find that~\eqref{eq:proof2} fails to be holomorphic at $z_2=z_3$,~\eqref{eq:proof3} fails to be holomorphic at $z_1=z_3$, while~\eqref{eq:proof1} and~\eqref{eq:proof4} satisfy all the regularity conditions and cannot be excluded. This completes the proof.
\end{proof}

\begin{proof}[Proof of Lemma~\ref{lemma:opeDDbar}]
	Let $\psi_k=P_{\sC_{\l_k}}(\sO_n(z_1)\widetilde\sO_n(z_2))\in \sC_{\l_k}$ and let $R=\cP_s^+$ or $R=\cC_s$ depending on whether $\l_k\geq \tfrac{1}{4}$ or $\l_k<\tfrac{1}{4}$,  since $\lambda_k=1/4-s^2$ (see \eqref{eq:sNota}). Recall that we chose a unitary isomorphism $\kappa_k:\C^{d_k}\otimes R\to\sC_k$ which also preserves the natural complex conjugations defined on $\C^{d_k}\otimes R$ and $\sC_k$. We must have
	\be
	\psi_k=\sum_{a=1}^{d_k}\kappa_k(e_a\otimes f_a)
	\ee
	for some $f_a\in \cP_s^+$ or $f_a\in \cC_s$ depending on whether $\l_k\geq \tfrac{1}{4}$ or $\l_k<\tfrac{1}{4}$. 
	Sending $\psi_k$ to $f_a$ we obtain a projection map $\pi_a:\sC_k\to \cP_s^+$ or $\pi_a:\sC_k\to \cC_s$. As in the proof of Lemma~\ref{lemma:DDbarselection}, we must have
	\be
	f_a(z_0)=\de_{z_0}(\pi_a(P_{\sC_{\l_k}}(\sO_n(z_1)\widetilde\sO_n(z_2))))= 
	\frac{\tl c_{k,a} z_2^{-2n}}{(1-z_1z_2^{-1})^{2n-\De_k}(1-z_1z_0^{-1})^{\De_k}(1-z_2^{-1}z_0)^{\De_k}}
	\ee
	for some $\tl c_{k,a}\in \C$. Using~\eqref{eq:coherentconj} it is easy to check that $\sO_n(0)\tl\sO_n(\oo)$ is a real function. With $z_1=0$ and $z_2=\oo$,${}^{\ref{footnote:infinity}}$ $f_a(z_0)$ becomes $f_a(z_0)=\tl c_{k,a}$. This, together with the fact that $\kappa_k$ preserves complex conjugation, implies that we can write $\tl c_{k,a}=c_{k,a}N_k$ with $c_{k,a}\in\R$ and $N_k\in \C$ characterized in the statement of Lemma~\ref{lemma:opeDDbar} (which determines $N_k$ up to a sign, which we chose once in an arbitrary way for every $R$). So, we conclude that
	\be\label{eq:proof:OObarOPE}
	P_{\sC_{k}}(\sO_n(z_1)\widetilde\sO_n(z_2))=\sum_{a=1}^{d_k}c_{k,a}\kappa_k(e_a\otimes \tl C_{k}(z_1,z_2))
	\ee
	where $\tl C_{k}(z_1,z_2)\in R$ is defined by
	\be
	\tl C_{k}(z_1,z_2)(z_0)=\frac{N_k z_2^{-2n}}{(1-z_1z_2^{-1})^{2n-\De_k}(1-z_1z_0^{-1})^{\De_k}(1-z_2^{-1}z_0)^{\De_k}}.
	\ee
\end{proof}

\subsection{Proof of Theorem~\ref{theorem:main}}

\begin{proof}[Proof of Theorem~\ref{theorem:main}, part (1)]
Using~\eqref{eq:coherentconj}, we have
	\be
	\overline{\p{\widetilde\sO_n(z_3)\widetilde\sO_n(z_4)}}=(\bar{z_{3}})^{-2n}(\bar{z_4})^{-2n}\sO_n((\bar{z_3})^{-1})\sO_n((\bar{z_4})^{-1}),
	\ee
	to which we can apply the expansion from Lemma~\ref{lemma:DDOPE}. Applying the same lemma to $\sO_n(z_1)\sO_n(z_2)$ we find
	\be\label{eq:proofCBDD}
	\<\sO_n(z_1)\sO_n(z_2)\widetilde\sO_n(z_3)\widetilde\sO_n(z_4)\>&=\p{\bar{z_{3}}^{-2n}\bar{z_4}^{-2n}\sO_n(\bar{z_3}^{-1})\sO_n(\bar{z_4}^{-1}),\sO_n(z_1)\sO_n(z_2)}\nn\\
	&=\sum_{p=2n\atop p\text{ even}}^{\oo}\sum_{a=1}^{\ell_p}|f_{p,a}|^2 z_{12}^{-2n}z_{34}^{-2n}\cG_p(z),
	\ee
	where $z$ is the cross-ratio~\eqref{eq:crossratio}  and $\cG_p$ is defined by
	\be
	z_{12}^{-2n}z_{34}^{-2n}\cG_p(z)=z_{3}^{-2n}z_4^{-2n}\p{C_{p}(\bar{z_3}^{-1},\bar{z_4}^{-1}),C_{p}(z_1,z_2)}.
	\ee
	The following lemma then brings the expansion~\eqref{eq:proofCBDD} to the form stated in (1) of Theorem~\ref{theorem:main},
	\begin{lemma}\label{lemma:toproveGp}
		$\cG_\De(z)=z^\De{}_2F_1(\De,\De;2\De;z)$, where ${}_2F_1$ is the Gauss hypergeometric function.
	\end{lemma}
	
	All the infinite sums of vectors leading to~\eqref{eq:proofCBDD} converged with respect to the Hilbert space norms. Therefore, the sum~\eqref{eq:proofCBDD} converges pointwise for $z\in \C\setminus[1,+\oo)$. To show that the convergence in~\eqref{eq:schannel} is uniform on compact subsets in $z$, it suffices to show that~\eqref{eq:proofCBDD} converges uniformly on compact subsets in $z_1,\cdots, z_4$. For the latter, note that the situation is essentially that of the sum
	\be
		(v(x),u(x))=\sum_{k=1}^\oo (P_kv(x), P_ku(x)),
	\ee
	where $v(x),u(x)$ are continuous families of vectors in a Hilbert space, parameterized by $x$ varying in a compact space $Q$, and $P_k$ are orthogonal projection operators with $\sum_{k=1}^\oo P_k=\mathrm{id}$ and with $P_kP_l=0$ for $k\neq l$. The above sum converges pointwise for $x\in Q$. To see that the convergence is uniform we first estimate the remainders as
	\be
		|\sum_{k=N}^\oo (P_kv(x), P_ku(x))|=|(\sum_{k=N}^\oo P_kv(x), u(x))|\leq \|\sum_{k=N}^\oo P_kv(x)\| \|u(x)\|,
	\ee
	and note that the functions $a_N(x)=\|\sum_{k=N}^\oo P_kv(x)\|=\sqrt{\sum_{k=N}^\oo \|P_kv(x)\|^2}\geq 0$ are continuous on $Q$ and, as $N\to\oo$, monotonically converge to 0 pointwise in $x$. Dini's theorem then implies that $a_N(x)$ converge to 0 uniformly in $x$. Since $\|u(x)\|$ is continuous, and hence bounded, this implies the uniform convergence of the sum.

	 To show the convergence for the derivatives of~\eqref{eq:proofCBDD}, we apply the same arguments to the products like $\ptl^a_1\sO_n(z_1)\ptl^b_2\sO_n(z_2)$ and use the fact that continuity of the projectors $P_R$ implies $\ptl^a_1\ptl^b_2 P_R(\sO_n(z_1)\sO_n(z_2))=P_R(\ptl^a_1\sO_n(z_1)\ptl^b_2\sO_n(z_2))$.
\end{proof}

\begin{proof}[Proof of Theorem~\ref{theorem:main}, part (2)]
Similarly to part (1), using~\eqref{eq:presum} and Lemma~\ref{lemma:opeDDbar} twice, we find
	\be
	&\<\sO_n(z_1)\sO_n(z_2)\widetilde\sO_n(z_3)\widetilde\sO_n(z_4)\>\nn\\
	&=z_{14}^{-2n}z_{23}^{-2n}+\sum_{k=1}^\oo\sum_{a=1}^{d_k}c_{k,a}^2
	z_{2}^{-2n}z_3^{-2n}(\tl C_{k}(\bar{z_3}^{-1},\bar{z_2}^{-1}),\tl C_{k}(z_1,z_4))\nn\\
	&=z_{14}^{-2n}z_{23}^{-2n}+\sum_{k=1}^\oo\sum_{a=1}^{d_k}c_{k,a}^2
	z_{12}^{-2n}z_{34}^{-2n}\p{\frac{z}{1-z}}^{2n}\cH_{\De_k}(z),
	\ee
	where $\cH_{\De_k}(z)$ is defined by
	\be
	z_{12}^{-2n}z_{34}^{-2n}\p{\frac{z}{1-z}}^{2n}\cH_{\De_k}(z)=z_{2}^{-2n}z_3^{-2n}(\tl C_{k}(\bar{z_3}^{-1},\bar{z_2}^{-1}),\tl C_{k}(z_1,z_4)).
	\ee
	We obtain the expansion of the required form by applying the following lemma.
	\begin{lemma}\label{lemma:toproveH}
			$\cH_{\De_k}(z)= {}_2F_1(\De_k,1-\De_k;1;\tfrac{z}{z-1}).$
	\end{lemma}
	This leads to the expansion~\eqref{eq:tchannel}. As in the case of the $s$-channel expansion, it converges pointwise for all $z\in \C\setminus[1,+\oo)$. The same argument as in the proof of part (1) shows that its derivatives also converge pointwise.	
\end{proof}

It remains to prove Lemma~\ref{lemma:toproveGp} and  Lemma~\ref{lemma:toproveH}. We start with the former.

\begin{proof}[Proof of Lemma~\ref{lemma:toproveGp}]
	We have to compute the inner product
	\be\label{eq:innerprodGp}
	z_{12}^{-2n}z_{34}^{-2n}\cG_p(z)=z_{3}^{-2n}z_4^{-2n}\p{C_{p}(\bar{z_3}^{-1},\bar{z_4}^{-1}),C_{p}(z_1,z_2)}.
	\ee
To do that, we use the observation from~\cite{DO2} that such inner products satisfy a second-order differential equation, coming from the action of the quadratic Casimir element $c_2$, defined in~\eqref{eq:casimir}. Recall that $C_p$ represents a $G$-invariant map $\cD^{\infty}_n\times\cD^{\infty}_n\rightarrow \cD_{p}$, and that $c_2$ acts on $\cD_{p}^{\infty}$ by multiplication by the constant $p(p-1)$. It follows that if we precompose $C_p$ by $c_2$, we get $p(p-1) C_p$. Now, the generators of $\mathfrak{sl}_2(\mathbb{R})$ act on the vectors $\cO(z)$ through the differential operators~\eqref{eq:coherentactionLs}. It follows that $C_{p}(z_1,z_2)$ satisfies the following partial differential equation, which can also be easily verified using the explicit form~\eqref{eq:cpExplicit}:
\be
\begin{aligned}
&\left[
(L^{(1)}_0+L^{(2)}_0)^2 -\frac{(L^{(1)}_{-1}+L^{(2)}_{-1})(L^{(1)}_{1}+L^{(2)}_{1}) + (L^{(1)}_{1}+L^{(2)}_{1})(L^{(1)}_{-1}+L^{(2)}_{-1})}{2}
\right]C_{p}(z_1,z_2) \\
&= p(p-1)C_{p}(z_1,z_2)\,.
\end{aligned}
\ee
The superscript $(1),\,(2)$ on the Lie algebra generators denotes whether the derivatives are taken with respect to $z_1$ or $z_2$. The explicit form of this differential equation is
\be
\left[-(z_1-z_2)^2 \partial_{z_1}\partial_{z_2}+2 n (z_1-z_2) (\partial_{z_1}-\partial_{z_2})+2 n (2 n-1)-p(p-1)
\right] C_{p}(z_1,z_2)
= 0\,.
\ee
It follows that both sides of equation~\eqref{eq:innerprodGp} satisfy the same differential equation. Since the left-hand side depends on $z_1$ and $z_2$ through the single variable $z = z_{12}z_{34}/(z_{13}z_{24})$, this becomes the following ordinary differential equation for $\cG_p(z)$
\be
		z^2(1-z)\ptl_z^2\cG_p(z)-z^2\ptl_z\cG_p(z)=p(p-1)\cG_p(z).
\ee

	For $p\geq 2n>1$ the solutions of this equation have the form\footnote{The function ${}_2F_1(1-p,1-p;1-2p;z)$ is well-defined even though $1-2p$ is a negative integer because the hypergeometric series truncates before this can pose a problem.}
	\be
		\cG_p(z)=Az^p{}_2F_1(p,p;2p;z)+Bz^{1-p}{}_2F_1(1-p,1-p;1-2p;z).
	\ee
	We know that $\cG_p(z)$ has to be holomorphic at $z=0$, which implies $B=0$. To compute $A$, we set $z_1=z_2=0$ and $z_3=z_4=\oo$. This gives
	\be
		\cG_p(z)\sim z^p\|v\|^2_{\cD_n}=z^p,
	\ee
	where $v(z)=\sqrt{\frac{2n-1}{p}}z^{-2p}$. This implies $A=1$ and completes the proof of the lemma.
\end{proof}

\begin{proof}[Proof of Lemma~\ref{lemma:toproveH}]
	We have to compute the inner product
	\be
	z_{12}^{-2n}z_{34}^{-2n}\p{\frac{z}{1-z}}^{2n}\cH_{\De_k}(z)=z_{2}^{-2n}z_3^{-2n}(\tl C_{k}(\bar{z_3}^{-1},\bar{z_2}^{-1}),\tl C_{k}(z_1,z_4)).
	\ee
	We use the same strategy as in the proof of Lemma~\ref{lemma:toproveGp} to show that the Casimir equation and regularity conditions imply
	\be
		\cH_{\De_k}(z)= A{}_2F_1(\De,1-\De;1;\tfrac{z}{z-1}).
	\ee
	It only remains to fix the normalization. We again set $z_1,z_2=0$ and $z_3,z_4=\oo$, which implies
	\be
		\cH_{\De_k}(0)=(N_{k},N_{k})=1.
	\ee
	This shows that $A=1$ and completes the proof of the lemma.
\end{proof}

\subsection{Derivation of the bounds}
\label{sec:derivation}

We now explain how to extract bounds on the eigenvalues $\l_k$ from the crossing equation~\eqref{eq:crossineq}. The resulting bounds themselves will be summarized in section~\ref{sec:bounds}. While the following material is mostly standard for a conformal bootstrap expert, there are a few important distinctions from the usual story.

First of all, we rewrite~\eqref{eq:crossineq} in the following form,\footnote{The fact that odd $p$ do not appear follows from the exchange symmetry between $\sO_n(z_1)$ and $\sO_n(z_2)$. Taking it into account is equivalent to imposing the third crossing equation which equates $s$- and $t$-channels to $u$-channel, where $u$-channel is defined by taking the product expansions in pairs $\sO_n(z_1)\widetilde\sO_n(z_3)$ and $\sO_n(z_2)\widetilde\sO_n(z_4)$. This is the reason why we are only equating two channels.}
\be\label{eq:crossineq2}
\sum_{p=2n\atop p\text{ even}}^{\oo}S_p \cG_p(z)=\p{\frac{z}{1-z}}^{2n}\p{1+\sum_{k=1}^{\oo}T_k \cH_{\De_k}(z)},
\ee
where $S_p=\sum_{a=1}^{\ell_p}|f_{p,a}|^2\geq 0$ and $T_k=\sum_{a=1}^{d_k}c_{k,a}^2\geq 0$. According to theorem~\ref{theorem:main}, both sides converge for $z\in \C\setminus[1,+\oo)$ (uniformly on compact subsets).

The following comment is intended for conformal bootstrap experts and is not essential for what follows. The left-hand side is exactly the standard expansion in $s$-channel four-point $d=1$ conformal blocks, and it converges where one would expect. What is not usual is that the right-hand side, the $t$-channel, converges nicely in the $s$-channel OPE limit $z\to 0$. The blocks $\cH_{\De_k}(z)$ appearing in the right-hand side are just the sums of the usual $t$-channel blocks and their shadows, which makes this even stranger. The reason for this behavior is that $\De_k$ go to infinity in the imaginary direction, so the standard factor\footnote{Here $\r_t$ is the radial coordinate of~\cite{Hogervorst:2013sma} in $t$-channel, i.e. $\r_t=\frac{1-z}{(1+\sqrt z)^2}$.} $\rho_t^{\De_k}$ penalizes not for large $|\rho_t|$, but for large $|\arg \rho_t|$, which grows towards the cut $z\in [1,+\oo)$. The fact that $\De_k$ are not generally real also means that the usual technique of expanding around $z=\thalf$ will not work because the factor $\rho_t^{\De_k}$ is not positive anymore. We will use a different strategy.

The uniform convergence of the sums in~\eqref{eq:crossineq2} allows us to solve for the coefficients $S_p$ by applying the orthogonality relation~\eqref{eq:orthogonalityG}. The result is the following set of identities
\be
S_p=\cF_{0;p}^n(0)+\sum_{k=1}^\oo T_k \cF_{0;p}^n(\De_k)\,.
\ee
Here
\be
\cF_{0;p}^n(\Delta) = \frac{1}{2\pi i}\oint dz z^{-2}\cG_{1-p}(z)\left(\tfrac{z}{1-z}\right)^{2n}\cH_{\Delta}(z)\,.
\ee

From the explicit expressions~\eqref{eq:sblock} and~\eqref{eq:tblock} for $\cG_p$ and $\cH_\De$ it follows that $\cF_{0;p}^n(\De)$ is a polynomial in $\De$. Furthermore, $\cH_{\De}=\cH_{1-\De}$ implies that $\cF_{0;p}^n(\De)=\cF_{0;p}^n(1-\De)$, and so we can write $\cF_{0;p}^n(\De)=\cF_{p}^n(\l)$ where $\l=\De(1-\De)$ and $\cF_{p}^n(\l)$ is a polynomial in $\l$. This allows us to write the above equation as
\be\label{eq:crossedcrossing}
S_p=\cF_{p}^n(0)+\sum_{k=1}^\oo T_k \cF_{p}^n(\l_k).
\ee
The explicit expression for $\cF_{p}^n(\l)$ is
\be\label{eq:Fexpression}
\cF_{p}^n(\l)=
	\sum_{a+b+c=p-2n}
	\frac{(-1)^{a}(2n+a)_c (1-p)_{b}^2}{c!(2-2p)_{b}b!(a!)^2}
	\prod_{k=0}^{a-1}(\l+k+k^2),
\ee
where the sum is over non-negative $a,b,c$ subject to $a+b+c=p-2n$, and $(a)_n=a(a+1)\cdots(a+n-1)$.
It is a polynomial in $\l$ of degree $p-2n$ with rational coefficients. Let us stress that we get finite-degree polynomials without resorting to any approximations, which is another feature that is different form the usual numerical conformal bootstrap story.

We have ignored so far the fact that odd $p$ are absent from the sum in the left-hand side of~\eqref{eq:crossineq2}.\footnote{Also $p<2n$ are absent, but the expansion of right-hand side also starts at $z^{2n}$, so these do not lead to non-trivial constraints.} We can extend the sum to all $p$, with the understanding that $S_p=0$ for odd $p$. The argument leading to~\eqref{eq:crossedcrossing} is valid for all $p$, and so we then simply get
\be\label{eq:crossingfinal}
-S_p+\cF_{p}^n(0)+\sum_{k=1}^\oo T_k \cF_{p}^n(\l_k)&=0,\qquad \text{even }p\geq 2n,\\
\cF_{p}^n(0)+\sum_{k=1}^\oo T_k \cF_{p}^n(\l_k)&=0,\qquad\text{odd }p> 2n.\label{eq:crossingfinal2}
\ee

We now choose some integer $\L\geq 0$ and consider a sequence of real numbers $\a=\{\a_l\}_{l=0}^\L$. We will call such a sequence a functional. We take a linear combination of~\eqref{eq:crossingfinal} and~\eqref{eq:crossingfinal2} with coefficients $\a$, i.e.
\be\label{eq:functional_acted}
0=\sum_{l=0}^{\L}\cF_{2n+l}^n(0)\a_l+\sum_{k=1}^\oo \sum_{l=0}^{\L}\cF_{2n+l}^n(\l_k)\a_l T_k+\sum_{l=0\atop l\text{ even}}^{\L}(-\a_l) S_{2n+l}.
\ee
Suppose we manage to find a functional $\a$ such that $\a_l\leq 0$ for all even $l$, and the polynomial 
\be
P_\a^n(\l)\equiv \sum_{l=0}^{\L}\cF_{2n+l}^n(\l)\a_l
\ee
satisfies the following conditions for some value $\l_{gap}\in \R$,
\begin{itemize}
	\item $P_\a^n(0)=1$,
	\item $P_\a^n(\l)\geq 0$ for all $\l\geq \l_{gap}$.
\end{itemize}
Then we claim that $\l_1<\l_{gap}$. Indeed, equation~\eqref{eq:functional_acted} becomes
\be\label{eq:functionalcross}
0=1+\sum_{k=1}^\oo P_\a^n(\l_k) T_k+\sum_{l=0\atop l\text{ even}}^{\L}(-\a_l) S_{2n+l}.
\ee
If $\l_1\geq \l_{gap}$, then all the terms in this equation are non-negative (recall $T_k,S_p\geq 0$), and the first term is strictly positive, which leads to a contradiction.

We search for such functionals $\a$ using the following strategy. First, we choose some $\L$ and a large $\l_{gap}$. Then we use the \texttt{SDPB} software~\cite{Simmons-Duffin:2015qma,Landry:2019qug}, which was designed to solve this very problem, to find an $\a$ which satisfies the above conditions. This is not rigorous since internally \texttt{SDPB} uses non-exact (even if highly precise) arithmetic. Nevertheless, we proceed with \texttt{SDPB} to get some approximate answers which we then verify using a rigorous procedure.

We use binary search to determine the smallest (up to some predefined tolerance) $\l_{gap}$ for which \texttt{SDPB} can find an $\a$. We then improve $\l_{gap}$ by repeating this strategy with increasing $\L$, until we see no meaningful improvement (typically we stop at $\L$ between $25$ and $41$). At this point, we have a strong but non-rigorous bound $\l_{gap}$ and the corresponding numerical functional $\a_{numerical}$.

Finally, we approximate $\l_{gap}$ by a rational number which we call $\l_1^\text{single}(n)$, and similarly we approximate $\a_{numerical}$ by a functional $\a_{exact}$ with rational coefficients. We do not know yet if the functional $\a_{exact}$ satisfies the required conditions, so we verify that it does. This can be done rigorously since $P_{\a_{exact}}^n(\l)$ is a polynomial with rational coefficients. We describe the specific algorithm that we use in appendix~\ref{app:polynomial}. We find that $\a_{exact}$ always satisfies the required conditions, provided it approximates $\a_{numerical}$ with sufficiently accuracy. Mathematica notebooks which perform the verification procedure in rational arithmetic are included with the submission.

The result of this procedure is a rational number $\l_1^\text{single}(n)$ such that the bound
\be\label{eq:boundsingle}
	\l_1<\l_1^\text{single}(n)
\ee
is satisfied by any hyperbolic manifold for which $\ell_n>0$, i.e.\ for which the crossing equation~\eqref{eq:crossineq} makes sense. In this paper we report the values $\l_1^\text{single}(n)$ (and their generalizations that we describe below) as exact rational numbers in decimal notation. There is some arbitrariness in the reported number due to the choice of $\L$ and at which point the binary search is stopped. We expect that for the values we report, at most the last reported digit can be improved by further increasing $\L$ or running the binary search for more iterations. The expectation about $\L$ is mostly speculative, we did not estimate the rate of convergence in most cases.

We now discuss several generalizations of the bound~\eqref{eq:boundsingle} that will be needed in section~\ref{sec:bounds}. These generalizations are technical and can safely be skipped on the first reading.

\subsubsection*{Bounds with multiple coherent states.}

The above discussion used the crossing equation~\eqref{eq:crossineq} that is valid whenever $\ell_n>0$, i.e. there exists a discrete series coherent state $\sO_n(z)$. We will see that $\ell_n>1$ in many cases, and then it makes sense to consider the correlators involving all $\sO_{n,\a}(z)$ with $\a=1,\ldots,\ell_n$. In other words, we focus on
\be\label{eq:correlatormutli}
	\<\sO_{n,a_1}(z_1)\sO_{n,a_2}(z_2)\widetilde\sO_{n,a_3}(z_3)\widetilde\sO_{n,a_4}(z_4)\>.
\ee
The situation here is closely analogous to conformal bootstrap with global symmetries~\cite{Kos:2013tga}, except that we do not assume that there is any particular symmetry acting on the indices $a_i$. As we will see, simply assuming a degeneracy $\ell_n>1$, perhaps surprisingly, leads to a non-trivial improvement of the bounds.

The crossing equation corresponding to the above correlator is a straightforward generalization of~\eqref{eq:crossineq},
\be\label{eq:crossineqmulti}
	\sum_{p=2n}^{\oo}\sum_{a=1}^{\ell_p}f_{p,a}^{a_1,a_2}\overline{f_{p,a}^{a_4,a_3}} \cG_p(z)=\p{\frac{z}{1-z}}^{2n}\p{\de^{a_1a_4}\de^{a_2a_3}+\sum_{k=1}^{\oo}\sum_{a=1}^{d_k}c_{k,a}^{a_1;a_4}c_{k,a}^{a_2;a_3} \cH_{\De_k}(z)}.
\ee
Note that there is no restriction to even $p$ anymore. Here, the coefficients $f_{p,a}^{a_1,a_2}$ and $c_{k,a}^{a_1;a_4}$ have the following properties
\be\label{eq:symmetries_of_matrices}
f_{p,a}^{bc}=(-1)^p f_{p,a}^{cb},\qquad \overline{c_{p,a}^{b;c}}=c_{p,a}^{c;b}.
\ee
Taking into account the above symmetry property of $f_{p,a}^{\a\b}$ is, as before, equivalent to adding $u$-channel to the crossing equations, which is why we don't consider it separately. 

In principle, one could proceed with analyzing this system using a semidefinite generalization of the linear program considered in the previous section. This has a disadvantage that the size of this system grows quickly with $\ell_n$, significantly slowing down the calculation. It turns out that one can significantly optimize the calculations using the following observation. 

One can imagine a hypothetical situation in which the unitary group $U(\ell_n)$ acts on $\G\@ G$. In such a situation the function correlator~\eqref{eq:correlatormutli} as well as the product expansions would be constrained by $U(\ell_n)$ symmetry. It is possible to write down a $U(\ell_n)$-symmetric analogue of~\eqref{eq:correlatormutli} which is, as we will see, much easier to analyze. Of course, we have no reason to expect that such a $U(\ell_n)$ action exists, and so the bounds derived under the assumption of $U(\ell_n)$ symmetry do not obviously apply to~\eqref{eq:correlatormutli}. Instead, we expect them to be strictly stronger than the bounds we can obtain without assuming $U(\ell_n)$ symmetry.

A surprise first observed in~\cite{Poland:2011ey} is that in practice it often happens that bootstrap bounds obtained under a weaker symmetry assumption coincide with the bounds obtained under a stronger symmetry assumption. In our case, we will show that the bounds obtainable from~\eqref{eq:correlatormutli} are the same as the bounds obtainable under the assumption of $U(\ell_n)$ symmetry.  Coincidences of this kind were first explained in~\cite{Li:2020bnb} using a detailed analysis of matrices entering the semidefinite programs. Here we will present a more conceptual explanation, which is easily generalizable.

To show that the bounds coincide we demonstrate the following statements. First, if we have an $U(\ell_n)$-symmetric function correlator, it also satisfies~\eqref{eq:crossineqmulti}, and thus the functionals constructed for~\eqref{eq:crossineqmulti} can be used to construct functionals for the $U(\ell_n)$-symmetric case. This is obvious. The less trivial point is as follows: there is a linear map which maps any function correlator satisfying~\eqref{eq:crossineqmulti} to a $U(\ell_n)$-symmetric function correlator. This allows us to pull back the $U(\ell_n)$-symmetric functionals to construct functionals for~\eqref{eq:crossineqmulti}. Instead of explaining this map precisely (which is possible) we will simply demonstrate that $U(\ell_n)$ symmetric-bounds apply to~\eqref{eq:crossineqmulti}.

For this, consider the following expression 
\be
	\sum_{a'_1,\cdots,a_4'}\int dU U^{a'_1}_{a_1}U^{a'_2}_{a_2}\overline{U^{a'_3}_{a_3}}\overline{U^{a'_4}_{a_4}} \<\sO_{n,a'_1}(z_1)\sO_{n,a'_2}(z_2)\widetilde\sO_{n,a'_3}(z_3)\widetilde\sO_{n,a'_4}(z_4)\>,
\ee
where the integral is taken with the unit-normalized Haar measure over the unitary group $U(\ell_n)$. This expression can be interpreted as a new correlator
\be\label{eq:symmcorr}
=\<\sO_{n,a_1}(z_1)\sO_{n,a_2}(z_2)\widetilde\sO_{n,a_3}(z_3)\widetilde\sO_{n,a_4}(z_4)\>_\text{sym}
\ee
which is now invariant under $U(\ell_n)$ action. Plugging in the expansions in the left-hand side of~\eqref{eq:crossineqmulti}, we encounter sums of the form
\be
	\sum_{a'_1,\cdots,a_4'}\int dU U^{a'_1}_{a_1}U^{a'_2}_{a_2}\overline{U^{a'_3}_{a_3}}\overline{U^{a'_4}_{a_4}} f_{p,a}^{a'_1,a'_2}\overline{f_{p,a}^{a'_4,a'_3}}.
\ee
Note that $f_{p,a}^{a'_1,a'_2}=f_{p,a}^{[a'_1,a'_2]}+f_{p,a}^{(a'_1,a'_2)}$,\footnote{We define $A^{[a,b]}=\thalf\p{A^{a,b}-A^{b,a}}$, $A^{(a,b)}=\thalf\p{A^{a,b}+A^{b,a}}$.} where each term transforms irreducibly under $U(\ell_n)$. Schur orthogonality relations then imply
\be
	&\sum_{a'_1,\cdots,a_4'}\int dU U^{a'_1}_{a_1}U^{a'_2}_{a_2}\overline{U^{a'_3}_{a_3}}\overline{U^{a'_4}_{a_4}} f_{p,a}^{a'_1,a'_2}\overline{f_{p,a}^{a'_4,a'_3}}\nn\\
	&=(\delta_{a_1,a_4}\delta_{a_2,a_3}-\delta_{a_1,a_3}\delta_{a_2,a_4})A_{p,a}+(\delta_{a_1,a_4}\delta_{a_2,a_3}+\delta_{a_1,a_3}\delta_{a_2,a_4})S_{p,a},
\ee
where $0\leq A_{p,a}\propto \sum_{a_1,a_2}|f_{p,a}^{[a_1,a_2]}|^2$ and $0 \leq S_{p,a}\propto \sum_{a_1,a_2} |f_{p,a}^{(a_1,a_2)}|^2,$ up to some inessential normalization constants. Of course,~\eqref{eq:symmetries_of_matrices} implies that $A_{p,a}=0$ for even $p$ and $S_{p,a}=0$ for odd $p$.

A similar decomposition holds for $c_{k,a}^{x;y}=\p{c_{k,a}^{x;y}-\frac{1}{\ell_n}\de^{x,y}\sum_i c^{i,i}_{k,a}}+\frac{1}{\ell_n}\de^{x,y}\sum_i c^{i,i}_{k,a}$, which leads to
\be
	&\sum_{a'_1,\cdots,a_4'}\int dU U^{a'_1}_{a_1}U^{a'_2}_{a_2}\overline{U^{a'_3}_{a_3}}\overline{U^{a'_4}_{a_4}}c_{k,a}^{a_1;a_4}c_{k,a}^{a_2;a_3}\nn\\
	&=\de_{a_1,a_4}\de_{a_2,a_3} Q_{k,a}+\p{\de_{a_1,a_3}\de_{a_2,a_4}-\frac{1}{\ell_n}\de_{a_1,a_4}\de_{a_2,a_3}}T_{k,a},
\ee
where $Q_{k,a},T_{k,a}\geq 0$. Importantly, if at least one component of $c_{k,a}$ is non-zero, then at least one of $Q_{k,a}$ or $T_{k,a}$ is non-zero.  This implies that the $U(\ell_n)$-symmetric correlator has exactly the same spectral gap as~\eqref{eq:correlatormutli}.

Physicists will recognize that what we have found is that the $U(\ell_n)$-symmetric correlator~\eqref{eq:symmcorr} has $U(\ell_n)$-symmetric conformal block decompositions, and coefficients $A,S,Q,T$ correspond, respectively, to exchanges in the anti-symmetric, symmetric, scalar, and traceless tensor channels. Analysis of such correlators is a standard procedure in the physics literature~\cite{Rattazzi:2010yc}. We now sketch how it is performed.

The result of this calculation is that the correlator~\eqref{eq:symmcorr} has the following two expansions
\be
	&\p{\frac{z}{1-z}}^{-2n}\sum_{p=2n}^{\oo}\sum_{a=1}^{\ell_p}\cG_p(z)\p{(\delta_{a_1,a_4}\delta_{a_2,a_3}-\delta_{a_1,a_3}\delta_{a_2,a_4})A_{p,a}+(\delta_{a_1,a_4}\delta_{a_2,a_3}+\delta_{a_1,a_3}\delta_{a_2,a_4})S_{p,a}}\nn\\
	&=\de^{a_1a_4}\de^{a_2a_3}+\sum_{k=1}^{\oo}\sum_{a=1}^{d_k}\p{\de_{a_1,a_4}\de_{a_2,a_3} Q_{k,a}+\p{\de_{a_1,a_3}\de_{a_2,a_4}-\frac{1}{\ell_n}\de_{a_1,a_4}\de_{a_2,a_3}}T_{k,a}} \cH_{\De_k}(z),
\ee
where we moved the factor $\p{\frac{z}{1-z}}^{2n}$ to the left-hand side. Equating the coefficients of the products of Kronecker deltas we find
\be
	\sum_{p=2n}^{\oo}\sum_{a=1}^{\ell_p}\cG_p(z)\p{A_{p,a}+S_{p,a}}&=\p{\frac{z}{1-z}}^{2n}\p{1+\sum_{k=1}^{\oo}\sum_{a=1}^{d_k}\p{Q_{k,a}-\frac{1}{\ell_n}T_{k,a}}} \cH_{\De_k}(z),\\
	\sum_{p=2n}^{\oo}\sum_{a=1}^{\ell_p}\cG_p(z)\p{-A_{p,a}+S_{p,a}}&=\p{\frac{z}{1-z}}^{2n}\sum_{k=1}^{\oo}\sum_{a=1}^{d_k}T_{k,a} \cH_{\De_k}(z).
\ee
Using the same logic that lead us to~\eqref{eq:crossedcrossing} we conclude for even $p\geq 2n$
\be\label{eq:red1}
S_p&=\cF_{p}^n(0)+\sum_{k=1}^\oo \p{Q_{k}-\frac{1}{\ell_n}T_{k}} \cF_{p}^n(\l_k)=\sum_{k=1}^\oo T_{k}\cF_{p}^n(\l_k),
\ee
and for odd $p>2n$
\be\label{eq:red2}
A_p&=\cF_{p}^n(0)+\sum_{k=1}^\oo \p{Q_{k}-\frac{1}{\ell_n}T_{k}} \cF_{p}^n(\l_k)=-\sum_{k=1}^\oo T_{k}\cF_{p}^n(\l_k).
\ee
Above we set $A_{p}=\sum_a A_{p,a}$ and so on.
These equations can be analyzed in exactly the same way as in the case of a single $\sO_n(z)$, this time introducing slightly more components for the functional $\a$. We denote the resulting upper bound on $\l_1$ by $\l_1^\text{multi}(n,\ell_n)$. 

\subsubsection*{Bounds from a mixed system $\sO_n(z)$ and $\sO_m(z)$.}

Finally, we will find it useful to consider bounds derived from the function correlator involving two distinct discrete series representations, (i.e.\ $\sO_n(z)$ and $\sO_m(z)$ with $m\neq n$),  as well as three distinct discrete series representations. 

Compared to the cases considered above, the new ingredient is that now discrete series representations $\sD_l$ or $\bar\sD_l$ can contribute to the products $\sO_m(z_1)\widetilde\sO_n(z_2)$, with $l\leq |m-n|$. While it is straightforward to analyze these contributions using the methods used in the proof of Theorem~\ref{theorem:main}, we will not need these. This is because in all examples where we will apply the bounds coming from such crossing equations, we will have $\ell_l=0$ for all $l\leq |m-n|$, and so while these contributions are allowed by symmetries, they will not appear in practice. 

For the same reason, the four-point function correlators of the form $\<\sO\sO\sO\widetilde\sO\>$ and $\<\widetilde\sO\widetilde\sO\widetilde\sO\sO\>$, where $\sO$ stands for either $\sO_m$ or $\sO_n$, will vanish in our applications. So the only non-zero correlators that we can study are of the form $\<\sO\sO\widetilde\sO\widetilde\sO\>$. 

Their analysis proceeds in a way similar to before, the main difference being that the $s$- and $t$-channel expansions now take a different form, with more general versions of $\cG_p$ and $\cH_\De$ appearing. The new ingredient that we will need is an expansion of a general correlator
\be\label{eq:correlatormutliGeneral}
	\<\sO_{n_1,a_1}(z_1)\sO_{n_2,a_2}(z_2)\widetilde\sO_{n_3,a_3}(z_3)\widetilde\sO_{n_4,a_4}(z_4)\>
\ee
into s- and t-channel conformal blocks. By generalizing the previous steps to this case, we find the following crossing equation
\be\begin{aligned}\label{eq:crossineqmultiGeneral}
	&\sum_{p=K}^{\oo}\sum_{a=1}^{\ell_p}f_{p,a}^{a_1,a_2}\overline{f_{p,a}^{a_4,a_3}} \cG_{p;\{n_i\}}(z) =\\
	&=(-1)^{\sum_i n_i}\frac{z^{n_1+n_2}}{(1-z)^{n_2+n_3}}\p{\de_{n_1n_4}\de_{n_2n_3}\de_{a_1a_4}\de_{a_2a_3}+\sum_{k=1}^{\oo}\sum_{a=1}^{d_k}c_{k,a}^{a_1;a_4}c_{k,a}^{a_2;a_3} \cH_{\De_k;\{n_i\}}(z)}.
	\end{aligned}
\ee
where $K=\text{max}\left(n_1+n_2,n_3+n_4\right)$ and the conformal blocks are given by
\begin{equation}
\begin{aligned}
&\cG_{p;\{n_i\}}(z)=z^{p} \, _2F_1(p-n_{12},p+n_{34};2p;z)\\
&\cH_{\De_k;\{n_i\}(z)} =\\
&=\begin{cases}
(1-z)^{n_{32}} \, _2F_1\left(\De_k+n_{23},-\De_k+n_{23}+1;n_{13}+n_{24}+1;\frac{z}{z-1}\right)\,;\ &n_{13} \geq n_{42}\\
(1-z)^{n_{23}} z^{n_{41}+n_{32}} \, _2F_1\left(\De_k-n_{23},n_{32}+1-\De_k;n_{32}+n_{41}+1;\frac{z}{z-1}\right)&\text{otherwise}
\end{cases}
\end{aligned}
\end{equation}
where $n_{ij}=n_i-n_j$. As explained above, we have assumed that $\{n_i\}$ is such that there is no contribution from discrete series on the right hand side of~\eqref{eq:crossineqmultiGeneral}.

One can now set up a mixed correlator bootstrap, using all crossing equations $\<\sO_i\sO_j\widetilde\sO_k\widetilde\sO_l\>$, where $i,j,k,l$ range over some set. We will not spell out all the details here.\footnote{This generalization is completely standard in numerical conformal bootstrap and was introduced for the first time in~\cite{Kos:2014bka}.}

\section{Bounds on hyperbolic orbifolds}
\label{sec:bounds}
In this section, we will present our bounds on the spectral data of hyperbolic 2-orbifolds. Specifically, we will obtain upper bounds on the first positive eigenvalue $\lambda_1$, as well as bounds on triple product integrals of modular forms.
\subsection{Classification of hyperbolic orbifolds}
Let us recall that in this work by hyperbolic surface, we mean a closed connected smooth orientable 2-manifold with a hyperbolic metric of sectional curvature $-1$, and hyperbolic orbifold means the same but generalized from manifolds to orbifolds.

As we discussed in the Introduction, every hyperbolic surface is isometric to $X=\G\@ G/K$ for some discrete cocompact torsion-free subgroup $\G\subset G$. However, our method for deriving bounds applies also if $\Gamma$ has torsion, i.e.\ nontrivial elements of finite order. We will therefore take $\G$ to be a general discrete cocompact subgroup of $G$. If $\G$ has torsion, $X$ is not a smooth hyperbolic surface, but it is a hyperbolic orbifold. 

In practice, we can think about hyperbolic orbifolds in the following simple-minded way. Every $\G$ has a torsion-free finite-index normal subgroup $\G'\unlhd \G$~\cite{normaltorsionfree1,normaltorsionfree2,normaltorsionfree3,normaltorsionfree4}. A simple exercise then shows that the finite group $S=\G/\G'$ acts by isometries on the hyperbolic surface $X'=\G'\@ G/K$, and $X=S\@ X'$. In other words, any hyperbolic orbifold can be thought of as a quotient of a hyperbolic surface by a finite group of isometries. In particular, the Laplacian spectrum of $X$ consists of the eigenvalues on $X'$ whose eigenfunctions are invariant under $S$. The necessity of considering this option should sound familiar to a reader experienced in the conformal bootstrap: there too one can obtain consistent solutions of crossing equations by restricting to a singlet sector under a global symmetry.

The subgroups $\G$ have a simple classification. For each $\G$, there exist non-negative integers $g,r$ and positive integers $k_1,\cdots, k_r\geq 2$ such that
$\G$ is generated by elements $a_i,b_i,t_j\in G$ ($1\leq i\leq g,\,1\leq j\leq r$), subject only to the relations $t_j^{k_j}=1$ and $[a_1,b_1]\cdots[a_g,b_g]t_1\cdots t_r=1$. Here $[a,b]=a^{-1}b^{-1}ab$. In order to conveniently summarize this information, we associate to $\G$ its \emph{signature} $[g;k_1,\ldots,k_r]$.\footnote{Since there is no canonical ordering on the punctures, we will take $k_1\leq k_2\leq\ldots\leq k_r$ without loss of generality.} There exists a $\G$ with signature $[g;k_1,\ldots,k_r]$ if and only if $\chi(g;k_1,\ldots, k_r)<0$, where
\be\label{eq:chi}
	\chi(\G)=\chi(g;k_1,\ldots, k_r)=2-2g-\sum_{i=1}^r\frac{k_i-1}{k_i}
\ee
is the orbifold Euler characteristic. When $\G$ exists, $-2\pi\chi(\G)$ is the hyperbolic area of $X$.

If $\G$ has the signature  $[g;k_1,\ldots,k_r]$, then $X$ has genus $g$ and $r$ orbifold points of integer orders $k_1,\ldots, k_r$. In other words, the signature specifies the topological information about $X$. In most cases, there exist many different $\G$ with the same signature which lead to non-isometric $X$. In such cases we say that there is a non-trivial moduli space of $\G$ (or $X$) with signature $[g,k_1,\ldots k_r]$. 

\begin{figure}[h]
	\centering
	\includegraphics[width=.75\textwidth]{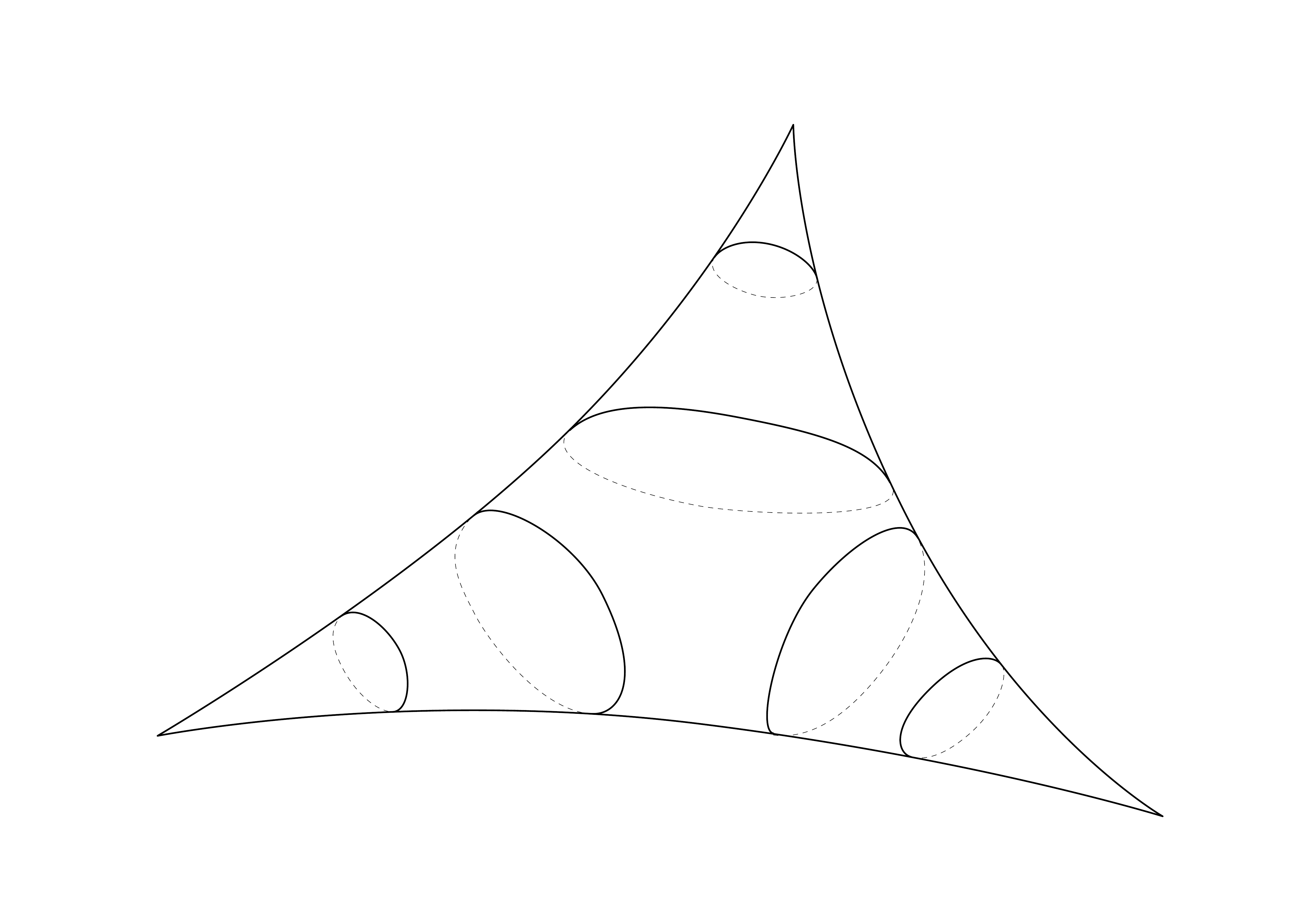}
	\caption{Artist's impression of a $[0;k_1,k_2,k_3]$ orbifold (not to scale).}
	\label{fig:sphere}
\end{figure}

The only exception to this is the case $g=0$ and $r=3$ (it is not possible to have $g=0$ and $r<3$ due to $\chi(\G)<0$). For a given signature $[0;k_1,k_2,k_3]$, there exists a unique $\G$ (up to conjugation), and correspondingly a unique hyperbolic orbifold $X$ (up to isometry). We will often refer to it as the $[0;k_1,k_2,k_3]$ orbifold. Intuitively, it is a sphere with three conical singularities with angle deficits $2\pi\tfrac{k_i-1}{k_i}$, see Figure~\ref{fig:sphere}. It can be constructed in the following way. First, consider a hyperbolic triangle in $\mathbb{H}$ with interior angles $\pi/k_1,\pi/k_2,\pi/k_3$, i.e.\ the $(k_1,k_2,k_3)$ hyperbolic triangle. It exists iff $\chi(0;k_1,k_2,k_3)<0$ and is unique up to an isometry of $\mathbb{H}$. Let $\G_0$ be the group generated by hyperbolic reflections in the sides of this triangle,  usually called the $(k_1,k_2,k_3)$ triangle group. Then $\G$ is just the orientation-preserving subgroup of $\G_0$, and is usually called the $(k_1,k_2,k_3)$ von Dyck group. Geometrically, $X$ consists of two copies of the $(k_1,k_2,k_3)$ triangle glued along the corresponding sides. The $[0;k_1,k_2,k_3]$ orbifolds will play an important role in the following discussion.

\subsection{Riemann-Roch theorem}
It is a standard result that one can introduce a unique complex structure on $X$ such that in local holomorphic coordinates $w$ the metric takes the form $ds^2=e^{2\f(z,\bar z)}dw d\bar w$ for some smooth function $\f(w,\bar w)$. Viewing $X$ as $\G\@ \mathbb{H}$, this complex structure is obtained from the standard complex structure $z=x+iy$ on $\mathbb{H}$.

Let $\cK$ be the canonical line bundle of $X$. Simply put, holomorphic sections of $\cK$ are, in local coordinates, objects of the form $f(w)dw$, where $f(w)$ is a holomorphic function. Similarly, sections of the $n$-th tensor power $\cK^{\otimes n}$ are objects of the form $f(w)dw^n$. If $f$ is a section of $\cK^{\otimes n}$ on $X=\G\@ \mathbb{H}$, then its pullback to $\mathbb{H}$ is an object of the form $f(z)dz^n$ that is invariant under the action of $\G$. The function $f(z)$ then transforms as a modular form of weight $2n$. In Section~\ref{sec:bootstrapMethod}, we identified the dimension of the space of such modular forms with the multiplicity $\ell_n$, which is the number of times $\cD_n$ appears in~\eqref{eq:SL2RdecompositionFine}.

This shows that the multiplicity $\ell_n$ is equal to the dimension of the space of holomorphic sections of $\cK^{\otimes n}$. Riemann-Roch theorem allows one to compute the latter dimension, and thus provides a formula for $\ell_n$. The result only depends on the signature $[g;k_1,\ldots,k_r]$ of $\G$ and reads (see e.g.\ Theorem 4.9 in \cite{MilneModularForms})
\be\label{eq:RRoch}
	\ell_n=(g-1)(2n-1)+\sum_{i=1}^r\left\lfloor n\frac{k_i-1}{k_i}\right\rfloor+\de_{n,1}\,.
\ee

Our strategy will be to use this formula to show that $\ell_n>0$ for some $n$, and then study the crossing equations for the four-point function correlator involving $\sO_{n,a}(z)\in\sD_n$ to obtain bounds on the Laplace eigenvalues $\l_k$.

\subsection{Numerical estimates of Laplace eigenvalues}

\label{sec:ff}
We are not aware of any exact formulas for the eigenvalues of hyperbolic orbifolds. However, it is possible to obtain numerical estimates using various methods. For example, very precise eigenvalues for the Bolza surface (which will be discussed in Section~\ref{sec:higherg}) with rigorous error bars were obtained in~\cite{strohmaier2013algorithm}. A non-rigorous approach was used in~\cite{cook2021properties}, where the eigenvalues of various hyperbolic surfaces were obtained using FreeFEM software~\cite{freefem}.\footnote{We were only able to reproduce several digits of the eigenvalues provided in~\cite{cook2021properties} using the method described there, the rest being sensitive to the choice of discretization.}

In this work, we used the latter approach in order to compare the eigenvalues against our bounds. In practice, this approach amounts to discretizing a fundamental domain of $\G$ in $\mathbb{H}$, computing the eigenvalues of a finite-dimensional system, and then seeing how the results change with the number of discretization points $N$. We took some care in ensuring that our results remain stable as we increase $N$. In some cases we found that we could do a reliable extrapolation by an inverse power law in $N$.

We do not discuss these computations in detail since they are, in any case, non-rigorous. However, we do expect that all approximate eigenvalues of specific surfaces quoted in this paper are correct up to a possible error in the last digit.

\subsection{Universal bounds on hyperbolic orbifolds}
\label{sec:universalorbifold}

\begin{table}[!t]
	\begin{center}
		\begin{tabular}{l|l|l|l||ll}

			$n_{min}$ & $r$ & $g$ & $\l_1^{\text{single}}(n_{min})$ & largest known $\l_1$ & signature \\
			\hline
			1& $\geqslant 0$&$\geqslant 1$& 8.47032& 8.46776 & [1;2] \\

			2& $\geqslant 4$& 0& 15.79144& 15.79023& [0;2,2,2,3] \\
						3& 3& 0& 23.07917& 23.07855 &[0;3,3,4]\\
						4& 3& 0& 30.35432& 28.07984 &[0;2,4,5]\\
						6 & 3& 0& 44.8883537& 44.88835& [0;2,3,7]\\
		\end{tabular}
	\end{center}
	\caption{Single correlator bound as a function of $n_{min}$. The values of $\l_1^{\text{single}}(n_{min})$ should be interpreted as exact rational numbers. For comparison, we provide also the approximate non-rigorous estimates of the largest known $\l_1$ among orbifolds with the given $n_{min}$. The examples with $n_{min}=1,2$ are at special points in the moduli space described in the main text.\label{table:singlecorr}}
\end{table}

In Section~\ref{sec:derivation}, we explained that if $\ell_n>0$ then one can study the crossing equation~\eqref{eq:crossineq} for the four-function correlator
\be
\<\sO_n(z_1)\sO_n(z_2)\widetilde\sO_n(z_3)\widetilde\sO_n(z_4)\>
\ee
to obtain bounds of the form
\be\label{eq:single}
	\l_1< \l_1^{\text{single}}(n),
\ee
where $\l_1^{\text{single}}(n)$ are obtained by a computer-assisted numerical search. 

The bound $\l_1^{\text{single}}(n)$ for some small values of $n$ is given in Table~\ref{table:singlecorr}. We can apply this bound for any $n$ for which $\ell_n>0$. However, we find that $\l_1^{\text{single}}(n)$ is an increasing function of $n$. Therefore, we get the strongest bound by applying it to $n=n_{min}$ which is defined to be the smallest value of $n$ for which $\ell_n>0$.

A simple counting argument using~\eqref{eq:RRoch} and the requirement that $\chi(\G)<0$, c.f.\ \eqref{eq:chi}, shows that $n_{min}\in \{1,2,3,4,6\}$ for any $\G$. The proof, which we omit, proceeds by verifying the following facts
\begin{itemize}
	\item $n_{min}=1$ for all orbifolds with $g>0$,
	\item $n_{min}=2$ for all orbifolds with $g=0$ and $r>3$,
	\item $n_{min}=3$ for all orbifolds with signature $[0;k_1,k_2,k_3]$ with $k_1,k_2,k_3\geq 3$,
	\item $n_{min}=4$ for all orbifolds with signature $[0;2,k_2,k_3]$ with $k_2,k_3\geq 4$,
	\item $n_{min}=6$ for all orbifolds with signature $[0;2,3,k_3]$ (these all have $k_3\geq 7$),
	\item this list exhausts all possibilities for $\G$.
\end{itemize}
The value of $\l_1^\text{single}(n_{min})$ then provides an upper bound on $\l_1$ among the hyperbolic orbifolds of the corresponding class described in the above list. In particular, by observing the values in Table~\ref{table:singlecorr}, we immediately obtain the following
\begin{theorem}
	For any hyperbolic 2-orbifold, the first non-zero eigenvalue $\l_1$ of the Laplace operator satisfies $\l_1< 44.8883537$.
\end{theorem}
More specialized versions of this theorem can be formulated based on the above discussion.

It is interesting to ask how close the bound~\eqref{eq:single} is to being saturated for various values of $n=n_{min}$. As shown in Table~\ref{table:singlecorr}, we have found concrete examples of $\G$ for which $\l_1$ is rather close to $\l_1^\text{single}(n_{min})$, except in the case $n_{min}=4$. Let us say a few words about these examples. Given our results, it is natural conjecture that they maximize $\l_1$ for the given $n_{min}$.\footnote{The question of uniqueness is more subtle. For example, an orbifold and its covering may have the same value of $\l_1$.}
\begin{itemize}
	\item For $n_{min}=6$, the bound appears to be saturated by the $[0;2,3,7]$ hyperbolic orbifold, i.e.\ a sphere with three orbifold singularities of orders 2,3,7. This is the hyperbolic orbifold of the smallest area, which is $\pi/21$.
	\item For $n_{min}=4$, the largest $\l_1$ that we found occurs in the orbifold $[0;2,4,5]$, although it is not particularly close to saturating the bound. This is the orbifold of the smallest area such that $n_{min}=4$. The area is $\pi/10$. 
	\item For $n_{min}=3$, the bound is nearly saturated by the $[0;3,3,4]$ hyperbolic orbifold of area $\pi/6$. The $[0;2,3,8]$ orbifold has the same $\l_1$ (it can be double-covered by $[0;3,3,4]$), but has $n_{min}=6$.
	\item For $n_{min}=2$, the bound is nearly saturated by the orbifold $X$ with signature $[0;2,2,2,3]$ (a sphere with four orbifold points), at the most symmetric point in the moduli space. Concretely, it has a $\Z_3$ symmetry which cyclically permutes the order-2 orbifold points and fixes the order-3 orbifold point. The quotient $\Z_3\@ X$ is the hyperbolic orbifold with signature $[0;2,3,9]$, which has the same $\l_1$.\footnote{Strictly speaking, we have not computed the spectrum of $X$. Instead, we computed the spectrum for $[0;2,3,9]$. We expect that the first non-constant Laplace eigenfunction on $X$ is $\Z_3$-invariant and thus is also present on $[0;2,3,9]$.}
	\item For $n_{min}=1$, the bound is nearly saturated by the most symmetric hyperbolic orbifold $X$ with signature $[1;2]$, i.e.\ a torus with one orbifold singularity. Specifically, it has a $\Z_6$ symmetry and $\Z_6\@ X$ is the $[0;2,3,12]$ orbifold (which has the same $\l_1$).
\end{itemize}

The best bound that we could find in the literature which applies to $g=0$ and $g=1$ orbifolds is the Yang-Yau bound~\cite{yang1980eigenvalues,el1983volume}. It will be discussed in more detail in Section~\ref{sec:higherg}. In the cases $g=0$ and $g=1$ it yields
\be
	\l_1\leq \frac{4(g+1)}{|\chi(\G)|}.
\ee
We can then obtain a bound at fixed $n_{min}$ by minimizing $|\chi(\G)|$ over $\G$ with that $n_{min}$. The result is given in the Table~\ref{table:comparisonOrbifold}. We see that in all cases our bound is much stronger. It should be noted, however, that Yang-Yau bound is not tailored to hyperbolic orbifolds and is more general.
\begin{table}[t]
	\begin{center}
		\begin{tabular}{l|l|l|l||l}
			$n_{min}$ & $r$ & $g$ & bound here & Yang-Yau bound \\
			\hline
			1& $\geqslant 0$&$\geqslant 1$& {8.47032}&16\\
			2& $\geqslant 4$& 0& {15.79144}& 24\\
			3& 3& 0& {23.07917}& 48 \\
			4& 3& 0& {30.35432}& 80 \\
			6 & 3& 0& {44.8883537}& 168\\
		\end{tabular}
	\end{center}
	\caption{
		Comparison of the bounds from Table~\ref{table:singlecorr} with the Yang-Yau bound~\cite{yang1980eigenvalues,el1983volume}.
		\label{table:comparisonOrbifold}
	}
\end{table}

\subsection{Bounds on hyperbolic orbifolds with $g\geq 2$.}\label{sec:higherg}

The bounds discussed above are very strong for $g=0$ ($n_{min}>1$) and $g=1$ ($n_{min}=1$), being nearly saturated by known orbifolds. However, from the point of view of the $n_{min}$ invariant, all $g\geq 2$ orbifolds belong to the $n_{min}=1$ case, and therefore the bounds described so far are not sensitive to the genus $g\geq 2$.

In this subsection, we use another observation to give bounds on the spectra of $g\geq 2$ orbifolds specific to a given genus. Concretely, for any $g$, the Riemann-Roch formula~\eqref{eq:RRoch} implies $\ell_1=g$. This allows us to use the bounds derived from the system of four-function correlators
\be
	\<\sO_{1,a}(z_1)\sO_{1,b}(z_2)\widetilde\sO_{1,c}(z_3)\widetilde\sO_{1,d}(z_4)\>,
\ee
where $1\leq a,b,c,d\leq g$, as discussed in Section~\ref{sec:derivation}. There, we explained how to obtain bounds of the form
\be
	\l_1<\l_1^\text{multi}(1,g),
\ee
where $\l_1^\text{multi}(1,g)$ is obtained by a computer-assisted numerical search.

We list the values of $\l_1^\text{mutli}(1,g)$ in Table~\ref{table:mutli}, along with the previously known bounds, which we will review below. The most notable cases in this table are $g=2$ and $g=3$.

\begin{table}[t]
	\begin{center}
		\begin{tabular}{l|l|l|l|l}
			
			$g$ & $\l_1^\text{multi}(1,g)$ & YY\cite{yang1980eigenvalues,el1983volume}  & R\cite{ros2020first}, KV\cite{karpukhin2021improved} & H\cite{huber1980spectrum}  \\
			\hline
			2& 3.8388976481\optimal  & 4 & 4 & 5.13439\\
			
			3& 2.6784823893\optimal   & 3 & $2.7085$&3.05862 \\
			
			4 &2.1545041334 & 2\optimal & 2\optimal & 2.33877\\
			
			5& 1.8526509456\optimal & 2 & $ 1.96788$& 1.96497 \\
			
			6& 1.654468363 & 1.6\optimal & 1.6\optimal &1.73262\\
			
			7 &  1.51326783\optimal & $1.66667$ & $1.61628 $ &1.57257\\
			
			8 & 1.40690466\optimal & $1.42858$ & $1.42858$ & 1.45464\\
			
			9  & 1.32348160\optimal& $1.5$ & $ 1.39445$ &1.36365\\
			
			10  & 1.25602193\optimal& $1.33334$ & $1.33334$ &1.29087\\
			
			11 & 1.2001524\optimal& $1.4$ & $ 1.34987$ & 1.23115\\
			
			12 & 1.1529856\optimal & $1.27273$ & $  1.26994$ & 1.18109\\
			
			13 & 1.1125346\optimal & $1.33334$ & $1.26177$ & 1.13840\\
			
			14 & 1.077385\optimal & $1.23078$ & $ 1.23078$ & 1.10147\\
			
			15 & 1.046501\optimal & $1.28572$ & $ 1.20059$ &1.06921 \\
			
			16 & 1.019105\optimal & $1.2$ & $ 1.17158$ &1.04054 \\
			
			17 & 0.9946005\optimal& $1.25$ & $1.17158$ & 1.01514\\
			
			18 & 0.972525\optimal& $1.17648$ & $1.15762$ & 0.992157\\
			
			19 & 0.95251\optimal& $1.22223$ & $ 1.17084$ & 0.971396\\
			
			20 & 0.93426\optimal & $1.15790$ & $ 1.11431$ &0.952519  \\
			
			$\vdots$ & $\vdots$ &  $\vdots$ &$\vdots$ & $\vdots$\\
			
			$\infty$ & 0.52 & 1 & 0.854061 & 0.25\optimal \\
			
		\end{tabular}
	\end{center}
	\caption{
		\label{table:mutli}
		The bound $\l_1^\text{mutli}(1,g)$ as a function of genus $g$. We also list previously known bounds for comparison. The second column contains the Yang-Yau bound \eqref{eq:ESI}. The third contains the improved version due to \cite{ros2020first} and \cite{karpukhin2021improved}. The fourth one comes from \cite{huber1980spectrum}. In each row, the best bound is marked with a red asterisk\optimal.}
\end{table}

In the case $g=2$, our bound is
\be
	\l_1\leq 3.8388976481.
\ee
This is to be compared with the largest known $\l_1$ of any $g=2$ orbifold, 
\be
	\l_1 = 3.838887258\ldots,
\ee
which occurs for $X$ the Bolza surface.\footnote{The Bolza surface is the genus 2 hyperbolic surface with the highest possible order of orientation preserving isometry group. This group is $\mathrm{GL}(2,3)$ and has order $48$. The group $\G$ defining the Bolza surface is an index $96$ subgroup of the $(2,3,8)$ triangle group; in particular, one can tile the Bolza surface by $96$ $(2,3,8)$ hyperbolic triangles. As a complex surface, the Bolza surface can be viewed as the smooth completion of the affine algebraic curve $y^2=x^5-x$ in $\mathbb{C}^2$.} This eigenvalue is known with high precision due to~\cite{strohmaier2013algorithm}. We see that our bound is extremely close to the Bolza $\l_1$, giving very strong evidence for the conjecture that the Bolza surface maximizes $\l_1$ among all $g=2$ orbifolds.

In the case $g=3$, our bound is
\be
	\l_1 \leq 2.6784823893,
\ee
which should be compared to the eigenvalue~\cite{cook2021properties}\footnote{See Section~\ref{sec:ff}.}
\be
	\l_1\approx2.6779
\ee
of the Klein quartic.\footnote{The Klein quartic is a genus 3 hyperbolic surface with the highest possible order of the orientation preserving isometry group. This group is $\mathrm{PSL}(2,7)$ and has order $168$. The group $\G$ defining the Klein quartic is an index $336$ subgroup of the $(2,3,7)$ triangle group; in particular the surface can be can be tiled by 336 $(2,3,7)$ hyperbolic triangles. As a complex surface, it can be defined by the equation $x^3y+y^3z+z^3x=0$ in the homogeneous coordinates on $\C \mathrm{P}^2$.} We again see that our bound is very close to an eigenvalue which is realized by a known surface. Our bound strongly supports the conjecture that the Klein quartic maximizes the value of $\l_1$ among all genus 3 orbifolds.

As we move to genus $g=4$, we find from Table~\ref{table:mutli} that our bound is rather far from being saturated since there exist stronger previously known bounds in this case. The same happens at $g=6$. While our bound is still the strongest that we know for other values of $g\leq 20$, these examples (and the large-$g$ behavior discussed below) make us skeptical that it can be nearly saturated by an orbifold for $g\geq 4$. 

Let us finally comment about the last line in the Table~\ref{table:mutli}, corresponding to $g=\oo$. By $\l_1^\text{mutli}(1,\oo)$ we mean the following. If one examines the system~\eqref{eq:red1}-\eqref{eq:red2} in Section~\ref{sec:derivation}, one finds that it has a well-defined limit as $\ell_1=g\to \oo$. The value $\l_1^\text{mutli}(1,\oo)$ is the one obtained from this system with $\ell_1=g=\oo$. While we expect that $\l_1^\text{mutli}(1,\oo)=\lim_{g\to\oo}\l_1^\text{mutli}(1,g)$, we have not attempted to prove this rigorously; the value $\l_1^\text{mutli}(1,\oo)$ should be interpreted with care. The values of other bounds in this line are, on the other hand, rigorous $g\to \oo$ limits of the finite-$g$ bounds. We have also found that the value $\l_1^\text{mutli}(1,\oo)$ converges extremely slowly as a function of $\L$.\footnote{We estimate the convergence rate to be $\L^{-\half}$.} While we do not expect to be able to improve the other bounds in Table~\ref{table:mutli} much beyond the last digit by merely increasing $\L$, we are confident that $\l_1^\text{mutli}(1,\oo)$ can be improved significantly. The value cited in the table is for $\L=81$, and we expect it to be as low as $0.34$ for very large $\L$ based on a naive extrapolation by a power law.

In the rest of this subsection we will review the bounds previously known in the literature and discuss how our bounds compare with them.
We will consider a more general setup in which one puts bounds on the eigenvalues of the Laplace-Beltrami operator on a connected closed orientable two-dimensional Riemannian manifold $X$ with metric $h$. We will only discuss the bounds on the first non-zero eigenvalue $\l_1$, but some of them admit generalizations for the subleading eigenvalues as well.

We define
\begin{equation}\label{eq:Ldefn}
	\Lambda_k(X) = \underset{h}{\mathrm{sup}} \left(\lambda_k(X, h) \mathrm{Area}(X,h) \right)
\end{equation}
where the supremum is taken over all Riemannian metrics $h$ on the Riemann surface $X$. In~\cite{yang1980eigenvalues} Yang and Yau derived the following bound 
\begin{equation}\label{eq:YY}
	\Lambda_1(X) \leqslant 8\pi (g+1)\,.
\end{equation}
The bound \eqref{eq:YY} holds for orbifolds as well. Their argument was shown to yield a better bound (henceforth called YY) in~\cite{el1983volume}
\begin{equation}\label{eq:ESI}
	\Lambda_1(X) \leqslant 8\pi \bigg\lfloor \frac{g+3}{2}\bigg\rfloor.
\end{equation}
This bound is sharp for $g=0$ and $g=2$~\cite{karpukhin2019yang,nayatani2019metrics}. This bound was improved first by A.~Ros for $g=3$~\cite{ros2020first}, and then by~\cite{karpukhin2021improved} for any $g$. These bounds are an improvement as long as $g\not\in\{4,6,8,10,14\}$. For $g\geqslant 102$ the best improved bound is
\begin{equation}
	\Lambda_{1}\left(X_{h}\right) \leqslant \frac{2 \pi}{13-\sqrt{15}}\left(g+(33-4 \sqrt{15})\left\lceil\frac{5 g}{6}\right\rceil+4(41-5 \sqrt{15})\right)\,,
\end{equation}
while for smaller $g$ different formulas should be used as described in~\cite{karpukhin2021improved}. The above expression is stronger than YY for $\g\geqslant 25$. Note that all the bounds discussed so far are valid for any Riemannian metric on $X$, not just the hyperbolic one. The bound for the hyperbolic metric is obtained from the definition~\eqref{eq:Ldefn} by using the fact that for a hyperbolic metric $h$ of sectional curvature $-1$, we have $\mathrm{Area}(X,h)=-2\pi\chi(\G)$ where $\chi(\G)$ is defined in~\eqref{eq:chi}.

A bound that is valid only for hyperbolic metrics was derived in~\cite{huber1980spectrum}, where it was shown that for a hyperbolic manifold, $\lambda_1$ is bounded from above by a $g$-dependent constant which goes to $1/4$ as $g\to \infty$. In fact, one can adapt their method to numerically put bound on $\lambda_1$ a finite $g$.\footnote{See Appendix~\ref{app:finiteg} for a short discussion of how this is done.} As was recently shown by Hide and Magee \cite{HideMagee}, there exists a sequence of closed hyperbolic surfaces with $g\rightarrow\infty$ and $\lambda_1\rightarrow 1/4$. This means the bound of ~\cite{huber1980spectrum} is sharp in the limit $g\rightarrow\infty$.

We list the numerical values of all these bounds in Table~\ref{table:mutli}, and mark the best bound with a red asterisk\optimal. We see that our bound is an improvement over the previous bounds for all $2\leqslant g\leqslant 20$ apart from $g=4,6$. On the other hand, it is clear that for sufficiently large $g$, the bound of~\cite{huber1980spectrum} becomes the best.

\subsection{Values of $\lambda_1$ attained by hyperbolic orbifolds}\label{ssec:mixed}
We would now like to address the following question. What is the image of the map $X\mapsto \lambda_1(X)$ when $X$ ranges over all hyperbolic orbifolds? Our proposal is shown in Figure~\ref{fig:conjecture}. We can summarize it as follows

\begin{conjecture}\label{conjecture:theoneandtheonlys}
Let $E\subset \R_{>0}$ be the set of $\lambda_1(X)$ as $X$ runs over all orbifolds. Then $E$ is the union of the interval $(0,\lambda_1^{[2,3,9]}]$ with the finite set $\{\lambda_1^{[2,3,8]}\}\cup\{\lambda_1^{[2,4,5]}\}\cup\{\lambda_1^{[2,3,7]}\}$. Here $\lambda_1^{[k_1,k_2,k_3]}$ denotes $\lambda_1$ of the orbifold with signature $[0;k_1,k_2,k_3]$.

Furthermore,
\begin{enumerate}
	\item $\lambda^{[2,3,7]}_1\approx44.88835$ only occurs for the type $[0;2,3,7]$,
	\item $\lambda^{[2,4,5]}_1\approx28.07984$ only occurs for the type $[0;2,4,5]$,
	\item $\lambda^{[2,3,8]}_1\approx23.07855$ only occurs for the types $[0;2,3,8]$ and $[0;3,3,4]$,
	\item $\lambda_1^{[2,3,9]}\approx15.79023$ only occurs for the type $[0;2,3,9]$ and for the $\Z_3$-symmetric point on the moduli space of type $[0;2,2,2,3]$.
\end{enumerate}
\end{conjecture}

In the rest of this subsection, we will use the bootstrap method and the Yang-Yau bound to make significant progress towards proving Conjecture~\ref{conjecture:theoneandtheonlys}. In particular, we will prove

\begin{theorem}\label{thm:setE}
Let $X$ be an orbifold such that $\lambda_1(X)\geq 15.7$. Then $X$ must have one of the following signatures:
\begin{align*}
&[0;2,3,7],\, [0;2,3,8],\,[0;2,3,9],\,[0;2,3,10],\,[0;2,4,5],\, [0;2,4,6],\,[0;2,5,5],\,[0;3,3,4],\\
&[0;2,2,2,3]\,.
\end{align*}
\end{theorem}
In order to complete the proof of Conjecture~\ref{conjecture:theoneandtheonlys} assuming Theorem~\ref{thm:setE}, one would have to show that
\begin{enumerate}
\item $\lambda_1^{[2,3,9]}\geq 15.7$,
\item $\lambda_1^{[2,3,10]}<\lambda_1^{[2,3,9]}$,
\item $\lambda_1^{[2,4,6]}<\lambda_1^{[2,3,9]}$,
\item $\lambda_1^{[2,5,5]}<\lambda_1^{[2,3,9]}$,
\item on the moduli space of $[0;2,2,2,3]$, $\lambda_1$ has a unique global maximum at the $\mathbb{Z}_3$-symmetric point, and it satisfies $\lambda_1 = \lambda_1^{[2,3,9]}$ at this point.
\end{enumerate}
We expect that points (1)--(4) can be taken care of by computing $\lambda_1$ for the surfaces in question using one of the rigorous methods, for example starting from the Selberg trace formula as developed in \cite{BookerStrombergsson}. Point (5) is more subtle and will require new ideas. Note that once (5) is established, it follows that every value $\lambda_1\in(0,\lambda_1^{[2,3,9]}]$ is realized for some $X$ of type $[0;2,2,2,3]$. This is because $\lambda_1$ is a continuous function on the moduli space and $\lambda_1 \rightarrow 0$ in the degeneration limit where a long tube is formed between pairs of punctures \cite{MR573440}.

\begin{proof}[Proof of Theorem~\ref{thm:setE}]
The idea of the proof is to use the bootstrap method and the Yang-Yau bound to show that $\lambda_1(X)<15.7$ for all signatures except those listed in the Theorem. There are several cases to consider. Firstly, all orbifolds of positive genus satisfy $n_{min}=1$, and therefore $\lambda_1<8.47032$ as a result of the single-correlator bound shown in Table~\ref{table:singlecorr}.

It remains to analyze genus-0 orbifolds. For genus-0 orbifolds with at least 5 orbifold points, the Yang-Yau bound given by \eqref{eq:ESI} implies $\lambda_1\leq 8$. This is because 8 is the Yang-Yau bound for the type $[0;2,2,2,2,2]$ and the bound is monotonic decreasing as a function of the Euler characteristic.

Moving on to genus 0 with 4 punctures, the Yang-Yau bound similarly takes care of all signatures except for $[0;2,2,2,3]$ and $[0;2,2,2,4]$. Indeed, the bound gives $\lambda_1\leq 40/3$ for $[0;2,2,2,5]$. To take care of $[0;2,2,2,4]$, we use the fact that orbifolds of this signature admit two linearly independent holomorphic forms with $n=4$. We set up the mixed correlator bootstrap for the system
\be
	\<\sO_{4,a}(z_1)\sO_{4,b}(z_2)\widetilde\sO_{4,c}(z_3)\widetilde\sO_{4,d}(z_4)\>,
\ee
with $a,b,c,d\in{1,2}$. We verified that the resulting bound is at least as good as $\lambda_1< 15.7$.

Let us set aside $[0;2,2,2,3]$ for now and move on to genus 0 with 3 punctures. To treat this case, the main new idea is to use the mixed correlator bootstrap for a system of holomorphic forms of different weights. We define $n'_{min}$ to be the second smallest $n$ for which $\ell_n>0$. It turns out that the $[0;k_1,k_2,k_3]$ orbifolds that came close to saturating our bounds in Section~\ref{sec:universalorbifold} have a distinguished value of $n'_{min}$. Specifically,
\begin{itemize}
	\item if $n_{min}=3$ then either $n'_{min}\leq 5$ or $n'_{min}=6$ and $X$ is the $[0;3,3,4]$ orbifold,\footnote{$n'_{min}=5$ for $[0;3,3,p]$ with $p>4$ and $n'_{min}=4$ otherwise.}
	\item if $n_{min}=4$ then either $n'_{min}\leq 6$ or $n'_{min}=8$ and $X$ is the $[0;2,4,5]$ orbifold,\footnote{$n'_{min}=6$ for $[0;2,4,p]$ with $p>5$ and $n'_{min}=5$ otherwise.}
	\item if $n_{min}=6$ then either $n'_{min}=8$ or $n'_{min}=12$ and $X$ is the $[0;2,3,7]$ orbifold.
\end{itemize}

Given this specificity of the pair $(n_{min},n'_{min})$, it is useful to study the crossing equations for correlators involving $\sO_{n_{min}}(z)$ and $\sO_{n'_{min}}(z)$ (and their conjugates). This leads to the following bound
\be
	\l_1<\l_1^{\text{mixed}}(n_{min},n'_{min}),
\ee
where the values of $\l_1^{\text{mixed}}(n_{min},n'_{min})$ are listed in Table~\ref{table:mixed}. The algorithm that we used to obtain them is discussed in Section~\ref{sec:derivation}.\footnote{Note that in all cases that we consider here $n'_{min}-n_{min}\leq 2$, and $n_{min}\geq 3$. This is important to exclude contributions of discrete series to the product $\sO_{n'_{min}}\widetilde\sO_{n_{min}}$, as was discussed in Section~\ref{sec:derivation}.} We also list in Table~\ref{table:mixed} the orbifolds which, as far as we know, come the closest to saturating the bound. We can see that we get more detailed bounds than in Section~\ref{sec:universalorbifold}. For example, suppose we have an orbifold $X$ with $n_{min}=4$. From Table~\ref{table:singlecorr} we can only conclude that $\l_1<30.35432$. On the other hand, from the above discussion and Table~\ref{table:mixed} we infer that either $X$ is isometric to $[0;3,3,4]$ or $\l_1<15.9536$.

\begin{table}[t]
	\begin{center}
		\begin{tabular}{l|l|l||ll}

			$n_{min}$& $n'_{min}$ & $\l_1^{\text{mixed}}(n_{min},n'_{min})$ & largest known $\l_1$ & signature \\
			\hline
			3 & 4 & 12.2906 & 11.1982 & $[0;3,4,4]$\\
			3 & 5 & 12.1527 & 12.1362 & $[0;3,3,5]$\\
			\hline
			4 & 5 &15.9536 & 13.2389 & $[0;2,5,5]$\\
			4 & 6 &15.8107 & 15.0315 & $[0;2,4,6]$\\
			\hline
			6 & 8 & 23.0997 & 23.0785 &$[0;2,3,8]$\\
 
		\end{tabular}
	\end{center}
	\caption{Bootstrap bound $\l_1^{\text{mixed}}(n_{min},n'_{min})$ from a pair of coherent states with dimensions $n_{min}$ and $n'_{min}$. The values for the largest known $\l_1$ are non-rigorous and were computed using FreeFEM~\cite{freefem}.\label{table:mixed}}
\end{table}

Bounds in Table~\ref{table:mixed} allow us to make progress towards proving Conjecture~\ref{conjecture:theoneandtheonlys} by showing that $\lambda_1(X)<15.7$ for a large class of orbifolds of type $[0;k_1,k_2,k_3]$. Indeed, suppose that $k_1\geq 3$ and $k_2\geq 4$ (hence $k_3\geq 4$ since $k_1\leq k_2\leq k_3$). Then $n_{min} = 3$ and $n'_{min}=4$. Hence $\lambda_1\leq 12.2906$. Similarly, suppose $k_1=k_2=3$ and $k_3\geq 5$. Then $n_{min} = 3$ and $n'_{min}=5$. Hence $\lambda_1\leq 12.1527$. In other words, we have shown that $\lambda_1<15.7$ for all triangles $[k_1,k_2,k_3]$ such that $k_1\geq 3$ with the exception of $[3,3,4]$.

It remains to treat the triangles of the form $[2,k_2,k_3]$. To do that, we will employ the system of bootstrap equations involving modular forms of three distinct weights $n_{min}<n'_{min}<n''$ and their antiholomorphic counterparts. Consider first $[2,k_2,k_3]$ with $k_2\geq 5$ and $k_3\geq 6$. These triangles have $n_{min}=4$, $n'_{min}=5$. They also have a modular form with $n''=6$. We consider the set of all four-point correlators $\langle\sO_{n_1}\sO_{n_2}\widetilde{\sO}_{n_3}\widetilde{\sO}_{n_4}\rangle$, where $n_1,n_2,n_3,n_4\in\{4,5,6\}$. We checked that the resulting bootstrap bound implies $\lambda_1<15.7$. The same set-up with $n_{min}=4$, $n'_{min}=6$ and $n''=7$ proves $\lambda_1<15.7$ for all triangles $[2,4,k_3]$ with $k_3\geq 7$. Finally, the choice $n_{min}=6$, $n'_{min} = 8$, $n''=11$ proves $\lambda_1<15.7$ for all triangles $[2,3,k_3]$ with $k_3\geq 11$.
\end{proof}

\subsection{Bounds on structure constants}

As we mentioned in Section~\ref{sec:derivation}, it is also possible to put bounds on the structure constants appearing in~\eqref{eq:schannel} and~\eqref{eq:tchannel}. Here we briefly explore the bounds on the coefficients $S_{2n}=\sum_{a=1}^{\ell_{2n}}f_{2n,a}^2$ that enter the expansion~\eqref{eq:schannel}. Recall that in terms of the coherent states, the coefficients $f_{2n,a}$ can be defined by
\be
\sO_n(0)\sO_n(0)=\sum_{a=1}^{\ell_{2n}}f_{2n,a} \sO_{2n,a}(0).
\ee

Using the methods described in Section~\ref{sec:derivation}, we can derive a lower bound
\be
S_{2n}\geq S_{2n}^\text{single},
\ee
valid for any orbifold with $\ell_n>0$. For $n=6$ we find the bound\footnote{We did not verify this particular bound rigorously using rational arithmetic.}
\be\label{eq:fbound}
S_{12}\geq S_{12}^\text{single}=1.15409694432.
\ee

In Appendix~\ref{app:opetheory} we explain how value of $S_{2n}$ can be computed to a high precision in the case of $[0;k_1,k_2,k_3]$ orbifolds. We find that for the $[0;2,3,7]$ orbifold and $n=6$,
\be\label{eq:237ope}
S_{12}=1.1540969443944852107791801492\cdots.
\ee
This is remarkably close to the bound~\eqref{eq:fbound}. We should note that this bound was computed at $\L=37$, and we do not know whether it converges to the exact answer~\eqref{eq:237ope}. We noticed that convergence slows down dramatically (from exponential to an inverse power law) around $\L=21$.

Finally, let us note that it is also possible to obtain upper bounds on $S_{2n}$, provided one assumes a lower bound $\l_{1,\text{lower}}$ on the value of $\l_1$. As $\l_{1,\text{lower}}$ approaches $\l_{1}^\text{single}(n)$, the upper bound merges with the lower bound.

Note that bounds on integrals of triple products of automorphic forms have been much studied in the literature \cite{Sarnak1994,BernsteinReznikov2010,MichelVenkatesh2010,Nelson2021}. This is in part due to their connection to special values of L-functions in the arithmetic cases thanks to Watson's formula \cite{Watson2008}. In these works, the interest is in asymptotic bounds when one or more of the eigenvalues goes to infinity. Our method is closely related to one of Bernstein and Reznikov \cite{BernsteinReznikov2010} who studied constraints from the consistency of the spectral decomposition of integrals of quadruple product of functions in the principal series in $L^2(\Gamma\backslash G)$. In the present work, we have shown that when combined with linear programming, the analogous consistency constraints with discrete series can lead to nearly sharp bounds. It would be interesting to apply our methods to the situation when the eigenvalue tends to infinity.

\section{Final remarks}
\label{sec:finalremarks}
We conclude this paper with an additional remark and possible future directions.

\subsection*{Extremal functionals and spectrum extraction}
In the cases where the upper bound on $\lambda_1$ is nearly saturated by actual hyperbolic orbifolds, we can use our bootstrap method to give estimates not only of $\lambda_1$ but also of the remaining eigenvalues $\lambda_2,\,\lambda_3,\ldots$. To explain why this is the case, let us focus for simplicity on the bounds $\lambda_1^{\text{single}}$ coming from the crossing equation of a single correlator. Furthermore, let us imagine that in the limit of infinite truncation order $\Lambda\rightarrow\infty$, the upper bound $\lambda_1^{\text{single}}$ indeed converges to $\lambda_1(X)$ of an actual hyperbolic orbifold $X$. Consider the functions $P_\a^n(\l)$ discussed in Section~\ref{sec:derivation}, obtained by optimizing the bound at fixed $\Lambda$. Based on the examples that we analyzed, we conjecture that we can choose constants $c_{\Lambda}>0$ such that $c_{\Lambda}\rightarrow 0$ as $\Lambda\rightarrow\infty$ and such that $c_{\Lambda}P_\a^n(\l)$ has a limit $P^{n}_{\infty}(\lambda)$ as $\Lambda\rightarrow \infty$, where $P^{n}_{\infty}(\lambda)$ is a smooth function, not identically zero. Under our assumptions, we must have $P^{n}_{\infty}(0) = 0$ and $P^{n}_{\infty}(\lambda)\geq 0$ for all $\lambda\geq\lambda_1(X)$. The only way that these properties are consistent with the existence of the hyperbolic orbifold $X$ is if in fact $P^{n}_{\infty}(\lambda)$ vanishes on all the eigenvalues $\lambda_k(X)$, $k\geq 1$ which contribute to the correlator in question. Furthermore, since $P^{n}_{\infty}(\lambda)\geq 0$ for $\lambda\geq\lambda_1(X)$, the zeros at $\lambda_k(X)$ with $k\geq 2$ must all be of even order. Assuming there are no accidental zeros, it follows that the entire spectrum can be extracted from zeros of $P^{n}_{\infty}(\lambda)$.

Let us compare this expectation with our data obtained from the bootstrap algorithm of Section~\ref{sec:derivation} at finite $\Lambda$. Figure~\ref{fig:functional} shows the plot of $P_\a^n(\l)$ with $n=6$, obtained at $\Lambda = 41$. We see that besides the simple zero at $\approx44.8835$, $P_\a^n(\l)$ also has a sequence of local minima (almost double zeros) at $\lambda \approx 142.5552,\,201.4709,\,323.40$ and $456.3$. These values are remarkably close to the Laplace eigenvalues of the $[0;2,3,7]$ orbifold, which we obtained by solving the Laplace equation numerically using FreeFEM
$$
44.88835,\,142.5551,\,201.4705,\,323.3603,\,373.2063,\,454.6009\ldots\,.
$$

We note that the $P_\a^n(\l)$ does not seem to capture the eigenvalue $373.2063$, which is presumably due to it not appearing in the correlator $\<\sO_6\sO_6\widetilde{\sO}_6\widetilde{\sO}_6\>$ for symmetry reasons.

We remark that the extremal functional $P^{n}_{\infty}(\lambda)$ plays the same role in our analysis as the magic functions known from the work on the sphere packing problem~\cite{MR1973059,sphere8,sphere24}, and the analytic extremal functionals for the 1D conformal bootstrap~\cite{Mazac:2016qev,Mazac:2018mdx}. It has been shown that there is a close connection between the latter two cases~\cite{Hartman:2019pcd}. It would be remarkable if one could give an analytic proof of exact saturation of the bounds on $\lambda_1$, similarly to what has been achieved in \cite{sphere8,sphere24,Mazac:2016qev,Mazac:2018mdx}.

\begin{figure}[h]
\centering
\includegraphics[width=.75\textwidth]{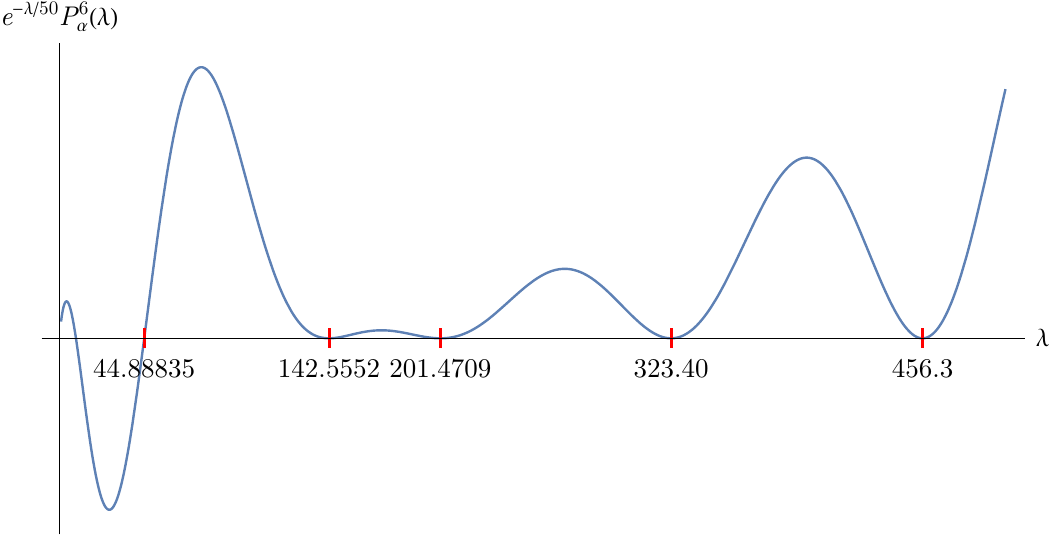}
\caption{Plot of the extremal functional $P_\a^n(\l)$ for $n=6$ and $\Lambda=41$, multiplied by $e^{-\lambda/50}$ for clarity. Besides the simple zero at $\lambda\approx 44.88835$, the functional has a sequence of local minima (almost double zeros). The location of the local minima are in a very good agreement with higher eigenvalues of the $[0;2,3,7]$ orbifold.}
\label{fig:functional}
\end{figure}

\subsection*{Generalizations}
In this paper, we obtained bounds on the Laplacian spectra of two-dimensional hyperbolic orbifolds by studying the crossing equation coming from consistency of several simple correlators of functions transforming in discrete series representations of $G=\PSL(2,\mathbb{R})$.

There are many possible generalizations of this simple set-up. Most directly, we can consider larger systems of mixed correlators of discrete series states, along the lines of Section~\ref{ssec:mixed}. We hope that this approach can lead to sharper bounds at large genus. Next, it is natural to study correlators of functions transforming in the principal series representations, as well the corresponding crossing equations, and thereby formalize the analysis of \cite{Bonifacio:2021msa}.

It is also very natural to replace $G=\PSL(2,\mathbb{R})$ with other semisimple Lie groups. For example, the analysis of this paper can be repeated with $G=\SO_0(1,d)$, which will lead to bounds on $d$-dimensional hyperbolic manifolds and orbifolds. We can also consider the less dramatic change $\PSL(2,\mathbb{R})\rightarrow \SL(2,\mathbb{R})$. In this way, we would introduce additional unitary irreducible representations. They correspond to spinor bundles on $X$. The set-up with $G=\SL(2,\mathbb{R})$ is thus expected to lead to bounds on the spectrum of the Dirac operator on hyperbolic orbifolds.

It would also be interesting to study in more detail the relation of our setup the conformal field theory. It appears that the role of Euclidean unitarity of the path integral has not so far been explored (at least in the context of the modern numerical bootstrap). One may speculate that it can provide new bounds on conformal field theories. Furthermore, our equations (and their generalization to external $\mathbb{O}_{k,a}$) provide a systematic way of numerically studying the recently proposed de Sitter bootstrap~\cite{Hogervorst:2021uvp,DiPietro:2021sjt}. We hope to report on these questions in near future.

\newpage

\appendix

\section{Smooth vectors of complementary series representations}
\label{app:complementarysmooth}

In this appendix we give an elementary proof of the following proposition.
\begin{proposition}
	$\cC_{s}^\oo=C^\oo(\ptl\D)$ (with coincidence of topologies).
\end{proposition}
\begin{proof}
The space $\cC_s$ in which the complementary series representation acts can be viewed as the completion of $L^2(\ptl \D)$ with respect to the norm~\eqref{eq:Cnorm}. A direct computation shows that if $f\in L^2(\ptl \D)$ is given by
\be
	f(\theta)=(2\pi)^{-\half}\sum_{m=-\oo}^{+\oo} f_m e^{im\theta},
\ee
then
\be
	&\|f\|_{L^2(\ptl \D)}^2=\sum_{m=-\oo}^{+\oo}|f_m|^2,\\
	&\|f\|_{\cC_s}^2=C\sum_{m=-\oo}^{+\oo}\frac{(\thalf-s)_{|m|}}{(\thalf+s)_{|m|}}|f_m|^2,\label{eq:CnormF}
\ee
where $C>0$ and $(a)_n=a(a+1)\cdots (a+n-1)$. We can thus identify elements $f$ of $\cC_s$ with sequences $f_m$ for which~\eqref{eq:CnormF} is finite. For a given $f$, the element $f_m$ of the sequence can be computed by taking inner product in $\cC_s$ with $e^{im\theta}$.

Let $L_0\in\mathfrak{su}(1,1)$ be the element of Lie algebra that is represented by constant shifts in $\theta$. Note for any $m$, $e^{im\theta}$ is $K$-finite and thus smooth. If $f\in \cC^\oo_s$ we can then compute the action of $L_0$ on $f$ in terms of $f_m$ by writing $f_m$ as an inner product with $e^{im\theta}$ and moving the action of $L_0$ onto $e^{im\theta}$. We can normalize $L_0$ so that
\be
	(L_0f)_m=m f_m.
\ee
By assumption, $L_0f$ should still be in $\cC_s$, and so the norm $\|L_0 f\|_{\cC_s}^2$ should converge. Iterating, we find that if $f\in \cC^\oo_s$, then $f_m$ should decay faster than any power of $m$. This implies that $f\in C^\oo(\ptl\D)$. 

Conversely, let $f\in C^\oo(\ptl\D)\subset\cC_s$ and $h\in \cC_s$. Define $\tl h_m=C\frac{(\half-s)_{|m|}}{(\half+s)_{|m|}}h_m$. Then $\tl h\in L^2(\ptl\D)$ and for $g\in G$
\be
	(h,g\.f)_{\cC_s}=(\tl h,g\.f)_{L^2(\ptl\D)}.
\ee
We can easily verify by the usual estimates in $\theta$ variable that $(\tl h,g\.f)_{L^2(\ptl\D)}$ is a smooth function of $g$, and so $f$ is weakly smooth. Then~\cite{Neto} it is smooth, and so $f\in \cC_s^\oo$. This proves $\cC_s^\oo=C^\oo(\ptl\D)$ as sets.

If a sequence $f_n\in C^{\oo}(\ptl\D)$ converges to $0$, $f_n\to 0$ in $C^\oo(\ptl\D)$, it is easy to check that $D\.g\.f_n\to 0$ in $L^2(\ptl\D)$ for any $D\in \cU(\mathfrak{su}(1,1))$ and uniformly in $g$ over compact subsets of $G$. Since $\|\.\|_{\cC_s}^2\leq C\|\.\|_{L^2(\ptl \D)}^2$, it follows that $D\.g\.f_n\to 0$ in $\cC_s$ (i.e.\ the seminorms of $f_n$ in $\cC^\oo_s$ all tend to zero), and so the identity map $C^\oo(\ptl\D)\to \cC^\oo_s$ is continuous. Since it is surjective, and $C^\oo(\ptl\D)$ and $\cC^\oo_s$ are both Frech\'et, it is also open by the open mapping theorem. This implies that the topologies on $C^\oo(\ptl\D)$ and $\cC^\oo_s$ coincide.
\end{proof}

\section{Proving positivity of polynomial matrices}
\label{app:polynomial}

Suppose that we are given a polynomial $P(x)\not\equiv 0$,
\be
	P(x)=p_0+p_1x+p_2x^2+\cdots +p_n x^n,
\ee
where $p_i\in \R$. In practice our $p_i$ will be rational, so we can perform the arithmetic computations described below
on a computer without any rounding errors. We would like to prove that $P(x)\geq 0$ for all $x\geq 0$. 

Let $Q(y)\equiv(1-y)^n P(y/(1-y))$. Note that $P(x)\geq 0$ for all $x\geq 0$ iff $Q(y)\geq 0$ for all $y\in [0,1]$. On the other hand $Q(y)\geq 0$ for all $y\in [a,b]$ iff the polynomial
\be
	P_{[a,b]}(x)\equiv (x+1)^n Q\p{\frac{bx+a}{x+1}}
\ee
satisfies $P_{[a,b]}(x)\geq 0$ for all $x\geq 0$. With this preparation, we can now describe the algorithm.

Let \texttt{proof}$(a,b)$ be the following procedure. 
\begin{enumerate}
	\item Check whether there is a trivial proof that $Q(y)\geq 0$ for all $y\in [a,b]$. For this, we verify whether $P_{[a,b]}(x)$ has only non-negative coefficients. If it does, then trivially $P_{[a,b]}(x)\geq 0$ for all $x\geq 0$ and thus $Q(y)\geq 0$ for all $y\in [a,b]$ by the above discussion. In this case \texttt{proof}$(a,b)$ terminates with a success.
	\item We check whether $P_{[a,b]}(x)$ has only non-positive coefficients. If it does, then trivially $Q(y)\leq 0$ for all $y\in [a,b]$, and \texttt{proof}$(a,b)$ terminates with failure.
	\item We run \texttt{proof}$(a,(a+b)/2)$ and \texttt{proof}$((a+b)/2,b)$. Terminate with success if both terminate with success, and terminate with failure otherwise.
\end{enumerate}

If \texttt{proof}$(0,1)$ terminates with a success, it generates a division of $[0,1]$ into closed subintervals on each of which $Q(y)$ is non-negative for a trivial reason.\footnote{In fact, positive in the interior.} If it terminates with a failure, it generates a closed subinterval of $[0,1]$ on which $Q(y)$ is non-positive.\footnote{In fact, negative in the interior.}

\texttt{proof}$(0,1)$ may not terminate if $P(x)$ has positive roots of even multiplicity. It can be shown that \texttt{proof}$(0,1)$ always terminates if $P(x)$ is strictly positive for $x>0$, but this is not necessary for our purposes: we simply state that it terminated on all inputs that we ran it for.

One can verify positivity of matrix polynomials $M(x)$ by applying \texttt{proof}$(0,1)$ to $P(x)=\det M(x)$ and verifying a finite number of additional inequalities (i.e.\ $M(1)\succ 0$ and that $P(x)$ doesn't vanish at the endpoints of the intervals generated by  \texttt{proof}$(0,1)$).

\section{Bounds of~\cite{huber1980spectrum} for finite $g$}
\label{app:finiteg}
In this appendix we describe briefly how the results of~\cite{huber1980spectrum} can be used to obtain bounds on $\l_1$ at finite genus. 

First, we define for $\mu>1/4$ the function $s(\mu)$ by
\begin{equation}
s(\mu)\equiv \frac{\left(\int_{1}^{t_*}\mathrm{d}t\ P(\frac{1}{2} \sqrt{1-4 \mu }-\frac{1}{2},t)\right)^2}{\int_{1}^{t_*}\mathrm{d}t\ P(\frac{1}{2} \sqrt{1-4 \mu }-\frac{1}{2},t)^2},
\end{equation}
where $P(x,y)\equiv P_x(y)$ are the Legendre $P$-functions, $t_*=\mathrm{min}\{ t \geq 1\vert P(\frac{1}{2} \sqrt{1-4 \mu }-\frac{1}{2},t)=0\}$. Note that for $\mu>1/4$ the $P(\frac{1}{2} \sqrt{1-4 \mu }-\frac{1}{2},t)$ has infinitely many zeros in $(1,\infty)$.

In~\cite{huber1980spectrum} the following inequality is shown to hold for any $\mu>\frac{1}{4}$,
\begin{equation}
\left(1-\frac{s(\mu)}{2(g-1)}\right)\lambda_1 \leq \mu \,.
\end{equation} 
We then simply scan over $\mu\in M_g$ where $M_g=\left\{\mu>\tfrac{1}{4}\Big\vert 1-\frac{s(\mu)}{2(g-1)}>0 \right\}$ to obtain the bound
\begin{equation}
\lambda_1 \leq \inf_{\mu\in M_g}\frac{\mu}{1-\frac{s(\mu)}{2(g-1)}}\,.
\end{equation} 
The resulting bounds are listed in Table~\ref{table:mutli}.

\section{Structure constants on $[0;k_1,k_2,k_3]$}
\label{app:opetheory}

In this appendix we briefly review how the coefficients of the form $f_{m+n,a}$ appearing in 
\be
	\sO_m(0)\sO_n(0)=\sum_{a=1}^{\ell_{m+n}} f_{m+n,a}\sO_{m+n}(0)
\ee
can be computed exactly in the case of $[0;k_1,k_2,k_3]$ orbifolds. Here we assume for simplicity that $\ell_m=\ell_n=\ell_{m+n}=1$, so there is only one coefficient $f_{m+n}$ to compute.

The basic idea is that the hyperbolic metric for $[0;k_1,k_2,k_3]$ can be computed exactly in the following sense. Consider the Riemann sphere as a complex surface. By this we mean, technically, two copies of $\C$ parametrized by $z$ and $z'$, glued via $z=1/z'$. We can then find a closed-form expression for the conformal factor $e^{2\f(z,\bar z)}$ such that the metric
\be
	ds^2=e^{2\f(z,\bar z)}dzd\bar z
\ee
is a smooth hyperbolic metric on the Riemann sphere apart from three conical singularities of orders $k_1,k_2,k_3$ at $z=0,1,\oo$.

Note that this is non-trivially different from constructing the metric in $\G\@\mathbb{H}$ picture. There, the hyperbolic metric is self-evident, but the complex structure is not: we need to work to construct the holomoprhic modular forms for $\G$. Here, there is no $\G$ and the complex structure is explicit, but the metric takes a non-trivial form, which is described below.

Suppose we found the hyperbolic metric. Then we can construct the holomorphic sections
\be
	g_n(z)dz^n
\ee
that correspond to $\sO_n(0)$ by first picking some holomorphic sections and then normalizing them according to our metric. The coefficient $f_{m+n}$ is then computed by
\be
	f_{m+n}=\frac{g_m(z)g_n(z)}{g_{m+n}(z)}.
\ee

Determining the basis holomorphic sections $g_n(z)dz^n$ is slightly non-trivial because the holomorphic sections of $\cK^{\otimes n}$ on an orbifold do not correspond to holomorphic functions $g_n(z)$. This is also related to the $k_i$-dependent correction terms in~\eqref{eq:RRoch}. First, let us explain what a holomorphic section of $\cK^{\otimes n}$ is near an orbifold point. 

Let $w$ be the coordinate in an orbifold chart which covers an orbifold point of order $k$, situated at $w=0$. Transition functions from this chart to other charts should be invariant under $w\to e^{2\pi i/k}w$. A holomorphic section of $\cK^{\otimes n}$ is locally a section
\be
	h(w)dw^n
\ee
where $h(w)$ is holomorphic in the chart, and the section is invariant under $w\to e^{2\pi i/k}w$. That is,
\be\label{eq:hinv}
	h(w)=e^{2\pi i n/k}h(e^{2\pi i/k}w).
\ee
We can always define a new chart with a coordinate $u$ related to $w$ by $u=w^k$. If we do this for every orbifold point of $X$, we give $X$ a new complex structure which makes it into a complex surface $X_{smooth}$. In the $g=0$ case $X_{smooth}$ is just the Riemann sphere. Holomorphic sections of $\cK^{\otimes n}(X)$ become sections of $\cK^{\otimes n}(X_{smooth})$ which are holomorphic away from the orbifold points of $X$, but can have poles at these points. Indeed, in $u$ coordinate we have the section
\be
	t(u)du^n=h(w)dw^n,
\ee
and so
\be
	t(u)=k^{-n} u^{-n(k-1)}h(w).
\ee
By analyzing the Taylor series expansion of $h(w)$ near $w=0$ under the condition~\eqref{eq:hinv}, we find that $t(u)$ is single-valued but can have a pole at $u=0$ of order $\lfloor n(k-1)/k\rfloor$.

So we conclude that holomorphic sections of $\cK^{\otimes n}(X)$ are meromorphic sections of $\cK^{\otimes n}(X_{smooth})$ which have poles of order at most $\lfloor n(k_i-1)/k_i\rfloor$ at an orbifold point of $X$ of order $k_i$. Using this fact and the Riemann-Roch theorem for smooth surfaces, one can recover the orbifold version~\eqref{eq:RRoch}.

The coordinates $z$ and $z'$ on our Riemann sphere describe the complex structure of $X_{smooth}$. So in order to construct a holomorphic section $g_n(z)dz^n$ of $\cK^{\otimes n}(X)$, we need to find a function $g_n(z)$ that is holomorphic in $\C\setminus\{0,1\}$, and 
\begin{itemize}
	\item has a pole of order at most $\lfloor n(k_1-1)/k_1\rfloor$ at $z=0$,
	\item has a pole of order at most $\lfloor n(k_2-1)/k_2\rfloor$ at $z=1$,
	\item grows at most as $z^{-\lceil n(k_3+1)/k_3\rceil}$ as $z\to \oo$.
\end{itemize}
In the rest of this appendix we focus mostly on the $[0;2,3,7]$ orbifold. For example, the unique section with $n=6$ for $[0;2,3,7]$ orbifold is
\be
	g_6(z)dz^6=N_6 z^{-3}(z-1)^{-4}dz^6,
\ee
where we choose $N_6$ so that $\|g_6\|=1$, computed below. Similarly, the unique choice for $g_{12}(z)$ is
\be
	g_{12}(z)dz^{12}=N_{12} z^{-6}(z-1)^{-8}dz^{12}.
\ee
The structure constant is then given by
\be
	f_{12}=\frac{N_6^2}{N_{12}}.
\ee
It only remains to discuss the computation of $N_6$ and $N_{12}$.

\subsection{Explicit form of the hyperbolic metric}

The problem of finding the hyperbolic metric on a sphere with $r$ singularities is well-known both in physics (in the context of semiclassical Liouville theory~\cite{Zamolodchikov:1995aa}) and in mathematics~\cite{Hempel}. In general it requires solving a Fuchsian differential equation with $r$ regular singularities and determining a set of $r-3$ accessory parameters from some monodromy constraints. In the case of $r=3$, the differential equation becomes hypergeometric, and there are no accessory parameters to fix. This allows one to obtain a relatively simple closed-form solution, which we give here without a derivation (see, e.g.~\cite{Harlow:2011ny}).

First we define the constants
\be
	\eta_i=\frac{1}{2}\frac{k_i-1}{k_i},
\ee
and the functions
\be
\tl w_1(z)&=z^{\eta_1}(1-z)^{\eta_2} {}_2F_1(\eta_1+\eta_2-\eta_3,\eta_1+\eta_2+\eta_3-1,2\eta_1,z),\\
\tl w_2(z)&=z^{1-\eta_1}(1-z)^{\eta_2} {}_2F_1(1-\eta_1+\eta_2-\eta_3,-\eta_1+\eta_2+\eta_3,2(1-\eta_1),z).
\ee
Then the hyperbolic metric is given by
\be\label{eq:trimetric}
e^{2\f}=\frac{4r(1-2\eta_1)^2}
{
	\p{r\tl w_1(z)\tl w_1(\bar z)-\tl w_2(z)\tl w_2(\bar z)}^2
},
\ee
where
\be
	r=\frac{\G(2(1-\eta_1))^2\G(1+\eta_1-\eta_2-\eta_3)\G(\eta_1+\eta_2-\eta_3)\G(\eta_1-\eta_2+\eta_3)\G(\eta_1+\eta_2+\eta_3-1)}
{\G(2\eta_1)^2\G(2-\eta_1-\eta_2-\eta_3)\G(1-\eta_1+\eta_2-\eta_3)\G(1-\eta_1-\eta_2+\eta_3)\G(-\eta_1+\eta_2+\eta_3)}.
\ee

\subsection{Evaluation of the normalization integrals}

In the rest of this appendix we outline the main steps of the computation of the normalization factors $N_a$. The $L^2$ norm of a section $f_a(z)dz^p$ is given by
\be
	\|f_a(z)dz^a\|^2=\frac{1}{\vol X}\int_X dzd\bar z e^{(2-2a)\f(z,\bar z)}|f_a(z)|^2=\frac{i}{2\vol X}\int_X dz\wedge d\bar z e^{(2-2a)\f(z,\bar z)}f_a(z)\overline{f_a(z)},
\ee
From the expression~\eqref{eq:trimetric} we see that for $a\geq 1$ the integrals that we need to compute have the general form
\be
\cI=\int_X dz\wedge d\bar z \sum_{i} F_i(z) G_i(\bar z).
\ee
The functions $F_i$ and $G_i$ are not completely independent. Each of these has cuts, but the sum $\sum_{i} F_i(z) G_i(\bar z)$ is single-valued. These functions have singularities at $z=0,1$ and $\oo$. We put the cuts on $(-\oo,0]$ and $[1,+\oo)$. We pick some $p\in \C$ and define
\be
H_i(z)=\int_{p}^z dz F_i(z),
\ee
with contour avoiding the cuts. The dependence on $p$ will go away eventually. In the cut plane we have $dH_i=F_idz$, so
\be
dz\wedge d\bar z \sum_{i} F_i(z) G_i(\bar z)=d\p{H_i(z)G_i(\bar z)d\bar z},
\ee
and thus the integral becomes
\be
\cI=\sum_i \int_C d\bar z H_i(\bar{z})G_i(\bar z),
\ee
where $C$ wraps the two cuts. We can separate it into two contributions $\cI=\cI_{left}+\cI_{right}$. Let us consider the left cut $(-\oo,0]$, and split the contour in definition of $H_i(z)$ into a segment going from $p$ to $0$ and from $0$ to $z$ along a straight line. We then have
\be
\cI_{left}=&\sum_i \int_{-\oo}^0 dt G_i(t-i0)\p{\int_p^0 dt' F_i(t')+\int_0^{t} dt' F_i(t'+i0)}\nn\\
&-\sum_i \int_{-\oo}^0 dt G_i(t+i0)\p{\int_p^0 dt' F_i(t')+\int_0^{t} dt' F_i(t'-i0)}.
\ee
The single-valuedness property of $\sum_i F_i(z) G_i(\bar z)$ mentioned above means that $\sum_i G_i(t-i0)F_i(t'+i0)=\sum_i G_i(t+i0)F_i(t'-i0)$ for $t,t'<0$. So we get for the left cut
\be
\cI_{left}=\sum_i \int_{-\oo}^0 dt (G_i(t-i0)-G_i(t+i0))\int_p^0 dt' F_i(t').
\ee
Similarly, on the right cut we get
\be
\cI_{right}=\sum_i \int_{1}^{+\oo} dt (G_i(t-i0)-G_i(t+i0))\int_p^1 dt' F_i(t').
\ee
For $G_i$ well-behaved at infinity we have
\be
\int_{1}^{+\oo} dt (G_i(t-i0)-G_i(t+i0))=-\int_{-\oo}^0 dt (G_i(t-i0)-G_i(t+i0))
\ee
And so in total the integral splits into a sum of products of one-dimensional integrals
\be
\cI=\sum_i \int_{1}^{+\oo} dt (G_i(t-i0)-G_i(t+i0)) \int_0^1 dt' F_i(t').
\ee
These integrals are products of powers of $z$, $1-z$ and of hypergeometric functions of $z$. We don't know how to evaluate them analytically, but they can be computed numerically to arbitrarily high precision. In particular, we find the OPE coefficient,
\be
	S_{12}=(f_{12})^2=1.1540969443944852107791801492\cdots.
\ee

\printbibliography


\end{document}